\documentclass[12pt,reqno]{article}
\usepackage[usenames,dvipsnames]{xcolor}

\usepackage{amsmath,mathtools,tikz,pgfplots,soul}

\usepackage{amssymb,amsthm,mathtools,bbm}
\usepackage{enumitem,cite}
\usepackage[longnamesfirst]{natbib}
\usepackage{multirow,abstract}
\usepackage{framed,mdframed,cancel,setspace}
\usepackage[bottom,hang]{footmisc}
\usepackage{graphicx,tikz}
\usepackage[colorlinks=true,allcolors=NavyBlue,citecolor=NavyBlue]{hyperref}
\hypersetup{breaklinks=true}
\usetikzlibrary{calc,positioning}
\usepackage[font={small,it}]{caption}
\usepackage[noabbrev]{cleveref}
\usepackage[T1]{fontenc}

\usepackage{breakurl}

\newcommand{\mysetminusD}{\hbox{\tikz{\draw[line width=0.6pt,line cap=round] (3pt,0) -- (0,6pt);}}}
\newcommand{\mysetminusT}{\mysetminusD}
\newcommand{\mysetminusS}{\hbox{\tikz{\draw[line width=0.45pt,line cap=round] (2pt,0) -- (0,4pt);}}}
\newcommand{\mysetminusSS}{\hbox{\tikz{\draw[line width=0.4pt,line cap=round] (1.5pt,0) -- (0,3pt);}}}

\newcommand{\varsetminus}{\mathbin{\mathchoice{\mysetminusD}{\mysetminusT}{\mysetminusS}{\mysetminusSS}}}

\newtheoremstyle{ans}{15pt}{20pt}{}{}{\bfseries}{}{ }{\thmname{#1}\thmnote{ #3}.}
\theoremstyle{ans}
\newtheorem*{soln}{Solution}

\newcounter{lemno}
\newcounter{app_lemno}
\newcounter{propno}
\newcounter{app_propno}
\newcounter{thmno}
\newcounter{rethmno}
\newcounter{clmno}
\newcounter{app_thmno}
\newcounter{defno}
\newcounter{factno}
\newcounter{corno}
\newcounter{assno}
\newcounter{reassno}
\newcounter{examplenum}
\newcounter{remarkno}
\newmdenv[backgroundcolor=blue!7,linewidth=0pt]{bBox}
\newmdenv[backgroundcolor=green!7,linewidth=0pt]{gBox}
\newmdenv[backgroundcolor=red!7,linewidth=0pt]{rBox}
\newmdenv[backgroundcolor=gray!7,linewidth=0pt]{grayBox}

\newtheoremstyle{statementStyle}
{12pt}
{12pt}
{}
{}
{}
{{\bf .}\;}
{0.25em}
{{\bf\thmname{#1}\thmnumber{ #2}}
\thmnote{ (#3)}} 
\makeatother

\theoremstyle{statementStyle}

\newtheoremstyle{statementStyle}
{12pt}
{12pt}
{\it}
{}
{}
{{\bf .}\;}
{0.25em}
{{\bf\thmname{#1}\thmnumber{ #2}}
\thmnote{ (#3)}} 
\makeatother

\theoremstyle{plain}
\newtheorem{lemma}[lemno]{Lemma}

\newtheorem{theorem}[thmno]{Theorem}

\newtheorem{coro}[corno]{Corollary}
\newtheorem{exerciseT}[thmno]{Exercise}
\newtheorem{exampleT}[examplenum]{Example}
\newtheorem{defn}[defno]{Definition}
\newtheorem{assumption}[assno]{Assumption}
\Crefname{assumption}{Assumption}{Assumptions}
\newtheorem{proposition}[propno]{Proposition}
\Crefname{proposition}{Proposition}{Propositions}
\newtheorem{remark}[remarkno]{Remark}
\Crefname{remark}{Remark}{Remarks}

\newtheorem{claim}[clmno]{Claim}
\Crefname{claim}{Claim}{Claims}

\newtheoremstyle{nbstatementStyle}
{12pt}
{12pt}
{\it}
{}
{}
{{\bf .}\;}
{0.25em}
{{\bf\thmname{#1}}
\thmnote{ (#3)}} 
\makeatother

\theoremstyle{nbstatementStyle}

\newtheoremstyle{restatementStyle}
{12pt}
{12pt}
{\it}
{}
{\bfseries}
{{\bf .}\;}
{0.25em}
{\thmname{#1}\thmnumber{ #2'}
\thmnote{ (#3)}} 
\makeatother

\theoremstyle{restatementStyle}

\Crefname{retheorem}{Theorem}{Theorems}

\newtheorem{prob}{Problem}

\newtheorem{conj}{Conjecture}
\newtheorem{prope}{Property}

\newtheoremstyle{reexerciseStyle}
{12pt}
{12pt}
{}
{}
{\bfseries}
{{\bf .}\;}
{0.25em}
{\thmname{#1}\thmnumber{ #2}
\thmnote{ (#3)}} 
\makeatother

\newtheoremstyle{exerciseStyle}
{12pt}
{12pt}
{}
{}
{\bfseries}
{{\bf .}\;}
{0.25em}
{\thmname{#1}\thmnumber{ #2}
\thmnote{ (#3)}} 
\makeatother

\theoremstyle{exerciseStyle}

\newenvironment{corollary}{
\begin{coro}}
{\end{coro}
}

\newtheoremstyle{teach}{15pt}{\topsep}{}{}{\itshape}{}{ }{\thmname{#1}\thmnote{ #3}.}
\theoremstyle{teach}
\newtheorem*{prooflemma}{Proof of lemma}

\newtheoremstyle{rem}{5pt}{0pt}{\color{black}}{}{\bfseries}{}{ }{\thmname{#1}\thmnote{ #3}.}
\theoremstyle{rem}

\definecolor{darkgreen}{rgb}{0.0, 0.75, 0.2}

\makeatletter
\renewcommand*\env@matrix[1][*\c@MaxMatrixCols c]{%
  \hskip -\arraycolsep
  \let\@ifnextchar\new@ifnextchar
  \array{#1}}
\makeatother

\renewcommand{\a}{\alpha}

\renewcommand{\d}{\delta}

\newcommand{\dd}{\mathrm{d}}

\newcommand{\la}{\lambda}

\newcommand{\La}{\Lambda}

\newcommand{\R}{\mathbb{R}}

\newcommand{\s}{\sigma}

\newcommand{\x}{\mathbf{x}}

\newcommand{\bone}{\boldsymbol{1}}

\newcommand{\BAR}{\overline}

\newcommand{\und}[1]{\underline{#1}}

\DeclarePairedDelimiter{\norm}{\lVert}{\rVert} 
\DeclarePairedDelimiter{\abs}{\lvert}{\rvert}

\newcommand{\AND}{\quad\text{and}\quad}

\renewcommand{\(}{\left(}
\renewcommand{\)}{\right)}

\newcommand{\eg}{e.g.}
\newcommand{\ie}{i.e.}

\DeclareMathOperator*{\argmax}{arg\,max}

\DeclareMathOperator{\E}{\mathbf{E}}
\DeclareMathOperator{\im}{im}

\def\calC{\mathcal{C}}

\def\calI{\mathcal{I}}

\def\calL{\mathcal{L}}
\def\calM{\mathcal{M}}

\def\calO{\mathcal{O}}

\usepackage{datetime}

\newdateformat{monthyeardate}{%
  \monthname[\THEMONTH] \THEYEAR}
\newdateformat{daymonthyeardate}{%
  \monthname[\THEMONTH] \THEDAY, \THEYEAR}

\title{\LARGE\hspace*{\fill}\\[1ex]\papertitle\\[2ex]}
\author{\large\name\\[1ex] \affiliation}
\date{\monthyeardate\today}

\emergencystretch=\maxdimen
\makeatletter
\def\@xfootnote[#1]{%
  \protected@xdef\@thefnmark{#1}%
  \@footnotemark\@footnotetext}
\makeatother

\makeatletter
\ifFN@para
\else
  \long\def\@makefntext#1{%
    \ifFN@hangfoot
      \bgroup
      \setbox\@tempboxa\hbox{%
        \ifdim\footnotemargin>0pt
          \hb@xt@\footnotemargin{\@makefnmark\hss}%
        \else
          \@makefnmark\hskip-\footnotemargin      
        \fi
      }%
      \leftmargin\wd\@tempboxa
      \rightmargin\z@
      \linewidth \columnwidth
      \advance \linewidth -\leftmargin
      \parshape \@ne \leftmargin \linewidth
      \footnotesize
      \@setpar{{\@@par}}%
      \leavevmode
      \llap{\box\@tempboxa}%
      \parskip\hangfootparskip\relax
      \parindent\hangfootparindent\relax
    \else
      \parindent1em
      \noindent
      \ifdim\footnotemargin>\z@
        \hb@xt@ \footnotemargin{\hss\@makefnmark}%
      \else
        \ifdim\footnotemargin=\z@
          \llap{\@makefnmark}%
        \else
          \llap{\hb@xt@ -\footnotemargin{\@makefnmark\hss}}%
        \fi
      \fi
    \fi
    \footnotelayout#1%
    \ifFN@hangfoot
      \par\egroup
    \fi
  }
\fi
\makeatother

\makeatletter
\def\@endtheorem{\endtrivlist}
\makeatother

\newcommand{\citepos}[1]{\citeauthor{#1}'s (\citeyear{#1})}

\usepackage[T1]{fontenc}

\usepackage{newtxtext}

\usepackage{lmodern,newtxtext,newtxmath}
\usepackage[cal=cm]{mathalfa}

\usepackage{graphicx} 
\usepackage{amsmath}  
\usepackage{scalerel} 

\usepackage{booktabs,array,tablefootnote,subcaption} 
\usepackage[titletoc]{appendix}
\newcolumntype{C}[1]{>{\centering\arraybackslash}m{#1}}

\makeatletter
\newcommand{\vast}{\bBigg@{4}}
\newcommand{\Vast}{\bBigg@{5}}
\makeatother

\DeclareMathOperator{\LF}{LF}
\DeclareMathOperator{\TU}{TU}

\DeclareMathOperator{\co}{\mathbf{co}}
\DeclareMathOperator{\cx}{\mathbf{cx}}

\definecolor{azure(colorwheel)}{rgb}{0.0, 0.5, 1.0}
\definecolor{brandeisblue}{rgb}{0.0, 0.44, 1.0}
\definecolor{ceruleanblue}{rgb}{0.16, 0.32, 0.75}
\definecolor{airforceblue}{rgb}{0.36, 0.54, 0.66}
\definecolor{bleudefrance}{rgb}{0.19, 0.55, 0.91}
\definecolor{darkspringgreen}{rgb}{0.09, 0.45, 0.27}
\definecolor{rulecolor}{rgb}{0.13, 0.35, 0.5}

\usepgfplotslibrary{fillbetween,decorations.softclip}
\usetikzlibrary{arrows.meta}
\usetikzlibrary{patterns,calc}

\newcommand{\papertitle}{{\bf Topping Up and Optimal Redistribution}\footnote[$*$]{\thanktext}
}

\newcommand{\name}{Zi Yang Kang\footnote[$\dag$]{\affiliationi}\qquad \qquad Mitchell Watt\footnote[$\ddag$]{\affiliationii}}

\newcommand{\affiliationi}{Department of Economics, University of Toronto; \href{mailto:zy.kang@utoronto.ca}{\tt zy.kang@utoronto.ca}.}
\newcommand{\affiliationii}{Department of Economics, Monash University; \href{mailto:mitch.watt@monash.edu}{\tt mitch.watt@monash.edu}.}

\newcommand{\thanktext}{We are especially indebted to Laura Doval, Paul Milgrom, Marcin P\k{e}ski, and Andy Skrzypacz for  illuminating discussions.  We thank Piotr Dworczak, Liran Einav, Marina Halac, Caroline Hoxby, Ravi Jagadeesan, Rishabh Kirpalani, Anton Kolotilin, Simon Loertscher, Guillaume Roger, Ilya Segal, Xianwen Shi, Philipp Strack, Satoru Takahashi, Filip Tokarski, and Kai Hao Yang for comments and suggestions.   We gratefully acknowledge the support of the Center of Mathematical Sciences and Applications at Harvard University (Kang) and the Koret Fellowship, the Ric Weiland Graduate Fellowship, and the Gale and Steve Kohlhagen Fellowship in Economics at Stanford University (Watt).  This paper subsumes and extends our previous two working papers circulated under the titles of ``Optimal In-Kind Redistribution'' and ``Optimal Redistribution Through Subsidies.''  {\tt Refine.ink} was used to check the paper for consistency and clarity.
}
\newcommand{\abstracttext}{
\noindent 
This paper studies how topping up---allowing recipients of in-kind transfers to supplement subsidized consumption in a private market---affects optimal redistribution.  Consumers can access a competitive private market, while a social planner offers an alternative nonlinear price schedule.  We show that the effect of topping up depends on the correlation between redistributive priority and demand. When the correlation is positive, topping up does not affect the optimal mechanism. When the correlation is negative, topping up weakens screening and reduces redistribution.  At the extensive margin, topping up reduces the set of environments in which intervention is optimal.  At the intensive margin, topping up weakly reduces both the scope of a free public option and the mass of consumers served, and shifts redistribution away from the consumers with the highest redistributive priority.  We characterize the optimal mechanisms and show how topping up changes the comparative statics of optimal redistribution with respect to redistributive priorities. 
\\[2ex]
{\em JEL classification}: D47, D63, D82, H21, H42\\[1ex]
{\em Keywords}: in-kind redistribution, topping up, public provision, subsidies, optimal mechanism design}
\newcommand{\inserttitle}{\begin{center}\hfill\\[0ex]
\LARGE\papertitle\\[3ex]
\large
\name\\[3ex]
{\monthyeardate\today}
\\[4ex]
\end{center}}	

\begin{document}
\onehalfspacing
\inserttitle

\begin{abstract}
\abstracttext
\end{abstract}
\setstretch{1.5}
\thispagestyle{empty}

\clearpage

\setcounter{page}{1}

\section{Introduction}
\label{sec:introduction}

Governments often redistribute in kind.  Food stamps, public housing, childcare subsidies, and health care programs all transfer resources to eligible individuals by changing the effective prices they face.  These programs differ, however, in an important respect: whether recipients can \emph{top up} their subsidized consumption through private-market purchases.  For example, food stamp recipients may use their benefits to purchase groceries and then supplement those benefits with cash, whereas public housing recipients generally cannot supplement the housing services provided by the program in the same way.  This distinction is often itself a design choice: housing assistance, for instance, can be delivered through housing vouchers rather than direct public provision, making it easier for recipients to rent more expensive apartments by supplementing their vouchers with cash.

Importantly, topping up can affect the social planner's ability to screen consumers in redistribution programs.  When topping up is allowed, high-demand consumers can choose subsidized allocations intended for low-demand consumers and then purchase additional units privately.  However, preventing topping up may itself be costly, because it can require monitoring supplementary purchases or restricting recipients' ability to combine program benefits with private spending.  While such administrative and enforcement costs may often be estimated directly, the effect of topping up on screening is much less transparent.  Despite a growing literature on inequality-aware mechanism design, the role of topping up in shaping optimal redistribution programs remains poorly understood.

This paper studies how topping up affects the optimal design of redistribution programs.  We ask three questions.  First, how does topping up affect when the social planner should intervene?  Second, when does topping up change the optimal mechanism?  Third, when topping up changes the mechanism, how does it change the structure of optimal redistribution?

To answer these questions, we develop a mechanism-design model that accommodates both topping up and no topping up.  Consumers differ in their demand for a homogeneous good and can purchase any quantity in a perfectly competitive private market at marginal cost.  A social planner can offer an alternative, potentially nonlinear, price schedule for the same good.  Following the recent literature on inequality-aware mechanism design (e.g., \citealp{dworczaketal21}; \citealp{akbarpouretal24}), we model the social planner's redistributive objective using heterogeneous welfare weights.  The welfare weight $\omega(\theta)$ represents the gross social value of giving one unit of money to a type-$\theta$ consumer, where higher types have higher demand.  A higher welfare weight thus captures higher redistributive priority, such as lower income, disability, or other dimensions of social need.  To focus on in-kind redistribution, we rule out lump-sum cash transfers within the mechanism.  Instead, we capture the opportunity cost of public funds by assigning a welfare weight $\a$ to profit.

The key feature of our model is that the social planner has only partial control of the market.  On the one hand, the social planner faces a screening problem: she wants to distort consumption in order to redirect surplus toward consumers with higher welfare weights.  On the other hand, consumers retain access to the private market.  This generates participation constraints that limit the social planner's ability to distort consumption and redirect surplus.  Topping up changes these participation constraints: with topping up, each consumer can purchase from both the social planner's price schedule and the private market; but without topping up, each consumer must choose between the two.

Our first main result characterizes the extensive margin: when the social planner should intervene.  Without topping up, the social planner can strictly improve on laissez-faire if and only if some consumer's welfare weight exceeds the opportunity cost of public funds:
\[\max_{\hat\theta}\omega(\hat\theta)>\a.\]
With topping up, the corresponding condition is stronger and requires the maximal upper-tail average welfare weight to exceed the opportunity cost of public funds:
\[\max_{\hat\theta}\E\!\left[\omega(\theta)\mid\theta\geq\hat\theta\right]>\a.\]
This comparison gives a precise, yet intuitive, sense in which topping up weakens screening.  A small subsidy targeted at type $\hat\theta$ cannot benefit only that type when topping up is allowed: all higher types can also take the subsidized allocation and supplement in the private market.  The first-order benefit of intervention is thus governed by the average welfare weight in the upper tail rather than by the welfare weight of the targeted type alone.  We formalize this intuition to show that these are precisely the relevant sufficient statistics that determine when the social planner should intervene.

This extensive-margin result has sharp implications.  First, suppose welfare weight and demand are positively correlated, so that higher-demand consumers have higher redistributive priority---such as in disability care programs.  In this case, topping up does not affect the social planner's intervention decision, because
\[\max_{\hat\theta}\E\!\left[\omega(\theta)\mid\theta\ge \hat\theta\right]=\max_{\hat\theta}\omega(\hat\theta).\]
Second, suppose welfare weight and demand are negatively correlated, so that lower-demand consumers have higher redistributive priority---such as in housing assistance programs.  In this case, topping up reduces the scope for intervention, because
\[\max_{\hat\theta}\E\!\left[\omega(\theta)\mid\theta\ge \hat\theta\right]=\E[\omega]\leq\max_{\hat\theta}\omega(\hat\theta),\]
with strict inequality whenever welfare weights are not constant.  Thus, the direction of correlation between welfare weights and demand determines when topping up changes the intervention decision.

Our second main result characterizes the intensive margin: when topping up changes the optimal mechanism.  Whereas the extensive margin is governed by simple sufficient statistics, the intensive margin depends on the entire distribution of welfare weights.  The key object is a cutoff $\theta^L$, the lowest type above which the upper-tail average welfare weight exceeds the opportunity cost of public funds.  We show that topping up and no topping up yield the same optimal mechanism if and only if two conditions simultaneously hold: \emph{(i)}~there is no redistributive motive to intervene below $\theta^L$, and \emph{(ii)}~the optimal mechanism without topping up weakly expands the consumption of consumers above $\theta^L$.

To show how topping up changes the optimal mechanism, we then characterize optimal mechanisms under two benchmark correlations between welfare weights and demand.  By solving for the optimal allocation functions explicitly, we also obtain clean comparative static results on how the distribution of welfare weights affects optimal redistribution.

Under positive correlation, topping up does not affect the optimal mechanism.  Thus, both the social planner's extensive-margin and intensive-margin decisions are unaffected by topping up.  The reason is self-targeting: higher-demand consumers have higher redistributive priority, so topping up does not create leakage away from the intended beneficiaries.  In this case, the optimal mechanism has a lower-cutoff structure, shown in \Cref{fig:summary_positive}.  High types are served by a nonlinear subsidy program, while low types are either served by a free public option or left to the private market, depending on whether the average welfare weight exceeds the opportunity cost of public funds.  

Under negative correlation, however, topping up can change the optimal mechanism.  Without topping up, the social planner can use low quantities to screen consumers: high-demand consumers must give up their private-market consumption if they choose subsidized allocations intended for low-demand consumers.  With topping up, the social planner's screening ability is weakened because high-demand consumers can consume subsidized allocations and then supplement privately.  In each case, the optimal mechanism may involve a combination of a free public option, a nonlinear subsidy program, and the private market, as shown in \Cref{fig:summary_negative}.  However, topping up reduces redistribution by reducing the scope of intervention, the use of a free public option, and the mass of consumers served.  Moreover, topping up shifts redistribution away from consumers with higher redistributive priority to those with lower redistributive priority.

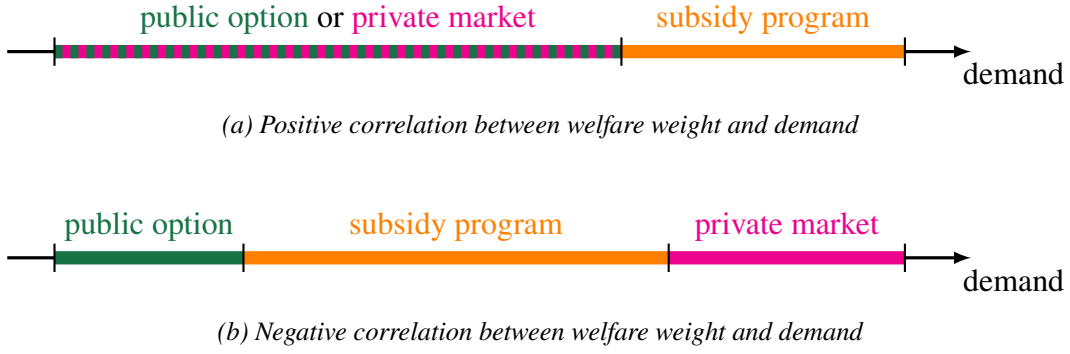
\begin{figure}[t!]
\centering
\begin{subfigure}[t]{\textwidth}
\centering
\begin{tikzpicture}[scale=1.25]
\draw[-latex,very thick] (-5,0) -- (5.2,0) ; 
\node[below right] at (5, 0.) {demand};

\draw[line width=5pt,magenta] (-4.5,0) -- (-1.5,0) node[above] {{\color{darkspringgreen}public option} {\color{black}or} {\color{magenta}private market}} -- (1.5,0);
\draw[line width=5pt,darkspringgreen,dashed] (-4.5,0) -- (1.5,0);
\draw[line width=5pt,orange] (1.5,0) -- (3.,0) node[above] {subsidy program} -- (4.5,0);
\draw[thick] (-4.5, 0.15) -- (-4.5, -0.15);
\draw[thick] (1.5, 0.15) -- (1.5, -0.15);
\draw[thick] (4.5, 0.15) -- (4.5, -0.15);

\end{tikzpicture}
\caption{Positive correlation between welfare weight and demand}
\label{fig:summary_positive}
\end{subfigure}
\vspace{20pt}

\begin{subfigure}[t]{\textwidth}
\centering
\begin{tikzpicture}[scale=1.25]
\draw[-latex,very thick] (-5,0) -- (5.2,0) ; 
\node[below right] at (5, 0.) {demand};

\draw[line width=5pt,darkspringgreen] (-4.5,0) -- (-3.5,0) node[above] {public option} -- (-2.5,0);
\draw[line width=5pt,orange] (-2.5,0) -- (-0.25,0) node[above] {subsidy program} -- (2,0);
\draw[line width=5pt,magenta] (2,0) -- (3.25,0) node[above] {private market} -- (4.5,0);
\draw[thick] (-4.5, 0.15) -- (-4.5, -0.15);
\draw[thick] (-2.5, 0.15) -- (-2.5, -0.15);
\draw[thick] (2, 0.15) -- (2, -0.15);
\draw[thick] (4.5, 0.15) -- (4.5, -0.15);

\end{tikzpicture}
\caption{Negative correlation between welfare weight and demand}
\label{fig:summary_negative}
\end{subfigure}\vspace{15pt}

\caption{Graphical depictions of the optimal mechanism.}\label{fig:summary}
\end{figure}

Our comparative static results further highlight the differences between topping up and no topping up.  With topping up, a pointwise increase in welfare weights expands redistribution on all margins, including a pointwise weakly higher optimal allocation.  Without topping up, this monotonicity can fail: under negative correlation, a pointwise increase in welfare weights can reduce the optimal allocations of high-demand consumers.  This is because low quantities retain screening value when topping up is not allowed: the social planner may optimally distort subsidized allocations for higher-demand consumers downward to redirect surplus toward the bottom.

This difference is even sharper under mean-preserving spreads of welfare weights, when demand becomes statistically more informative about redistributive priority.  With topping up, a mean-preserving spread of welfare weights \emph{reduces} redistribution under negative correlation: the loss from leakage to higher types always dominates the gain in statistical informativeness.  By contrast, without topping up, a mean-preserving spread expands the scope of a free public option.  Thus, topping up changes not just optimal redistribution, but also how optimal redistribution responds to changes in redistributive priorities.

From a methodological perspective, we obtain explicit characterizations and comparative statics by solving two mechanism-design problems with type-dependent outside options.  With topping up, the social planner faces a pointwise lower-bound constraint requiring each type's consumption to remain above his laissez-faire consumption.  Without topping up, private-market access instead generates participation constraints that can be written as majorization constraints.  These problems are related, but mathematically distinct.  Recent work solves linear programs with majorization constraints using extreme-point methods (e.g., \citealp{kleineretal21}); by contrast, our programs maximize strictly concave objectives, so the optimum need not be an extreme point of the feasible set.  Instead, we use a Lagrangian approach and generalized ironing methods.  Whereas \cite{jullien00} uses duality methods to show the existence of Lagrange multipliers for such problems, we explicitly characterize the relevant Lagrange multipliers in terms of primitives by adapting the approach of \cite{amadorbagwell13}.  In this sense, our approach is closer to that of \cite{dworczakmuir24}, who explicitly solve a class of linear programs with majorization constraints in this way.\footnote{A contemporaneous contribution by \citet{kolotilinwolitzky26} also studies concave maximization with a majorization constraint.  In their many-to-one matching model, majorization describes feasible workforce compositions, and the solution is characterized using dual prices and first-order conditions.  In our no-topping-up problem, majorization constraints instead arise from type-dependent participation constraints.  We construct the corresponding Lagrange multipliers directly from primitives, yielding generalized-ironing formulas and comparative statics.}

\subsection{Related Literature}

Our paper contributes to the growing literature on inequality-aware mechanism design, which builds on \citepos{weitzman77} observation that departures from competitive allocations can facilitate redistribution.  More recent work, including by \citet{cheetal13}, \citet{condorelli13}, \citet{dworczaketal21}, \citet{akbarpouretal24}, and \citet{paistrack24}, uses mechanism-design tools to characterize optimal redistributive allocations.  Whereas most of this literature studies settings in which the social planner designs the entire market, we study a social planner with only partial control of the market.  Consumers can always access a competitive private market, and this outside option generates participation constraints that shape the optimal redistribution program.  

The presence of a private market also connects our paper to the literature on partial mechanism design, where only part of the market can be designed.  This includes work on optimal interventions in markets with adverse selection (\citealp{philipponskreta12}; \citealp{tirole12}; \citealp{fuchsskrzypacz15}), pricing with resale \citep{loertschermuir22}, and optimal redistribution with a public option \citep{kang23}.  Relative to this literature, we study a novel---but economically important---benchmark in which the social planner and the private market have access to the same production technology.  This represents a benchmark in which the social planner can costlessly contract with private producers to supply the good.  This benchmark also presents a greater technical challenge as it precludes any simplifications to the mechanism-design problem that may arise from productive-efficiency and technological asymmetries between the social planner and private producers.

Finally, our paper also complements a large public finance literature on in-kind redistribution and public provision of private goods.  Starting with \citet{nicholszeckhauser82}, this literature emphasizes that in-kind transfers can screen consumers better than cash transfers when intended beneficiaries value the subsidized good differently from other consumers.  Subsequent work studies how private markets affect this screening logic, including by \citet{blackorbydonaldson88}, \citet{besleycoate91}, and \citet{gahvarimattos07}.\footnote{Following \citet{ramsey27} and \citet{diamond75}, a separate literature studies how commodity taxes should trade off efficiency and redistributive concerns.  The Atkinson--Stiglitz theorem \citep{atkinsonstiglitz76} provides a benchmark in which differential commodity taxation is unnecessary when nonlinear income taxation is available and preferences are suitably separable.  This benchmark does not apply in our environment because consumers' demand for the good is informative about their welfare weights.  In this sense, our paper is closer to work emphasizing preference heterogeneity as a reason for goods-market interventions, including by \citet{saez02} and \citet{doligalskietal25}.  Relative to this literature, we focus on the design of nonlinear redistribution programs when the social planner cannot freely redistribute in cash and consumers retain access to a private market.}  The closest paper in this literature is that of \citet{blomquistchristiansen98}, who study whether public provision should allow recipients to top up their consumption in a model of labor supply \emph{\`a la} \cite{mirrlees71}.  Our paper differs in two respects.  First, we focus on screening through heterogeneous consumption preferences.  Second, we allow the social planner to design an arbitrary nonlinear price schedule rather than choosing only whether to provide a fixed quantity of the good, which allows us to characterize how topping up affects both the extensive margin of intervention and the intensive margin of optimal mechanism design.

\subsection{Organization}

The rest of the paper proceeds as follows.  \Cref{sec:framework} introduces the model, and \Cref{sec:main} characterizes how topping up affects the extensive and intensive margins of optimal redistribution.  \Cref{sec:positive} studies the case of positive correlation between welfare weights and demand, while \Cref{sec:negative} studies the case of negative correlation.  Section~\ref{sec:conclusion} concludes.  Proofs are collected in \Cref{app:proofs,app:online}.

\section{Framework}
\label{sec:framework}

In this section, we develop a framework that accommodates both topping up and no topping up.  We begin with a standard model of a private market and describe the laissez-faire benchmark.  We then formulate the social planner's problem under the two participation structures.

\subsection{Setup}

There is a unit mass of consumers.  Each consumer demands a quantity $q\in[0,A]$ of a homogeneous, perfectly divisible good, where $A$ is the maximum feasible quantity.  As is standard, $q$ can equivalently be interpreted as the quality level of an indivisible good of which each consumer consumes one unit.  Although that interpretation may be more natural in some applications, we use the language of quantity throughout for consistency.  The good is supplied competitively in a private market by producers with constant marginal cost $c>0$.

Consumers have quasilinear preferences over money and differ in their demand for the good.  Each consumer privately observes his type $\theta\in[\und\theta,\BAR\theta]\subset\R_{++}$, drawn from a continuously differentiable distribution $F$ with positive density $f$ on $[\und\theta,\BAR\theta]$.  A type-$\theta$ consumer who consumes $q$ units of the good and makes total payment $t$ obtains utility
\[\theta v(q)-t.\]
We assume that $v:[0,A]\to\R_+$ is twice continuously differentiable, increasing,\footnote{Throughout, ``increasing'' and ``decreasing'' mean strictly increasing and strictly decreasing, respectively; weak monotonicity is denoted by ``nondecreasing'' and ``nonincreasing.''} and strictly concave with $v''<0$, so consumers have positive but diminishing marginal utility from consumption.

\paragraph{Laissez-faire benchmark.} Given a constant marginal price $p$ for the good, denote the demand curve of a type-$\theta$ consumer by
\[D(p,\theta)\in\argmax_{q\in[0,A]}\left[\theta v(q)-pq\right].\]
Because $v$ is strictly concave, $D(p,\theta)$ is uniquely defined for any $p,\theta>0$.  For notational convenience, extend the domain of $(v')^{-1}$ to $\R$ by setting $(v')^{-1}(z)=0$ for $z\geq v'(0)$ and $(v')^{-1}(z)=A$ for $z\leq v'(A)$; then $D(p,\theta) = (v')^{-1}(p/\theta)$ for $p,\theta>0$.

Under perfect competition, the good is sold at marginal cost $c$.  Let $q^{\LF}(\theta)\coloneq D(c,\theta)$ denote the laissez-faire consumption of a type-$\theta$ consumer, and denote his corresponding laissez-faire utility by $U^{\LF}(\theta)$.  For simplicity, we assume throughout that laissez-faire demand is interior: $\im q^{\LF}\subseteq (0,A)$.  

\subsection{Social Planner's Problem}

A social planner seeks to redistribute by offering an alternative price schedule for the good.  Unlike the private-market price schedule, which is linear under perfect competition, the social planner's price schedule may be nonlinear.  

We assume that the social planner has access to the same production technology as the private market and hence faces the same marginal cost $c$.  Equivalently, the social planner may be able either to contract with producers or to reimburse consumers for private purchases at no additional resource cost.  This assumption allows us to isolate the effects of topping up on screening, abstracting from other frictions.

We can therefore formulate the social planner's problem in terms of {\em total} consumption and payment.  The social planner chooses a direct mechanism $(q,t)$ consisting of:
\begin{enumerate}[label={\em(\roman*)}]
\item an allocation function $q:[\und\theta,\BAR\theta]\to[0,A]$, where $q(\theta)$ is the consumer's total consumption; and
\item a payment function $t:[\und\theta,\BAR\theta]\to\R$, where $t(\theta)$ is the consumer's total payment.
\end{enumerate}
This formulation implicitly restricts attention to deterministic mechanisms; as we show in \Cref{app:proofs}, this restriction is without loss of generality.

We require the mechanism to satisfy incentive compatibility, so truthful reporting is optimal:
\begin{equation}\label{eq:IC}
    \theta\in\argmax_{\theta'\in[\und\theta,\BAR\theta]}\left[\theta v(q(\theta'))-t(\theta')\right],\qquad\text{for all }\theta\in[\und\theta,\BAR\theta].\tag{IC}
\end{equation}
By the revelation principle \citep{myerson79}, this restriction is without loss of generality.

We also focus on mechanisms for in-kind redistribution by ruling out lump-sum cash transfers within the mechanism.  Specifically, we impose the no-lump-sum-transfers constraint
\begin{equation}\label{eq:LS}
t(\theta)\geq 0,\qquad\text{for all }\theta\in[\und\theta,\BAR\theta].\tag{LS}
\end{equation}  
Thus, while the social planner may fully subsidize some units of the good---that is, set their price to zero---she cannot make direct cash transfers through the mechanism.  Cash transfers outside the mechanism are nonetheless permitted in reduced form through the social planner's objective, via the weight assigned to the monetary surplus generated by the mechanism, as described below.

The key difference between topping up and no topping up lies in how private market access constrains feasible mechanisms:
\begin{enumerate}[label=\emph{(\roman*)}]
\item {\bf No topping up.} When topping up is not allowed, each consumer may purchase the good either from the social planner or in the private market, but not both.  The consumer's outside option is therefore the laissez-faire utility:
\begin{equation}\label{eq:IR}
    \theta v(q(\theta)) - t(\theta) \geq U^{\LF}(\theta),\qquad\text{for all }\theta\in[\und\theta,\BAR\theta].\tag{IR}
\end{equation}

\item {\bf Topping up.} When topping up is allowed, each consumer may supplement consumption purchased from the social planner by buying additional units in the private market at marginal price $c$.  Equivalently, the social planner may subsidize marginal units of consumption, but cannot raise the consumer's effective marginal price above $c$.  Relative to the no-topping-up case, this imposes an additional feasibility constraint on total consumption:
\begin{equation}\label{eq:TU}
q(\theta)\geq q^{\LF}(\theta),\qquad\text{for all }\theta\in[\und\theta,\BAR\theta].\tag{TU}
\end{equation}
\end{enumerate}

This distinction captures the key tradeoff between flexibility and targeting.  Holding fixed the social planner's price schedule, topping up weakly benefits consumers because it expands their choice set.  At the same time, it restricts the set of feasible mechanisms available to the social planner, because the mechanism must now deter consumers from a broader set of deviations from their intended allocations.  

The social planner evaluates mechanisms according to weighted total surplus, consisting of consumer surplus and profit (\ie, total revenue net of production cost):
\begin{enumerate}[label={\em(\roman*)}]
\item {\bf Consumer surplus.}  The social planner assigns welfare weight $\omega(\theta)$ to a consumer of type $\theta$.  This is the gross social value of transferring one unit of money to that consumer.\footnote{\citet{dworczaketal21} microfound such welfare weights as expected marginal utilities of money conditional on consumers' demand types.} Throughout, we assume for simplicity that $\omega:[\und\theta,\BAR\theta]\to\R_+$ is continuous.  

\item {\bf Profit.}  The social planner assigns welfare weight $\a\in\R_{++}$ to the monetary surplus generated by the mechanism, namely revenue net of production cost.  Thus, $\a$ captures the opportunity cost of public funds.  This opportunity cost is high when, for example, the social planner has alternative redistributive instruments---such as cash transfers---that compete for the same funds, and low when the social planner faces political constraints\footnote{For instance, \cite{liscowpershing22} show that the general population in the U.S.~largely prefers in-kind redistribution to cash transfers.} that limit redistribution outside the mechanism.
\end{enumerate}

Accordingly, the social planner chooses a feasible mechanism to maximize weighted total surplus:
\begin{align}
\max_{(q,t)}&\int_{\und\theta}^{\BAR\theta}\bigg[\omega(\theta)\underbrace{\left[\theta v(q(\theta))-t(\theta)\right]}_{\text{consumer surplus}} + \a\underbrace{\left[t(\theta)-cq(\theta)\right]}_{\text{total profit}}\bigg]\,\dd F(\theta),\tag{OBJ}\label{eq:OBJ}\\
\text{s.t. }&(q,t) \text{ satisfies \eqref{eq:IC}, \eqref{eq:IR}, \eqref{eq:LS}, and, when topping up is allowed, \eqref{eq:TU}}.\tag*{}
\end{align}
The set of feasible mechanisms is nonempty, since the laissez-faire mechanism $(q^{\LF},cq^{\LF})$ is feasible under both participation structures.  In \Cref{app:proofs}, we show that each program admits an essentially unique solution.\footnote{Throughout, we consider mechanisms to be equivalent if their allocations differ only on a set of measure zero.  We show the optimal allocation function is unique in both programs.  For any given allocation function, incentive compatibility pins down the payment function up to an additive constant.  When $\E[\omega]\ne\alpha$, the social planner's objective pins down this constant at one endpoint of the feasible interval.  When $\E[\omega]=\alpha$, the objective is independent of this constant, so there may be multiple optimal payment functions implementing the same allocation.  In that case, whenever a selection is needed, we select the pointwise-smallest feasible payment function.  This selection is consumer-optimal among optimal implementations because it gives every type the highest feasible utility conditional on the optimal allocation.}  We write $(q^*,t^*)$ for the optimal mechanism without topping up and $(q^*_{\text{TU}},t^*_{\text{TU}})$ for the optimal mechanism with topping up.

\section{Topping Up vs.~No Topping Up}
\label{sec:main}

We now compare optimal redistribution with and without topping up.  We first characterize the extensive margin of intervention: when the optimal mechanism strictly improves on the laissez-faire outcome under each participation structure.  We then turn to the intensive margin: when, conditional on intervention, topping up changes the optimal mechanism itself.

\subsection{How Does Topping Up Change When the Social Planner Intervenes?}

Our first main result addresses the extensive margin of intervention by characterizing the intervention region under each participation structure.\footnote{For ease of notation, we extend the domain of the upper-tail conditional expectation function $\hat\theta\mapsto\E\!\left[\omega(\theta)\mid\theta\geq\hat\theta\right]$ to $\hat\theta=\BAR\theta$ by continuity.  By the Lebesgue differentiation theorem, $\E\!\left[\omega(\theta)\mid\theta\geq \BAR\theta\right]\coloneq \lim_{\hat\theta\to\BAR\theta}\E\!\left[\omega(\theta)\mid\theta\geq\hat\theta\right]=\omega(\BAR\theta)$.}

\begin{theorem}\label{thm:scope}
The optimal mechanism strictly improves on the laissez-faire outcome if and only if
\[\begin{dcases}
\max_{\hat\theta\in[\und\theta,\BAR\theta]}\E\!\left[\omega(\theta)\mid\theta\geq\hat\theta\right]>\a,&\text{when topping up is allowed},\\
\max_{\hat\theta\in[\und\theta,\BAR\theta]}\omega(\hat\theta)>\a,&\text{when topping up is not allowed}.
\end{dcases}\]
\end{theorem}

\Cref{thm:scope} shows that the extensive-margin intervention decision is pinned down by two sufficient statistics: with topping up, the maximal upper-tail average welfare weight; and without topping up, the pointwise maximal welfare weight.  If all welfare weights are equal to $\a$, then \Cref{thm:scope} implies that laissez-faire is optimal under both participation structures, which is the First Welfare Theorem benchmark in the present environment.  More generally, departures of welfare weights from $\a$ can justify intervention, and topping up matters precisely because these two sufficient statistics need not coincide.

By \Cref{thm:scope}, topping up replaces the pointwise maximal welfare weight with the maximal upper-tail average welfare weight in determining the intervention decision.  Since averages cannot exceed maxima,
\[\max_{\hat\theta\in[\und\theta,\BAR\theta]} \omega(\hat\theta) \geq \max_{\hat\theta\in[\und\theta,\BAR\theta]}\E\!\left[\omega(\theta)\mid\theta\geq\hat\theta\right].\]
It follows that the social planner intervenes in weakly fewer environments when topping up is allowed.  Equality holds if and only if the highest type has the highest welfare weight: $\omega(\BAR\theta)=\max_{\hat\theta\in[\und\theta,\BAR\theta]}\omega(\hat\theta)$.

An immediate corollary characterizes \emph{when} topping up changes the optimal intervention decision:

\begin{corollary}\label{cor:scope}
Topping up changes the optimal intervention decision if and only if
\[\max_{\hat\theta\in[\und\theta,\BAR\theta]} \omega(\hat\theta) >\a\geq \max_{\hat\theta\in[\und\theta,\BAR\theta]}\E\!\left[\omega(\theta)\mid\theta\geq\hat\theta\right].\]
\end{corollary}

Taken together, \Cref{thm:scope} and \Cref{cor:scope} identify a precise sense in which topping up weakens screening: it reduces the social planner's marginal incentive to intervene.  The intuition is easiest to see from the first-order welfare effects of a small subsidy:
\begin{enumerate}[label=\emph{(\roman*)}]
\item {\bf Topping up.} When topping up is allowed, the subsidy cannot be targeted to a single type in isolation.  If the social planner subsidizes consumption beyond the laissez-faire quantity of some type $\hat\theta$, then all higher types benefit, since they can always mimic the consumer of type $\hat \theta$ and then supplement their subsidized allocations in the private market.  The first-order benefit of intervention therefore depends on whether the upper-tail average welfare weight $\E\!\left[\omega(\theta)\mid\theta\geq\hat\theta\right]$ exceeds $\a$.  Although some nearby lower types may also distort their consumption upward to qualify for the subsidy, those distortion costs are local and higher-order for sufficiently small interventions.
\item {\bf No topping up.} When topping up is not allowed, the social planner can target the subsidy more narrowly as consumers must choose between the subsidized allocation and the private market.  It therefore suffices that some type has welfare weight above $\a$.
\end{enumerate}
We prove \Cref{thm:scope} in \Cref{app:proofs} by formalizing this intuition.

\subsection{When Does Topping Up Change the Optimal Mechanism?}

While \Cref{thm:scope} shows how topping up reduces the social planner's marginal incentive to intervene, it does not yet say when topping up changes the solution to the social planner's full problem.  As such, we now compare the optimal mechanisms across the two participation structures.

Our second main result characterizes when topping up changes the optimal mechanism.  To state it, we introduce three objects.  First, for any $\mu\geq0$, define the \emph{generalized virtual valuation},
\begin{equation}\label{eq:virtual}
J_\mu(\theta)\coloneq \theta + \frac{\mu\und\theta\cdot\d_{\und\theta}(\theta) + \int_\theta^{\BAR\theta}\left[\omega(s)-\a\right]\,\dd F(s)}{\a f(\theta)}.
\end{equation}
Here, $\mu\geq0$ is the shadow cost of the \eqref{eq:LS} constraint.  Thus, $J_\mu(\theta)$ captures the social value of increasing the allocation of type $\theta$ once the in-kind restriction and incentive constraints are taken into account.  Second, following \citet{myerson81} and \citet{toikka11}, define the \emph{ironing operator} as follows.  For any generalized function $h$ with domain $[\und\theta,\BAR\theta]$, let its ironed version $\BAR{h}:[\und\theta,\BAR\theta]\to\R$ be the nondecreasing function 
\[\BAR h(\theta) \coloneq \left.-\frac{\dd}{\dd z}\(z\mapsto \co\int_{F^{-1}(z)}^{\BAR\theta} h(s)\ \dd F(s)\)\right|_{z=F(\theta)},\]
where $\co H$ is the pointwise smallest concave function that weakly exceeds a given function $H:[0,1]\to\R$.  More generally, for any interval $I\subseteq[\und\theta,\BAR\theta]$, we write $\BAR{h|_I}$ for the ironing of the restriction of $h$ to $I$, defined analogously with the concavification over the corresponding quantile interval $F(I)$.  Finally, define the \emph{topping-up cutoff type},
\begin{equation}
\theta^L\coloneq \min\left\{\hat\theta\in[\und\theta,\BAR\theta]:\int_{\hat\theta}^{\BAR\theta}\left[\omega(s)-\a\right]\,\dd F(s)\geq0\right\}.\label{eq:cutoff}
\end{equation}
Equivalently, $\theta^L$ is the lowest type above which the upper-tail average welfare weight weakly exceeds $\a$.  

\clearpage
\begin{theorem}\label{thm:optimal}
Let $\mu^*\coloneq\(\E[\omega]-\a\)_+$.  The optimal mechanisms with and without topping up coincide if and only if:
\[\begin{dcases}
\BAR{J_{\mu^*}|_{[\theta^L,\BAR\theta]}}(\theta)\geq\theta,&\text{for }\theta^L\leq \theta\leq\BAR\theta,\\
\omega(\theta)\leq \a,&\text{for }\und\theta\leq\theta<\theta^L.
\end{dcases}\]
Moreover, if these conditions are satisfied, then the optimal allocation function is
\[q^*(\theta) = \begin{dcases}
D\!\(c,\BAR{J_{\mu^*}|_{[\theta^L,\BAR\theta]}}(\theta)\),&\text{for }\theta^L\leq\theta\leq\BAR\theta,\\
q^{\LF}(\theta),&\text{for }\und\theta\leq\theta<\theta^L.
\end{dcases}\]
\end{theorem}

\Cref{thm:optimal} provides a sharp characterization of when topping up changes the optimal mechanism.  The optimal mechanisms with and without topping up coincide if and only if the social planner has no redistributive motive to intervene below $\theta^L$ and, above $\theta^L$, the no-topping-up solution already distorts consumption upward relative to the laissez-faire allocation.

Although \Cref{thm:optimal} is proven formally in \Cref{app:proofs}, the intuition is straightforward.  Below $\theta^L$, topping up is most restrictive because any subsidized option can be supplemented by higher types in the private market, making narrowly targeted transfers difficult to implement.  If some type below $\theta^L$ has welfare weight above $\a$, then by \Cref{thm:scope} the social planner would like to intervene there without topping up, but cannot do so in the same way once topping up is allowed.  Above $\theta^L$, intervention becomes potentially socially desirable, and the relevant question is whether the optimal allocation without topping up satisfies the \eqref{eq:TU} constraint.  The corresponding condition in \Cref{thm:optimal} characterizes precisely when that constraint is satisfied for these types.

\Cref{thm:optimal} also clarifies how the intensive-margin effects of topping up differ from the extensive-margin effects in \Cref{thm:scope}.  The intervention decision depends only on two sufficient statistics of the welfare weights, whereas whether topping up changes the optimal mechanism depends on the entire distribution of welfare weights.

In the next sections, we study how topping up changes the structure of optimal redistribution.

\section{Optimal Redistribution: Positive Correlation}
\label{sec:positive}

We now study the case in which welfare weights are positively correlated with demand.

\begin{assumption}[positive correlation]\label{ass:positive}
The welfare weight function $\omega:[\und\theta,\BAR\theta]\to\R_+$ is increasing.
\end{assumption}

\Cref{ass:positive} captures environments in which the social planner's redistributive priorities are aligned with consumer demand.  A natural example is disability care: social programs may place greater weight on individuals with more severe disabilities, whose consumption of care services is likely to be higher.  

Under \Cref{ass:positive}, the two sufficient statistics in \Cref{thm:scope} coincide:
\[\max_{\hat\theta\in[\und\theta,\BAR\theta]}\E\!\left[\omega(\theta)\mid\theta\geq\hat\theta\right] = \omega(\BAR\theta) = \max_{\hat\theta\in[\und\theta,\BAR\theta]}\omega(\hat\theta).\]
Hence topping up does not affect the extensive margin of optimal redistribution.  In fact, under positive correlation it does not affect the intensive margin either:

\begin{corollary}\label{cor:positive}
Under \Cref{ass:positive}, the optimal mechanisms with and without topping up coincide.
\end{corollary}

\subsection{Characterization of the Optimal Mechanism}

Given \Cref{cor:positive}, we can characterize the optimal mechanism under positive correlation without distinguishing between the two participation structures.

\begin{theorem}\label{thm:positive}
Let $\mu^* \coloneq \(\E[\omega]-\a\)_+$.  Under \Cref{ass:positive}, the optimal allocation function is 
\[q^*(\theta) = \begin{dcases}
D\!\(c,\BAR{J_{\mu^*}|_{[\theta^L,\BAR\theta]}}(\theta)\), &\text{for }\theta^L\leq \theta\leq\BAR\theta,\\
q^{\LF}(\theta), &\text{for }\und\theta\leq \theta < \theta^L.
\end{dcases}\]
\end{theorem}

As in \Cref{thm:optimal}, the structure of the optimal mechanism under positive correlation is governed by two objects: the topping-up cutoff type $\theta^L$, defined in \cref{eq:cutoff}, and the multiplier $\mu^*=\(\E[\omega]-\a\)_+$.  As we now explain, $\theta^L$ determines which consumers are served by the social planner rather than left to the private market, while $\mu^*$ determines whether a free public option is used.

By \Cref{thm:positive}, the optimal allocation has a cutoff structure.  Consumers with $\theta<\theta^L$ receive their laissez-faire allocation and make the corresponding private-market payment, while those with $\theta\geq\theta^L$ are served by the social planner.  Under positive correlation, these higher types have higher welfare weights, and they face a nonlinear subsidy schedule that weakly raises consumption relative to laissez-faire.

The multiplier $\mu^*$ determines whether the bottom of this subsidized region takes the form of a free public option.  When $\mu^*>0$, the generalized virtual valuation $J_{\mu^*}$ has an atom at $\und\theta$, so $\BAR{J_{\mu^*}|_{[\theta^L,\BAR\theta]}}$---and hence the allocation $q^*$---is flat on an initial interval.  Moreover, the proof of \Cref{thm:positive} shows that $\mu^*$ is the shadow cost of the \eqref{eq:LS} constraint.  Since $\mu^*>0$, that constraint binds on this interval, implying that consumers there pay zero. We therefore interpret this flat initial segment of the allocation as a free public option.  By contrast, when $\mu^*=0$, the social planner does not use a free public option.  Any redistribution uses only a nonlinear subsidy schedule, with the associated payment schedule pinned down by the allocations characterized by \Cref{thm:positive} and the envelope theorem.

These observations imply that the optimal mechanism takes one of the following forms:

\begin{enumerate}[label={\em(\roman*)}]

\item If $\a\geq \E[\omega]$, low types are left to the private market, while high types are served by the social planner.  In this case, $\mu^*=0$, so there is no free public option.  Moreover, the definition of $\theta^L$ implies that $\theta^L>\und\theta$ if $\a>\E[\omega]$ (with $\theta^L=\und\theta$ if $\a=\E[\omega]$).  Consumers with $\theta<\theta^L$ receive their laissez-faire allocation, whereas consumers with $\theta\geq\theta^L$ are subsidized and consume weakly more than under laissez-faire.  \Cref{fig:optimal_positive_allocation} illustrates this case.  As a special case, if $\a>\max\omega$, then $\theta^L=\BAR\theta$, so the subsidized region is empty and laissez-faire is optimal, consistent with \Cref{thm:scope}.  By contrast, if $\a=\E[\omega]$, then $\theta^L=\und\theta$ and all consumers are served by a nonlinear subsidy program.

\item If $\a<\E[\omega]$, all consumers are served by the social planner, and the optimal mechanism always includes a free public option as $\mu^*=\(\E[\omega]-\a\)_+>0$.  The lowest types are pooled into the free public option, while higher types purchase additional quantities through the social planner's nonlinear subsidy program.  \Cref{fig:optimal_positive_allocation_another} illustrates this case.
\end{enumerate}

\begin{figure}[t!]
\centering

\begin{subfigure}[t]{\textwidth}
\centering
\begin{tikzpicture}[scale=1.15]

\draw[scale=1, domain=1:11., densely dashed, variable=\x, line width=2.pt, color=black!50] plot ({\x}, {0.5*\x});
\draw (11.2,5.5) node[right] {$\color{black!50}q^{\LF}$};

\draw[scale=1, domain=5:11., smooth, variable=\x, line width=3pt, color=ceruleanblue] plot ({\x}, {-\x^2/12+11*\x/6-55/12});
\draw (8,5) node[above] {$\color{ceruleanblue}q^*$};

\draw[line width=1pt, dashed, color=black] (5,2.5) -- (5,0.15);

\draw[scale=1, domain=1:5., smooth, variable=\x, line width=3pt, color=ceruleanblue] plot ({\x}, {0.5*\x});

\draw (3,0.1) node [above] {{\color{magenta} private market}};
\draw (8,0.1) node [above] {{\color{orange} subsidy program}};

\draw [-{latex[scale=1.2]}, line width=1.5pt] (0,0) node [left] {0} -- (0,0) -- (12,0) node [below] {$\theta$};
\draw [-{latex[scale=1.2]}, line width=1.5pt] (0, -.5) -- (0,0) -- (0,6) node [above] {quantity, $q$};

\draw[line width=5pt,magenta] (1,0) -- (5,0);
\draw[line width=5pt,orange] (5,0) -- (11,0);

\draw[line width=1.5pt,] (1,0.15) -- (1,-0.15) node[below] {$\und\theta$};
\draw[line width=1.5pt,] (11.,0.15) -- (11.,-0.15) node[below] {$\BAR\theta$};
\draw[line width=1.5pt, color=black] (5,0.15) -- (5,-0.15) node[below] {$\theta^L$};

\end{tikzpicture}

\caption{$\a\geq\E[\omega]$}\label{fig:optimal_positive_allocation}
\end{subfigure}
\vspace{10pt}

\begin{subfigure}[t]{\textwidth}
\centering
\begin{tikzpicture}[scale=1.15]

\draw[scale=1, domain=1:11., densely dashed, variable=\x, line width=2.pt, color=black!50] plot ({\x}, {0.5*\x});
\draw (11.2,5.5) node[right] {$\color{black!50}q^{\LF}$};
\draw (9.5,5) node[above] {$\color{ceruleanblue}q^*$};

\draw[scale=1, domain=4:11., smooth, variable=\x, line width=3pt, color=ceruleanblue] plot ({\x}, {\x^3/1372+\x^2/686+279*\x/1372+726/343});
\draw[line width=1pt, dashed, color=black] (4,3) -- (4,0.15);

\draw[scale=1, domain=1:4., smooth, variable=\x, line width=3pt, color=ceruleanblue] plot ({\x}, {3});
\draw[line width=1pt, dashed, color=black] (1,3) -- (0.15,3);
\draw[line width=1.5pt, color=black] (0.15,3) -- (-0.15,3) node[left] {$q^*(\und\theta)$};

\draw (2.5,0.1) node [above] {{\color{darkspringgreen} public option}};
\draw (7.5,0.1) node [above] {{\color{orange} subsidy program}};

\draw [-{latex[scale=1.2]}, line width=1.5pt] (0,0) node [left] {0} -- (0,0) -- (12,0) node [below] {$\theta$};
\draw [-{latex[scale=1.2]}, line width=1.5pt] (0, -.5) -- (0,0) -- (0,6) node [above] {quantity, $q$};

\draw[line width=5pt,darkspringgreen] (1,0) -- (4,0);
\draw[line width=5pt,orange] (4,0) -- (11,0);

\draw[line width=1.5pt,] (1,0.15) -- (1,-0.15) node[below] {$\und\theta$};
\draw[line width=1.5pt,] (11.,0.15) -- (11.,-0.15) node[below] {$\BAR\theta$};

\draw[line width=1.5pt, color=black] (4,0.15) -- (4,-0.15) node[below] {};

\end{tikzpicture}

\caption{$\a<\E[\omega]$}\label{fig:optimal_positive_allocation_another}
\end{subfigure}

\caption{Structure of the optimal allocation function under positive correlation characterized by \Cref{thm:positive}.} 
\end{figure}
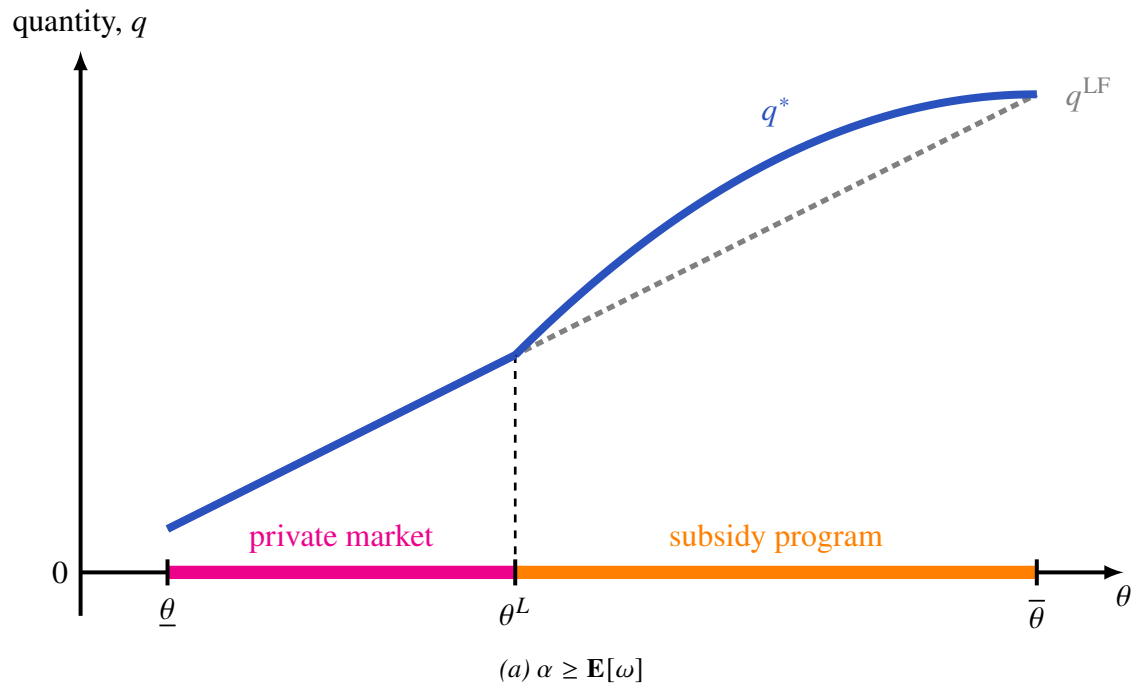
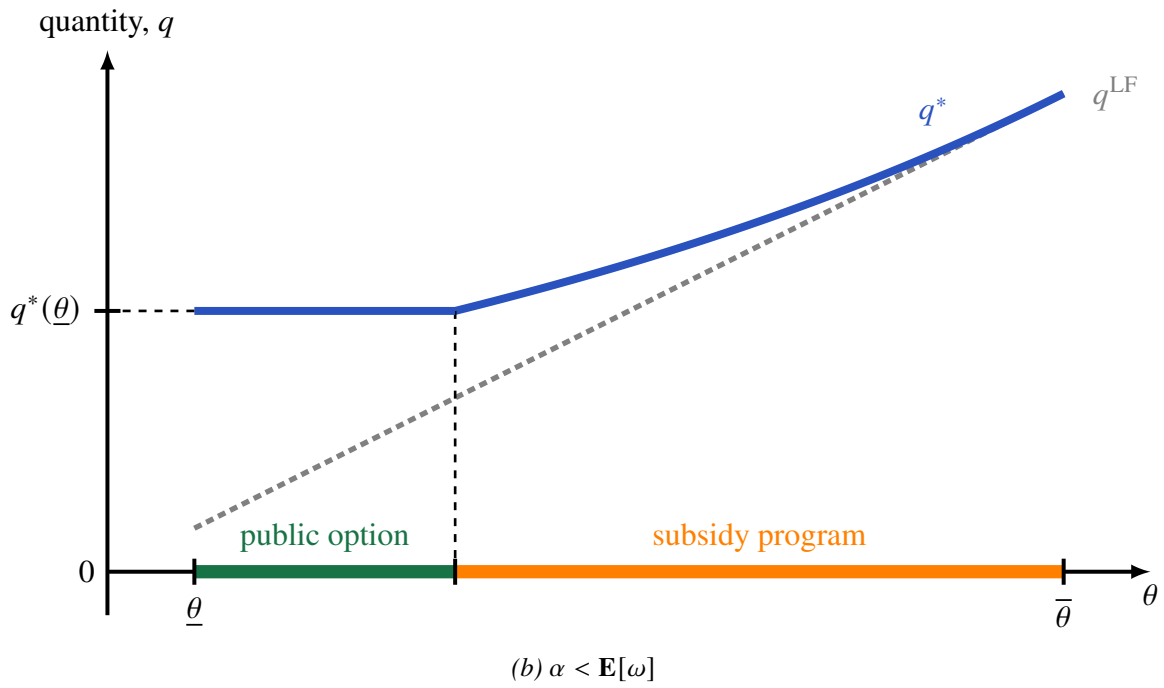

\clearpage
\subsection{Discussion}

The explicit characterization of the optimal allocation in \Cref{thm:positive} yields transparent comparative statics with respect to the welfare weight function $\omega$.  We consider two partial orders on welfare weight functions.  The first compares welfare weights pointwise and captures increases in redistributive priorities relative to the opportunity cost of public funds $\a$.  The second compares welfare weights through mean-preserving spreads and captures changes in how redistributive priorities are distributed across types while holding their average level fixed.

We begin with pointwise increases in welfare weights.

\begin{proposition}\label{prop:positive_shift}
Let $\omega_H$ and $\omega_L$ be welfare weight functions satisfying \Cref{ass:positive}, such that $\omega_H\geq\omega_L$ pointwise.  Then:
\begin{enumerate}[label={\em(\roman*)}]
\item the multiplier $\mu^*$ is weakly higher under $\omega_H$;\label{it:positive_shift_multiplier}
\item the topping-up cutoff type $\theta^L$ is weakly lower under $\omega_H$; and\label{it:positive_shift_cutoff}
\item the optimal allocation function $q^*$ is weakly higher under $\omega_H$.\label{it:positive_shift_allocation}
\end{enumerate}
\end{proposition}

\Cref{prop:positive_shift} shows that a pointwise increase in welfare weights weakly expands redistribution on all margins.  By \Cref{thm:positive}, the social planner serves weakly more types and weakly raises consumption for a larger set of consumers.  Moreover, sufficiently large increases in welfare weights can change the form of optimal redistribution: when $\E[\omega]$ rises above $\a$, the social planner begins to use a free public option.

Because only the welfare weights relative to $\a$ matter for the social planner's problem, \Cref{prop:positive_shift} can equivalently be interpreted as a comparative static in the opportunity cost of public funds.  Holding $\omega$ fixed, a decrease in $\a$ increases the relative value of redistribution in exactly the same way as a pointwise increase in welfare weights: more consumers are served, consumption rises for a larger set of types, and the social planner is more likely to provide a free public option.

Next, consider changes in the distribution of welfare weights across types while holding their average level fixed.  For increasing welfare weight functions with the same mean, we say that $\omega_H$ is a \emph{mean-preserving spread} of $\omega_L$ if
\begin{equation}\label{eq:positive_mps}
\int_{\theta}^{\BAR\theta}\omega_H(s)\,\dd F(s)\geq\int_{\theta}^{\BAR\theta}\omega_L(s)\,\dd F(s),\qquad\text{for all }\theta\in[\und\theta,\BAR\theta].
\end{equation}

\clearpage
\begin{proposition}\label{prop:positive_MPS}
Suppose $\omega_H$ and $\omega_L$ are welfare weight functions satisfying \Cref{ass:positive} such that $\omega_H$ is a mean-preserving spread of $\omega_L$, as defined in \cref{eq:positive_mps}.  Then:
\begin{enumerate}[label={\em(\roman*)}]
\item the multiplier $\mu^*$ is unchanged;\label{it:positive_MPS_multiplier}
\item the topping-up cutoff type $\theta^L$ is weakly lower for $\omega_H$; and\label{it:positive_MPS_cutoff}
\item the optimal allocation function $q^*$ is weakly higher under $\omega_H$.\label{it:positive_MPS_allocation}
\end{enumerate}
\end{proposition}

\Cref{prop:positive_MPS} shows that a mean-preserving spread of welfare weights, like a pointwise increase, also weakly expands redistribution. By shifting welfare weight toward consumers with higher types, a mean-preserving spread aligns redistributive priorities more closely with market demand.  By \Cref{thm:positive}, the social planner serves weakly more types and weakly raises consumption for a larger set of consumers.  Unlike \Cref{prop:positive_shift}, however, a mean-preserving spread does not affect the form of optimal redistribution: since $\mu^*$ is unchanged, this comparative static does not affect whether a free public option is used.

Taken together, \Cref{prop:positive_shift,prop:positive_MPS} identify two distinct channels through which welfare weights shape optimal redistribution.  Specifically, both comparative statics expand redistribution, but they do so for different reasons: pointwise increases in welfare weights raise the social value of redistribution overall, whereas mean-preserving spreads better align redistributive priorities with market demand.

More broadly, our results in this section complement the classic self-targeting logic in the literature.  That literature emphasizes that redistributive subsidies work best when intended beneficiaries are also the consumers most likely to purchase the subsidized good.  In the language of the IMF's \emph{Public Expenditure Handbook} \citep{mackenzie91}, ``[m]arketed goods with a negative income elasticity (i.e., inferior goods) are ideal candidates for a redistributive subsidy.''  Our results show that the same logic also governs how much topping up constrains optimal redistribution.  When welfare weight is positively correlated with demand, the market helps rather than hinders redistribution, so the screening loss from allowing topping up disappears. In this sense, the familiar principle that inferior goods are attractive targets for redistribution has a further implication here: such goods are not only easier to target, but also more compatible with topping up.

\section{Optimal Redistribution: Negative Correlation}
\label{sec:negative}

Next, we turn to the case in which welfare weights are negatively correlated with demand.

\begin{assumption}[negative correlation]\label{ass:negative}
The welfare weight function $\omega:[\und\theta,\BAR\theta]\to\R_+$ is decreasing.
\end{assumption}

\Cref{ass:negative} captures environments in which the social planner's redistributive priorities run against consumer demand.  Housing is a natural example: the intended beneficiaries of redistribution are often those with lower demand, reflecting lower ability to pay.  More generally, many normal goods are likely to exhibit this pattern, as higher demand is associated with lower marginal value for money \citep{dworczaketal21}.

This case reverses the self-targeting logic from \Cref{sec:positive}.  As \Cref{thm:scope} shows, when topping up is allowed, the first-order benefits of intervention accrue to all types above the targeted type.  Under positive correlation, this is consistent with the social planner's redistributive priorities, since higher types have higher welfare weights.  Under negative correlation, however, the same force generates leakage: subsidies intended for low types also benefit higher types, weakening the social planner's ability to screen.

This difference is already apparent at the extensive margin.  Under \Cref{ass:negative}, the two sufficient statistics in \Cref{thm:scope} are
\[\max_{\hat\theta\in[\und\theta,\BAR\theta]}\E\!\left[\omega(\theta)\mid\theta\geq\hat\theta\right]=\E[\omega]< \omega(\und\theta) = \max_{\hat\theta\in[\und\theta,\BAR\theta]}\omega(\hat\theta).\]
It follows from \Cref{cor:scope} that topping up strictly reduces the set of environments in which the social planner intervenes whenever the opportunity cost of public funds satisfies $\E[\omega]\leq \a < \omega(\und\theta)$.

The next corollary shows that topping up can affect the intensive margin as well.

\begin{corollary}\label{cor:negative}
Under \Cref{ass:negative}, there exists $\a_{\min}\in\R_+$ satisfying $\min\omega\leq \a_{\min} < \E[\omega]$, such that the optimal mechanisms with and without topping up differ if and only if $\a_{\min}<\a<\max\omega$.
\end{corollary}

\subsection{Characterization of the Optimal Mechanisms}

We now characterize how the optimal mechanisms differ when $\a$ lies in the region identified by \Cref{cor:negative}.  

Unlike under positive correlation, the optimal mechanisms have an upper-cutoff structure.  Low types, who have high welfare weights, are served by the social planner, whereas sufficiently high types are left to the private market.

The cutoff differs across the two participation structures because topping up weakens the screening value of low quantities.  With topping up, high types can choose a subsidized allocation and supplement it in the private market.  Without topping up, choosing a subsidized allocation requires giving up private-market consumption.

We characterize the optimal mechanisms using cutoffs $\theta^H_{\TU}$ and $\theta^H$, which determine the highest type served by the social planner with and without topping up, respectively.  For the topping-up problem, using the generalized virtual valuation in \cref{eq:virtual}, define $\theta^H_{\TU}:\R_+\to[\und\theta,\BAR\theta]$ by
\begin{equation}\label{eq:negative_cutoff_TU}
\theta^H_{\TU}(\mu) \coloneq \inf\left\{\theta\in[\und\theta,\BAR\theta]:\BAR{J_\mu|_{[\und\theta,\theta]}}(\theta)\leq\theta\right\},
\end{equation}
with the convention that $\theta^H_{\TU}(\mu)=\BAR\theta$ if the set is empty.  For the no-topping-up problem, let
\begin{equation}\label{eq:mu_max}
\mu_{\max} \coloneq \max_{\theta\in[\und\theta,\BAR\theta]} \int_{\und\theta}^{\theta}\left[\omega(s)-\a\right]\,\dd F(s).
\end{equation}
For $\mu\in[0,\mu_{\max}]$, define
\begin{equation}\label{eq:negative_cutoff}
\theta^H(\mu)
\coloneq\max\left\{\theta\in[\und\theta,\BAR\theta]:\int_{\und\theta}^{\theta}\left[\omega(s)-\a\right]\,\dd F(s)\geq\mu\right\}.
\end{equation}
Next, define the optimal multipliers on the no-lump-sum-transfers constraint \eqref{eq:LS}.  For the topping-up problem, the optimal multiplier is the same as under positive correlation (see \Cref{thm:positive}):
\begin{equation}\label{eq:negative_multiplier_TU}
\mu^*_{\TU} \coloneq \(\E[\omega]-\a\)_+.
\end{equation}
For the no-topping-up problem, let
\begin{equation}\label{eq:negative_auxiliary}
\begin{dcases}
\begin{aligned}
q_\mu(\theta) &\coloneq D\!\(c, \BAR{\left.\(J_\mu+\frac{\mu+\a-\E[\omega]}{\a f}\)\right|_{[\und\theta,\theta^H(\mu)]}}(\theta)\),\\
U^H(\mu) &\coloneq \und\theta v(q_\mu(\und\theta)) + \int_{\und\theta}^{\theta^H(\mu)} v(q_\mu(s))\,\dd s.
\end{aligned}
\end{dcases}
\end{equation}
The optimal multiplier is then given by
\begin{equation}\label{eq:negative_multiplier}
\mu^* \coloneq \min\left\{\mu\in[\(\E[\omega]-\a\)_+,\mu_{\max}]:U^H(\mu) \geq U^{\LF}(\theta^H(\mu))\right\}.
\end{equation}

\begin{theorem}\label{thm:negative}
Under \Cref{ass:negative}:
\begin{enumerate}[label={\em(\roman*)}] 
\item the optimal allocation function with topping up is
\[q_{\TU}^*(\theta) = \begin{dcases}
q^{\LF}(\theta), &\text{for }\theta^H_{\TU}(\mu^*_{\TU})<\theta\leq\BAR\theta,\\
D\!\(c,\BAR{J_{\mu^*_{\TU}}|_{[\und\theta,\theta_{\TU}^H(\mu_{\TU}^*)]}}(\theta)\), &\text{for }\und\theta\leq \theta\leq \theta^H_{\TU}(\mu^*_{\TU});
\end{dcases}\]
\item the optimal allocation function without topping up is
\[q^*(\theta) = \begin{dcases}
q^{\LF}(\theta), &\text{for }\theta^H(\mu^*)<\theta\leq\BAR\theta,\\
q_{\mu^*}(\theta), &\text{for }\und\theta\leq \theta\leq \theta^H(\mu^*).
\end{dcases}\]
\end{enumerate}
\end{theorem}

\Cref{thm:negative} shows that each optimal mechanism is governed by the same two kinds of objects---a cutoff and a multiplier---as in \Cref{thm:optimal,thm:positive}, albeit in a more complex way.  As before, the cutoffs $\theta^H_{\TU}(\mu^*_{\TU})$ and $\theta^H(\mu^*)$ determine which consumers are served by the social planner, while the multipliers $\mu_{\TU}^*$ and $\mu^*$ determine whether a free public option is used.  Under negative correlation, however, these objects are now constructed differently:
\begin{enumerate}[label={\em(\roman*)}]
\item {\bf Topping up.} When topping up is allowed, these objects are determined sequentially: the multiplier is pinned down by $\(\E[\omega]-\a\)_+$ as in \Cref{thm:optimal,thm:positive}, and the cutoff is then determined by the ironed generalized virtual valuation corresponding to this multiplier. 
\item {\bf No topping up.} When topping up is not allowed, these objects are jointly determined: the multiplier acts as the threshold for the cumulative welfare-weight condition determining the cutoff, while the cutoff determines the marginal participation constraint that pins down the multiplier.
\end{enumerate}

We now explain the implications for optimal mechanisms.  Although \Cref{thm:negative} applies more broadly, we focus on $\a_{\min}<\a<\max\omega$ to highlight the effect of topping up, following \Cref{cor:negative}.

\begin{enumerate}[label={\em(\roman*)}]

\item If $\E[\omega]\leq\a<\max\omega$, topping up changes the extensive margin.  \Cref{fig:optimal_negative_allocation} illustrates this case.

With topping up, the social planner does not intervene.  In this case, $\mu^*_{\TU}=0$, so there is no free public option; moreover, since $\omega$ is decreasing,
\[\int_\theta^{\BAR\theta}\left[\omega(s)-\a\right]\,\dd F(s)\leq 0 \implies J_0(\theta)\leq \theta,\qquad\text{for all }\theta\in[\und\theta,\BAR\theta].\]
Thus, by \cref{eq:negative_cutoff_TU}, $\theta^H_{\TU}(0)=\und\theta$, so the social planner does not use a subsidy program either.  This is consistent with \Cref{thm:scope}.  

Without topping up, high types are left to the private market, while low types are served by the social planner.  The cutoff $\theta^H(\mu^*)$ and multiplier $\mu^*$ are jointly determined by \cref{eq:negative_cutoff,eq:negative_multiplier}; the mechanism includes a free public option if and only if $\mu^*>0$.  In the limit as $\a\to\max\omega$, $\mu_{\max}=0$ and $\theta^H(0)=\und\theta$ by \cref{eq:mu_max,eq:negative_cutoff}.  \Cref{eq:negative_multiplier} then implies $\mu^*=0$, so the social planner does not intervene, consistent with \Cref{thm:scope}.  

\item If $\a_{\min}<\a<\E[\omega]$, topping up changes the intensive margin.  \Cref{fig:optimal_negative_allocation_another} illustrates this case.

With topping up, high types are left to the private market, while low types are served by the social planner.  Since $\mu^*_{\TU}=\(\E[\omega]-\a\)_+>0$, the optimal mechanism always includes a free public option and serves consumers up to the cutoff $\theta^H_{\TU}(\mu^*_{\TU})$ determined by \cref{eq:negative_cutoff_TU}. 

Without topping up, low types are served by the social planner, up to the cutoff $\theta^H(\mu^*)$.  The cutoff is equal to $\BAR\theta$ if $\mu^*=\E[\omega]-\a$, in which case all consumers are served; however, the cutoff is interior if $\mu^*>\E[\omega]-\a$, in which case high types are left to the private market.  Since $\mu^*\geq\E[\omega]-\a>0$ by \cref{eq:negative_multiplier}, the optimal mechanism always includes a free public option.  
\end{enumerate}

\clearpage

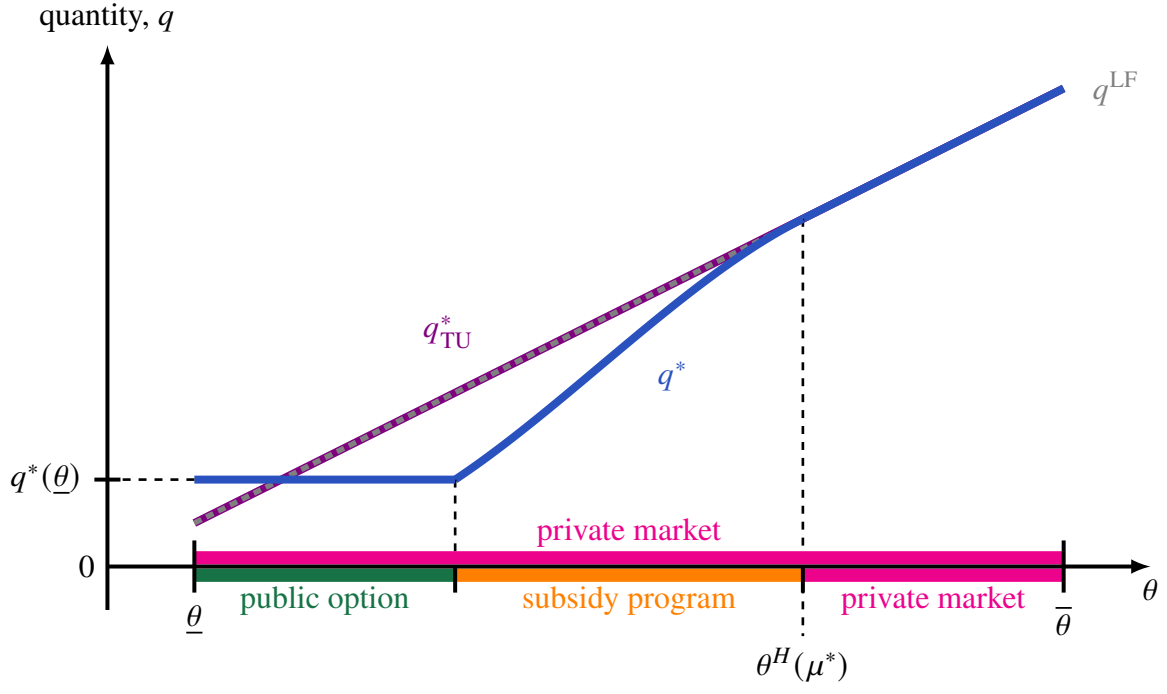
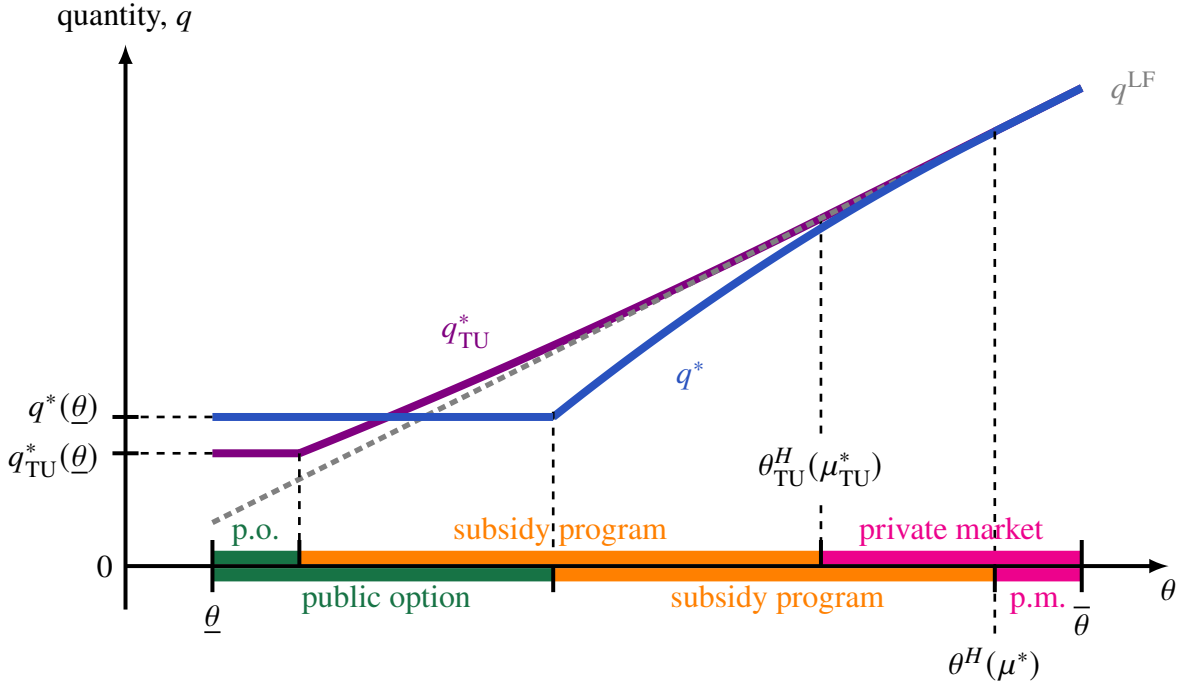
\begin{figure}[t!]
\centering

\begin{subfigure}[t]{\textwidth}
\centering
\begin{tikzpicture}[scale=1.15]

\draw[scale=1, domain=1:11., smooth, variable=\x, line width=3pt, color=violet] plot ({\x}, {0.5*\x});
\draw (3.5,2.7) node[right] {$\color{violet}q_{\TU}^*$};

\draw[scale=1, domain=1:11., densely dashed, variable=\x, line width=2.pt, color=black!50] plot ({\x}, {0.5*\x});
\draw (11.2,5.5) node[right] {$\color{black!50}q^{\LF}$};

\draw[scale=1, domain=8:11., smooth, variable=\x, line width=3pt, color=ceruleanblue] plot ({\x}, {0.5*\x});
\draw[line width=1pt, dashed, color=black] (8,4) -- (8,-0.8) node [below] {$\theta^H(\mu^*)$};
\draw (6.5,1.9) node[above] {$\color{ceruleanblue}q^*$};

\draw[scale=1, domain=4:8., smooth, variable=\x, line width=3pt, color=ceruleanblue] plot ({\x}, {-1*\x^3/48 + 17*\x^2/48 - 7*\x/6 + 4/3});
\draw[line width=1pt, dashed, color=black] (4,1) -- (4,0.15);

\draw[scale=1, domain=1:4., smooth, variable=\x, line width=3pt, color=ceruleanblue] plot ({\x}, {1});
\draw[line width=1pt, dashed, color=black] (1,1) -- (0.15,1);
\draw[line width=1.5pt, color=black] (0.15,1) -- (-0.15,1) node[left] {$q^*(\und\theta)$};

\draw (6,0.1) node [above] {{\color{magenta} private market}};
\draw (9.5,-0.1) node [below] {{\color{magenta} private market}};
\draw (2.5,-0.1) node [below] {{\color{darkspringgreen} public option}};
\draw (6,-0.1) node [below] {{\color{orange} subsidy program}};

\draw[line width=5pt,magenta] (1,0.1) -- (11,0.1);
\draw[line width=5pt,magenta] (8,-0.1) -- (11,-0.1);
\draw[line width=5pt,darkspringgreen] (1,-0.1) -- (4,-0.1);
\draw[line width=5pt,orange] (4,-0.1) -- (8,-0.1);

\draw [-{latex[scale=1.2]}, line width=1.5pt] (0,0) node [left] {0} -- (0,0) -- (12,0) node [below] {$\theta$};
\draw [-{latex[scale=1.2]}, line width=1.5pt] (0, -.5) -- (0,0) -- (0,6) node [above] {quantity, $q$};

\draw[line width=1.5pt,] (1,0.3) -- (1,-0.3) node[below] {$\und\theta$};
\draw[line width=1.5pt,] (11.,0.3) -- (11.,-0.3) node[below] {$\BAR\theta$};

\draw[line width=1.5pt, color=black] (8,0.) -- (8,-0.3);
\draw[line width=1.5pt, color=black] (4,0.) -- (4,-0.3);

\end{tikzpicture}

\caption{$\E[\omega]\leq\a<\max\omega$}\label{fig:optimal_negative_allocation}
\end{subfigure}
\vspace{10pt}

\begin{subfigure}[t]{\textwidth}
\centering
\begin{tikzpicture}[scale=1.15]

\draw[scale=1, domain=8:11., smooth, variable=\x, line width=3pt, color=violet] plot ({\x}, {0.5*\x});
\draw[scale=1, domain=2.:8., smooth, variable=\x, line width=3pt, color=violet] plot ({\x}, {\x^2/121 + 89*\x/242 + 64/121});
\draw[scale=1, domain=1:2., smooth, variable=\x, line width=3pt, color=violet] plot ({\x}, {1.2975});
\draw (3.5,2.7) node[right] {$\color{violet}q_{\TU}^*$};

\draw[line width=1pt, dashed, color=black] (1,1.2975) -- (0.15,1.2975);
\draw[line width=1.5pt, color=black] (0.15,1.2975) -- (-0.15,1.2975) node[left,yshift=-0.5ex] {$q^*_{\TU}(\und\theta)$};

\draw[line width=1pt, dashed, color=black] (8,4) -- (8,1.5) node [below] {$\theta^H_{\TU}(\mu_{\TU}^*)$};
\draw[line width=1pt, dashed, color=black] (8,0.7) -- (8,0.1);

\draw[scale=1, domain=1:11., densely dashed, variable=\x, line width=2.pt, color=black!50] plot ({\x}, {0.5*\x});
\draw (11.2,5.5) node[right] {$\color{black!50}q^{\LF}$};

\draw[scale=1, domain=4.92:10., smooth, variable=\x, line width=3pt, color=ceruleanblue] plot ({\x}, {0.0002324 * \x^3 - 0.0346457629115 * \x^2 + 1.123195258231 * \x - 2.999776291153});
\draw (6.5,1.9) node[above] {$\color{ceruleanblue}q^*$};

\draw[line width=1pt, dashed, color=black] (2,1.2975) -- (2,0.15);
\draw[line width=1pt, dashed, color=black] (4.92,1.7757) -- (4.92,0.15);

\draw[scale=1, domain=1:4.92, smooth, variable=\x, line width=3pt, color=ceruleanblue] plot ({\x}, {1.71537});
\draw[scale=1, domain=10:11, smooth, variable=\x, line width=3pt, color=ceruleanblue] plot ({\x}, {0.5*\x});
\draw[line width=1pt, dashed, color=black] (10,5) -- (10,-0.8) node [below] {$\theta^H(\mu^*)$};
\draw[line width=1pt, dashed, color=black] (1,1.71537) -- (0.15,1.71537);
\draw[line width=1.5pt, color=black] (0.15,1.71537) -- (-0.15,1.71537) node[left,yshift=0.5ex] {$q^*(\und\theta)$};

\draw (1.5,0.1) node [above] {{\color{darkspringgreen} p.o.}};
\draw (5,0.1) node [above] {{\color{orange} subsidy program}};
\draw (9.5,0.1) node [above] {{\color{magenta} private market}};
\draw (3,-0.1) node [below] {{\color{darkspringgreen} public option}};
\draw (7.5,-0.1) node [below] {{\color{orange} subsidy program}};
\draw (10.5,-0.1) node [below] {{\color{magenta} \vphantom{b}p.m.\vphantom{b}}};

\draw[line width=5pt,darkspringgreen] (1,0.1) -- (2,0.1);
\draw[line width=5pt,orange] (2,0.1) -- (8,0.1);
\draw[line width=5pt,magenta] (8,0.1) -- (11,0.1);

\draw[line width=5pt,darkspringgreen] (1,-0.1) -- (4.92,-0.1);
\draw[line width=5pt,orange] (4.92,-0.1) -- (10,-0.1);
\draw[line width=5pt,magenta] (10,-0.1) -- (11,-0.1);

\draw [-{latex[scale=1.2]}, line width=1.5pt] (0,0) node [left] {0} -- (0,0) -- (12,0) node [below] {$\theta$};
\draw [-{latex[scale=1.2]}, line width=1.5pt] (0, -.5) -- (0,0) -- (0,6) node [above] {quantity, $q$};

\draw[line width=1.5pt,] (1,0.3) -- (1,-0.3) node[below] {$\und\theta$};
\draw[line width=1.5pt,] (11.,0.3) -- (11.,-0.3) node[below] {$\BAR\theta$};

\draw[line width=1.5pt, color=black] (2,0.) -- (2,0.3);
\draw[line width=1.5pt, color=black] (8,0.) -- (8,0.3);
\draw[line width=1.5pt, color=black] (4.92,0.) -- (4.92,-0.3);
\draw[line width=1.5pt, color=black] (10,0.) -- (10,-0.3);

\end{tikzpicture}

\caption{$\a_{\min}<\a<\E[\omega]$}\label{fig:optimal_negative_allocation_another}
\end{subfigure}

\caption{Structure of the optimal allocation function under negative correlation characterized by \Cref{thm:negative}.} 
\end{figure}

\clearpage
\subsection{Discussion}

The explicit characterization of the optimal allocations in \Cref{thm:negative} allows us to compare the two participation structures directly.  We begin by comparing the levels of optimal redistribution with and without topping up.

\begin{proposition}\label{prop:comparison}
The multipliers $\mu_{\TU}^*,\mu^*$ and cutoffs $\theta_{\TU}^H(\mu_{\TU}^*),\theta^H(\mu^*)$ satisfy
\[\mu^*_{\TU} \leq \mu^*\AND \theta_{\TU}^H(\mu_{\TU}^*) \leq \theta^H(\mu^*).\]
Moreover, the optimal allocation function without topping up $q^*$ crosses the optimal allocation function with topping up $q_{\TU}^*$ at most once from above: that is, there exists $\hat\theta\in[\und\theta,\BAR\theta]$ such that 
\[\begin{dcases}
\begin{aligned}
q^*(\theta) \leq q_{\TU}^*(\theta),&\qquad\text{for }\theta>\hat\theta,\\
q^*(\theta) \geq q_{\TU}^*(\theta),&\qquad\text{for }\theta<\hat\theta.
\end{aligned}
\end{dcases}\]
\end{proposition}

As \Cref{prop:comparison} shows, topping up weakly reduces redistribution on two margins.  First, it weakly reduces the multiplier on the no-lump-sum-transfers constraint, and hence weakly reduces the scope for a free public option.  Second, it weakly lowers the highest type served by the social planner.  Thus, under negative correlation, topping up not only reduces the set of environments in which intervention is optimal, as shown in \Cref{thm:scope}; it also weakly reduces the level of redistribution conditional on intervention.

The reason is that topping up weakens the screening value of low quantities.  Without topping up, a high type who chooses a subsidized low quantity must give up his laissez-faire private-market consumption.  With topping up, the same high type can choose the subsidized allocation and then purchase additional units in the private market.  This additional deviation opportunity makes redistribution toward low types more costly.  The social planner responds by weakly reducing the scope for a free public option and serving weakly fewer types.

The same logic also explains how topping up changes the optimal allocation.  As \Cref{prop:comparison} shows, topping up results in a \emph{rotation} of the optimal allocation function in a single-crossing sense: relative to the optimal allocation without topping up, consumption is weakly lower for low types and weakly higher for high types.  Topping up therefore reduces the extent to which the social planner can tilt consumption toward low types and away from higher-demand consumers with lower redistributive priority---and hence the amount of surplus the social planner can redirect toward the bottom.

As in \Cref{sec:positive}, the explicit characterization of the optimal allocations in \Cref{thm:negative} also yields transparent comparative statics with respect to the welfare weight function $\omega$.  We again consider two partial orders on welfare weight functions: pointwise increases and mean-preserving spreads.  In contrast to the comparative statics of \Cref{sec:positive}, however, the effects of these changes on optimal redistribution now depend on whether topping up is allowed.

We begin with pointwise increases in welfare weights.

\begin{proposition}\label{prop:negative_shift}
Let $\omega_H$ and $\omega_L$ be welfare weight functions satisfying \Cref{ass:negative}, such that $\omega_H\geq\omega_L$ pointwise.  Then:
\begin{enumerate}[label={\em(\roman*)}]
\item the multipliers $\mu_{\TU}^*$ and $\mu^*$ are weakly higher under $\omega_H$;\label{it:negative_shift_multiplier}
\item the cutoffs $\theta_{\TU}^H(\mu_{\TU}^*)$ and $\theta^H(\mu^*)$ are weakly higher under $\omega_H$; and\label{it:negative_shift_cutoff}
\item the optimal allocation function $q_{\TU}^*$ is pointwise weakly higher under $\omega_H$, while there exists $\hat\theta\in[\und\theta,\BAR\theta]$ such that the optimal allocation function $q^*$ is pointwise weakly lower for $\theta>\hat\theta$ and pointwise weakly higher for $\theta<\hat\theta$ under $\omega_H$.\label{it:negative_shift_allocation}
\end{enumerate}
\end{proposition}

With topping up, \Cref{prop:negative_shift} gives the same qualitative comparative statics as in \Cref{prop:positive_shift}.  A pointwise increase in welfare weights expands redistribution by increasing the scope for a free public option, the mass of consumers served, and the allocations of all served consumers.  The only difference is the direction in which the served region expands: under positive correlation, service expands downward from the top; under negative correlation, service expands upward from the bottom.

Without topping up, the same pointwise increase in welfare weights has a different effect.  There is still an increased scope for a free public option and mass of consumers served, but the allocations of served consumers do not necessarily increase.  Instead, allocations rotate: low types receive weakly more, while high types receive weakly less.  This reflects the screening role of low quantities.  When the social planner places more weight on low types, she raises their consumption directly, but she may also reduce the consumption of higher served types in order to preserve incentives and redirect surplus toward the bottom.  This force is absent with topping up, because high types can supplement privately.

Next, consider mean-preserving spreads.  For decreasing welfare weight functions with the same mean, we say that $\omega_H$ is a \emph{mean-preserving spread} of $\omega_L$ if
\begin{equation}\label{eq:negative_mps}
\int_{\und\theta}^{\theta}\omega_H(s)\,\dd F(s)\geq\int_{\und\theta}^{\theta}\omega_L(s)\,\dd F(s),\qquad\text{for all }\theta\in[\und\theta,\BAR\theta].
\end{equation}
Thus, $\omega_H$ shifts welfare weight toward lower types while holding the average welfare weight fixed.

\begin{proposition}\label{prop:negative_MPS}
Suppose $\omega_H$ and $\omega_L$ are welfare weight functions satisfying \Cref{ass:negative} such that $\omega_H$ is a mean-preserving spread of $\omega_L$.  Then:
\begin{enumerate}[label={\em(\roman*)}]
\item the multiplier $\mu_{\TU}^*$ is unchanged while the multiplier $\mu^*$ is weakly higher for $\omega_H$;\label{it:negative_MPS_multiplier}
\item the cutoff $\theta_{\TU}^H(\mu_{\TU}^*)$ is weakly lower for $\omega_H$; and\label{it:negative_MPS_cutoff}
\item the allocation function $q_{\TU}^*$ is pointwise weakly lower for $\omega_H$.\label{it:negative_MPS_allocation}
\end{enumerate}
\end{proposition}

With topping up, \Cref{prop:negative_MPS} shows that a mean-preserving spread of welfare weights reduces redistribution.  This is the opposite of \Cref{prop:positive_MPS} and may be surprising: a mean-preserving spread increases dispersion in welfare weights, implying that consumer demand is more statistically informative about redistributive priority---a property that should improve screening, at least in principle \citep{dworczaketal21}.  Under positive correlation, this improves self-targeting, leading to increased redistribution.  However, under negative correlation, this generates greater leakage of subsidies to high types, who mimic the behavior of low types and supplement privately.  As \Cref{prop:negative_MPS} shows, this force dominates the increase in statistical informativeness and leads to an optimal reduction in redistribution.  

Without topping up, the effect of a mean-preserving spread is different.  As low quantities can now help screen, the increase in statistical informativeness leads to a nontrivial tradeoff with incentive constraints.  This tradeoff always favors weakly expanding the scope of the free public option, but leads to ambiguous comparative statics for the mass of consumers served and the allocations for served consumers.

Taken together, \Cref{prop:comparison,prop:negative_shift,prop:negative_MPS} show that topping up changes both the level of redistribution and its comparative statics.  Under negative correlation, topping up weakens screening, reduces the level of optimal redistribution, and rotates the optimal allocation away from consumers with the highest redistributive priority.  It also changes how the social planner responds to redistributive priorities: with topping up, increases in welfare weights have monotone allocation effects, while mean-preserving spreads reduce redistribution; without topping up, the social planner can still use low quantities to screen, so stronger redistributive priorities toward low types may expand the scope of a free public option while reducing allocations for higher served types.

\section{Concluding Remarks}
\label{sec:conclusion}

In an inspiring review article on redistribution, \cite{curriegahvari08} write that 
\begin{quote}
``[p]rograms with or without the possibility of topping up have different welfare properties.  Currently, there are no general results regarding the relative merits of the two. [\dots] Nor are there any general results on the characterization of optimal public provision policies, targeted or universal.''
\end{quote}

This paper provides such a comparison in a mechanism-design model of in-kind redistribution.  Our results show that topping up generally weakens the social planner's ability to screen and changes optimal redistribution in at least three ways---at the extensive margin, at the intensive margin, and through comparative statics with respect to changes in redistributive priorities.

At the extensive margin, we show that the social planner's optimal intervention decision under each participation structure is characterized by a sufficient statistic of the welfare weight distribution.  With topping up, the social planner optimally intervenes if and only if the maximal upper-tail average welfare weight exceeds the opportunity cost of public funds.  Without topping up, the relevant sufficient statistic is the maximal welfare weight.  This leads the social planner to optimally intervene in strictly fewer environments when topping up is allowed under negative correlation, but not positive correlation.

At the intensive margin, we show that the effect of topping up depends on the correlation between redistributive priorities and demand.  Under positive correlation, topping up does not change the optimal mechanism.  Under negative correlation, however, topping up weakens screening, reduces the level of optimal redistribution, and rotates consumption away from those with the highest redistributive priority.

The comparative statics of optimal redistribution with respect to changes in redistributive priorities reinforce this distinction.  Under positive correlation, stronger redistributive priorities and greater statistical informativeness of demand weakly expand optimal redistribution.  Under negative correlation, however, topping up undermines the social planner's ability to screen with low quantities: stronger redistributive priorities force a blunt expansion of consumption for all served types rather than a targeted rotation, and greater statistical informativeness of demand weakly \emph{reduces} the overall level of redistribution.  Thus, topping up changes not only the level of redistribution, but also how optimal redistribution responds to changes in redistributive priorities.

Our results can therefore be interpreted as characterizing the gross screening costs of topping up.  Under positive correlation, these costs are zero, providing a strong rationale for allowing consumers to top up.  Under negative correlation, by contrast, topping up generates screening costs at both the extensive and intensive margins.  A complete policy evaluation must therefore weigh these screening costs against the administrative and enforcement costs of preventing topping up.

From a methodological perspective, our analysis and results show that mechanism-design problems with type-dependent outside options can, somewhat surprisingly, remain tractable.  We solve two such problems explicitly.  With topping up, private-market access imposes a pointwise lower bound on consumption; without topping up, it generates participation constraints that can be written as majorization constraints.  Although these constraints are mathematically distinct, we solve both problems by adapting the approach of \cite{amadorbagwell13}: we characterize the relevant Lagrange multipliers directly from primitives and reduce the remaining optimization to generalized ironing \citep{myerson81,toikka11}.  The resulting tractability delivers not only explicit characterizations of the optimal mechanisms, but also global comparative statics with respect to pointwise changes and mean-preserving spreads in redistributive priorities.  Such comparative statics are often difficult to obtain in nonlinear mechanism-design problems because changes in primitives generally alter both the optimal allocation and the set of binding participation constraints.  We expect these methods to be useful in other mechanism-design environments with type-dependent outside options as well.

More broadly, the ability of consumers to access the private market is a prominent feature of many real-world redistribution programs, and---as we have shown in this paper---its inclusion in analyses affects the design of optimal mechanisms in substantial and realistic ways.  By developing a mechanism-design model of in-kind redistribution that incorporates these participation constraints, we have quantified when a social planner can strictly improve on the laissez-faire outcome and characterized the optimal redistribution program as a combination of a free public option, nonlinear subsidies, and the laissez-faire allocation.  Our results highlight that while private-market access limits the scope of in-kind redistribution, it also strengthens the case for non-market allocations, such as free public options.

Finally, while we have focused on redistribution as our application, our analysis is also useful for characterizing the Pareto frontier of optimal mechanisms \citep{dworczaketal21}.  In particular, although private market access is the source of participation constraints in our application, the same constraints can also be interpreted as a pointwise allocation lower bound or Pareto improvement requirement.  Such requirements may arise when a social planner with full control over the entire market designs a mechanism that benefits all consumers (in either allocation or utility space) relative to the laissez-faire outcome---an idea that naturally fits within political economy models of reform with consensus or majority voting, as noted by \cite{fuchsskrzypacz15}.   Recently, \cite{baronetal26} consider such a constraint in the reform of assignment schemes used to allocate Child Protective Services investigators.  \cite{dworczakmuir24} examine a related problem in the context of designing property rights.  Our approach in this paper also contributes to the set of tools for analyzing these models.

\clearpage

\bibliography{master_bibliography.bib}
\bibliographystyle{econ-econometrica}

\clearpage

\appendix

\section{Proofs of Main Results}
\label{app:proofs}

\subsection{Preliminary Analysis}\label{app:preliminary}

We begin by characterizing the \eqref{eq:IC} and \eqref{eq:LS} constraints in order to rewrite the social planner's problem in a more tractable form.  As these characterizations are well-known, we state them without proof.

\begin{claim}\label{clm:IC}
A mechanism $(q,t)$ satisfies \eqref{eq:IC} if and only if $q$ is nondecreasing and
\[\theta v(q(\theta)) - t(\theta) = \underbrace{\und\theta v(q(\und\theta)) - t(\und\theta)}_{=:\und U} + \int_{\und\theta}^\theta v(q(s))\,\dd s.\]
\end{claim}

\Cref{clm:IC} follows from standard arguments in the literature on mechanism design.  \citepos{myerson81} lemma implies that, given an allocation function $q$, there exists a payment function $t$ such that $(q,t)$ is incentive-compatible if and only if $q$ is nondecreasing.  The envelope theorem \citep{milgromsegal02} uniquely pins down the payment function up to an additive constant, $\und U$, equal to the utility that the lowest consumer type receives under the mechanism.

\begin{claim}\label{clm:LS}
An incentive-compatible mechanism $(q,t)$ satisfies \eqref{eq:LS} if and only if $\und U \leq \und\theta v(q(\und\theta))$.
\end{claim}

\Cref{clm:LS} shows that the no-lump-sum-transfers constraint is equivalent to imposing an upper bound on $\und U$ in terms of the quantity $q(\und\theta)$ allocated to the lowest type.  When \eqref{eq:LS} binds, $t(\und\theta)=0$, in which case the social planner allocates the good for free to the lowest type.

We now reformulate the social planner's problem in utility space rather than allocation space.  To this end, we make the change of variables $\nu\coloneq v\circ q$; we refer  to $\nu$ henceforth as the subutility function induced by the mechanism.  Since $v$ is increasing by assumption, \Cref{clm:IC} implies that any incentive-compatible mechanism induces a nondecreasing subutility function. 

Next, we apply \Cref{clm:IC} to rewrite the \eqref{eq:IR} constraints and the social planner's objective.  Let $\und U^{\LF}$ be the utility that the lowest consumer type receives under the laissez-faire mechanism, and let $\nu^{\LF}$ be the subutility function induced by the laissez-faire mechanism.  Then the envelope theorem (cf.~\Cref{clm:IC}) implies that the \eqref{eq:IR} constraints can be written as
\[\und U + \int_{\und\theta}^\theta \nu(s)\,\dd s \geq \und U^{\LF} + \int_{\und\theta}^{\theta} \nu^{\LF}(s)\,\dd s,\qquad\text{for all }\theta\in[\und\theta,\BAR\theta].\]
For notational convenience, extend the domain of $v^{-1}$ to $\R$ and its range to $\BAR\R\coloneq \R\cup\{+\infty\}$ by defining
\[\Psi(\hat\nu)\coloneq \begin{dcases}
v^{-1}(\hat\nu), &\text{if }\hat\nu\in[v(0),v(A)],\\
+\infty, &\text{otherwise}.
\end{dcases}\]
The envelope theorem also allows us to eliminate dependence on the payment function in the objective: 
\[\eqref{eq:OBJ} = \left[\E[\omega]-\a\right]\und U + \a\int_{\und\theta}^{\BAR\theta}\left[J_0(\theta)\nu(\theta) - c \Psi(\nu(\theta))\right]\,\dd F(\theta).\]
Here, we have used the generalized virtual valuation $J_\mu$ defined in \cref{eq:virtual} to simplify notation.

We summarize the above analysis by rewriting the social planner's problem.  Denoting the set of right-continuous, nondecreasing functions by $\calI\coloneq\left\{h:[\und\theta,\BAR\theta]\to\R\text{ is right-continuous and nondecreasing}\right\}$, \Cref{clm:IC,clm:LS} allow us to rewrite the social planner's problem as follows:
\begin{align*}
    \max_{\und U\in\R,\,\nu\in\calI}&\left\{\left[\E[\omega]-\a\right]\und U + \a\int_{\und\theta}^{\BAR\theta}\left[J_0(\theta)\nu(\theta)-c \Psi(\nu(\theta))\right]\,\dd F(\theta)\right\},\\
    \text{s.t. }&\begin{dcases}
        \und U \leq \und\theta\nu(\und\theta),\\
        \und U +\int_{\und\theta}^\theta \nu(s)\,\dd s \geq \und U^{\LF} + \int_{\und\theta}^{\theta} \nu^{\LF} (s)\,\dd s,\qquad\text{for all }\theta\in[\und\theta,\BAR\theta].
    \end{dcases}
\end{align*}
If topping up is allowed, we impose the additional constraint:
\[\nu(\theta)\geq \nu^{\LF}(\theta),\qquad\text{for all }\theta\in[\und\theta,\BAR\theta].\]

We conclude our preliminary analysis by making four observations about the social planner's problems with and without topping up that are clearer from this rewriting.

\begin{enumerate}
\item {\bf Existence of solution.}  To see that an optimal solution to each problem exists, we endow $\calI$ with the metric
\[d(\nu_1,\nu_2) = \norm{\nu_1-\nu_2}_{L^1} + \abs{\nu_1(\und\theta) - \nu_2(\und\theta)}.\] 
Observe that the social planner's objective is continuous in $(\und U,\nu)$.  Without loss of generality, we focus on the set $K\coloneq[\und U^{\LF},\und\theta v(A)]\times\left\{h\in\calI:[\und\theta,\BAR\theta]\to[v(0),v(A)]\right\}$.  By the Helly selection theorem and the dominated convergence theorem, $K$ is compact; hence the constraint set in each problem (as a closed subset of $K$) is also compact.  Finally, $(\und U^{\LF},\nu^{\LF})\in K$ satisfies the constraints; hence the constraint set is nonempty.
    
\item {\bf General uniqueness of solution.}  We now argue that the optimal solution to each problem is unique when $\E[\omega]\ne\a$.  Suppose that $(\und U_1,\nu_1)$ and $(\und U_2,\nu_2)$ are distinct optimal solutions with distinct allocation functions $\nu_1\ne\nu_2$, and consider $(\und U^*,\nu^*) = \((\und U_1 + \und U_2)/2,(\nu_1 + \nu_2)/2\)$.  Clearly, $\nu^*$ is nondecreasing, and $(\und U^*,\nu^*)$ satisfies the constraints (since the constraints are linear).  However, $v^{-1}$ is strictly convex since $v$ is increasing and strictly concave; because $\nu_1$ and $\nu_2$ are right-continuous, $\nu_1\ne\nu_2$ implies they differ on a set of positive measure, hence Jensen's inequality implies that the social planner's objective is strictly larger under $(\und U^*,\nu^*)$, a contradiction.  Thus, the optimal allocation function for each problem must be unique.  When the \eqref{eq:LS} constraint binds and/or when the \eqref{eq:IR} constraint binds for any type, $\und U^*$ is uniquely pinned down by the envelope theorem (cf.~\Cref{clm:IC}).  Consequently, the optimal solution can fail to be unique only if both the \eqref{eq:LS} and \eqref{eq:IR} constraints are slack, in which case $\E[\omega]=\a$.
    
\item {\bf Optimality of deterministic mechanisms.}  Next, we observe that the restriction to deterministic mechanisms entails no loss of generality.  For any stochastic mechanism, let $\nu^\dag(\theta)$ be the expected subutility allocated to type $\theta$, and let $\und U^\dag$ be the expected utility of the lowest type.  The constraints are preserved by linearity, and the deterministic mechanism $(\und U^\dag,\nu^\dag)$ delivers the same subutility.  Moreover, since $\Psi$ is convex, Jensen's inequality implies $\Psi(\nu^\dag(\theta))\leq \E_{\nu}\!\left[\Psi(\nu(\theta))\right]$, so the social planner's objective is weakly larger.
    
\item {\bf Convex program with majorization constraints.}  Finally, we point out that the problem without topping up can be written as a convex program with majorization constraints.  To illustrate, suppose that the \eqref{eq:IR} constraint binds for the highest type.  Then the \eqref{eq:IR} constraints can be rewritten as
\[\int_{\theta}^{\BAR\theta} \nu(s)\,\dd s \leq \int_{\theta}^{\BAR\theta} \nu^{\LF} (s)\,\dd s,\qquad\text{for all }\theta\in[\und\theta,\BAR\theta].\]
Since $\nu$ and $\nu^{\LF}$ are nondecreasing, this condition is equivalent to $\nu$ being weakly majorized by $\nu^{\LF}$.  This argument can be adapted for any type $\hat\theta$ whose \eqref{eq:IR} constraint binds by partitioning the type space into two intervals, $[\und\theta,\hat\theta]$ and $[\hat\theta,\BAR\theta]$: the problem on each interval is a convex program with majorization constraints.
    
Recent papers (\eg, \citealp{kleineretal21}; \citealp{akbarpouretal24}) solve linear programs with majorization constraints and show that their solutions coincide with extreme points of the constraint set.  However, extreme-point methods do not apply in our setting as our program has a strictly concave objective (due to diminishing marginal utility) over a convex feasible set; hence, the solution need not be an extreme point of the feasible set (\ie, it can be a convex combination of extreme points).  To overcome this technical challenge, we develop an alternate approach below.  
\end{enumerate}

\subsection{Proof of \texorpdfstring{\Cref{thm:scope}}{Theorem 1}}

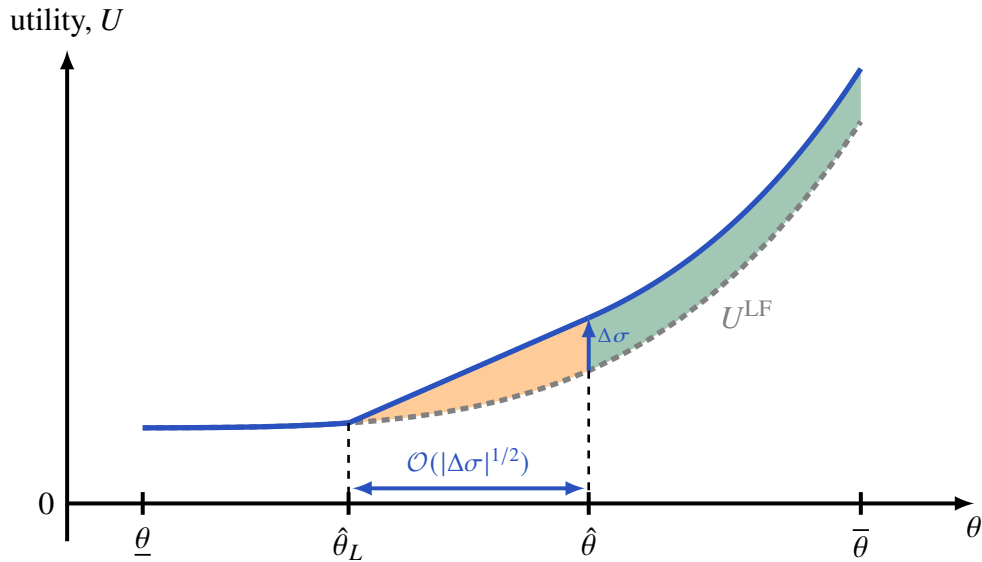
\begin{figure}[h!]
\centering
\begin{tikzpicture}[scale=1]

\begin{scope}
    \fill[darkspringgreen, opacity=0.4] (6.9,1.7555707) -- plot[domain=6.9:10.5, smooth] (\x, {\x^4/3000+1.7}) -- 
          plot[domain=10.5:6.9, smooth] (\x, {\x^4/3000+1}) -- cycle;
\end{scope}

\begin{scope}
    \fill[orange, opacity=0.4] (6.9,1.2555707) -- plot[domain=6.9:3.72309, smooth] (\x, {4*\x*6.9^3/3000-0.5667121}) -- 
          plot[domain=3.72309:6.9, smooth] (\x, {\x^4/3000+1}) -- cycle;
\end{scope}

\draw[scale=1, domain=1:10.5, densely dashed, variable=\x, line width=2.pt, color=black!50] plot ({\x}, {\x^4/3000+1});
\draw (8.5,2.5) node[right] {$\color{gray}U^{\LF}$};

\draw [-{latex[scale=1.2]}, line width=1.5pt, ceruleanblue] (6.9, 1.7555707) -- (6.9,2.1055707) node [xshift=0.8em,yshift=0.3em] {\scriptsize $\Delta\s$} -- (6.9, 2.4555707);

\draw[scale=1, domain=1:3.72309, smooth, variable=\x, line width=2.pt, color=ceruleanblue] plot ({\x}, {\x^4/3000+1});
\draw[scale=1, domain=3.72309:6.9, smooth, variable=\x, line width=2.pt, color=ceruleanblue] plot ({\x}, {4*\x*6.9^3/3000-0.5667121});
\draw[scale=1, domain=6.9:10.5, smooth, variable=\x, line width=2.pt, color=ceruleanblue] plot ({\x}, {\x^4/3000+1.7});

\draw[line width=1.pt, dashed, color=black] (3.72309,1.0640461288) -- (3.72309, 0.);
\draw[line width=1.pt, dashed, color=black] (6.9,1.7555707) -- (6.9, 0.);
\draw[{latex[scale=1.2]}-{latex[scale=1.2]}, line width=1.5pt, ceruleanblue] (3.77309,0.2) -- (5.311545,0.2) node [above] {\footnotesize $\calO(\abs{\Delta\s}^{1/2})$} -- (6.85,0.2);

\draw [-{latex[scale=1.2]}, line width=1.5pt] (0,0) node [left] {0} -- (0,0) -- (12,0) node [below] {$\theta$};
\draw [-{latex[scale=1.2]}, line width=1.5pt] (0, -.5) -- (0,0) -- (0,6) node [above] {utility, $U$};

\draw[line width=1.5pt,] (1,0.15) -- (1,-0.15) node[below] {$\und\theta$};
\draw[line width=1.5pt,] (3.72309,0.15) -- (3.72309,-0.15) node[below] {$\hat\theta_L$};
\draw[line width=1.5pt,] (6.9,0.15) -- (6.9,-0.15) node[below] {$\hat\theta$};
\draw[line width=1.5pt,] (10.5,0.15) -- (10.5,-0.15) node[below] {$\BAR\theta$};

\end{tikzpicture}
\caption{Proof sketch of \Cref{thm:scope} with topping up.}
\label{fig:scope_TU}
\end{figure}

\paragraph{Topping up.}  We begin by proving the result when topping up is allowed.  Consider the utility that each consumer receives, as shown in \Cref{fig:scope_TU}.  By the envelope theorem (cf.~\Cref{clm:IC}), the gradient of the utility curve must be nondecreasing in $\theta$ for any incentive-compatible mechanism; hence, incentive-compatible mechanisms must result in convex utility curves.  In addition, individually rational mechanisms must result in utility curves that lie weakly above $U^{\LF}$, shown as a gray dashed curve in \Cref{fig:scope_TU}.

Consider a small subsidy $\Delta\s>0$ that is targeted at consumers of type $\hat\theta$.  Because topping up is allowed, all consumers of higher type will mimic the behavior of type $\hat\theta$ by consuming the subsidy \emph{and} topping up their consumption in the private market.  This is shown by the blue utility curve in \Cref{fig:scope_TU}, which is a parallel upward shift of $U^{\LF}$ for consumer types in the interval $[\hat\theta,\BAR\theta]$.  Consumers with types just below $\hat\theta$ will also mimic the behavior of type $\hat\theta$; however, these types overconsume, and their utility is obtained by extending the blue utility curve linearly downwards until it intersects $U^{\LF}$ at type $\hat\theta_L$.

The welfare effect of this subsidy is the sum of two components, represented by the green and orange areas in \Cref{fig:scope_TU}.  In the green area, since the subsidy does not distort consumption, the net impact on weighted total surplus is exactly equal to $\Delta\s$, multiplied by the mass of consumers above $\hat\theta$ and the difference between the welfare weights of consumers and the opportunity cost of public funds:
\[\Delta\s\cdot\left[1-F(\hat\theta)\right]\left[\E\!\left[\omega(\theta)\mid\theta\geq\hat\theta\right]-\a\right].\]
In the orange area, the subsidy distorts each type's surplus by no more than $\calO(\Delta\s)$; however, the measure of affected types is bounded from above by $\calO(\abs{\Delta\s}^{1/2})$.  This is because $\hat\theta_L$ can be geometrically determined from \Cref{fig:scope_TU} and the envelope theorem:
\begin{align*}
\hat\theta_L v(q^{\LF}(\hat\theta_L)) - t^{\LF}(\hat\theta_L) 
&=\hat\theta v(q^{\LF}(\hat\theta)) - t^{\LF}(\hat\theta) + \Delta\s - \(\hat\theta-\hat\theta_L\)v(q^{\LF}(\hat\theta))\\
&=\hat\theta v(q^{\LF}(\hat\theta)) - t^{\LF}(\hat\theta) - \int_{\hat\theta_L}^{\hat\theta}v(q^{\LF}(\theta))\,\dd\theta.
\end{align*}
Equivalently, denoting $\nu^{\LF}\coloneq v\circ q^{\LF}$, Taylor's theorem implies that
\[\Delta\s=\int_{\hat\theta_L}^{\hat\theta}\left[\nu^{\LF}(\hat\theta)-\nu^{\LF}(\theta)\right]\,\dd\theta = \int_{\hat\theta_L}^{\hat\theta}\left[\(\hat\theta-\theta\)(\nu^{\LF})'(\hat\theta)+\calO\!\(\abs{\hat\theta-\theta}^2
\)\right]\,\dd\theta.\]
Since $v$ is strictly concave, observe that
\[(\nu^{\LF})'(\hat\theta)=v'(q^{\LF}(\hat\theta))\cdot(q^{\LF})'(\hat\theta) = \frac{c}{\hat\theta}\cdot \frac{-c}{\hat\theta^2 v''(q^{\LF}(\hat\theta))} = -\frac{c^2}{\hat\theta^3v''(q^{\LF}(\hat\theta))}>0.\]
Thus,
\[\Delta\s = \frac12\(\hat\theta-\hat\theta_L\)^2(\nu^{\LF})'(\hat\theta) + \calO\!\(\abs{\hat\theta-\hat\theta_L}^3\)\implies \hat\theta_L = \hat\theta - \calO\!\(\abs{\Delta\s}^{1/2}\).\]
Since $F$ is absolutely continuous, the mass of types in the interval $[\hat\theta_L,\hat\theta]$ is similarly bounded.  Hence, the net impact on weighted total surplus from the distorted types is bounded from above by $\calO(\abs{\Delta\s}^{3/2})$.  

In the limit as $\Delta\s\to0$, this small subsidy has a positive marginal impact if and only if the first-order term is positive:
\[\E\!\left[\omega(\theta)\mid\theta\geq \hat\theta\right]>\a.\]
Consequently, the optimal mechanism strictly improves on the laissez-faire outcome if there exists $\hat\theta$ satisfying this condition, that is,
\[\max_{\hat\theta\in[\und\theta,\BAR\theta]}\E\!\left[\omega(\theta)\mid\theta\geq\hat\theta\right]>\a.\]

We now prove the opposite direction.  Suppose that
\[\int_{\theta}^{\BAR\theta}\left[\omega(s)-\a\right]\,\dd F(s)\leq 0,\qquad\text{for all }\theta\in[\und\theta,\BAR\theta].\]
For any feasible mechanism $(\und U,\nu)$, let $\Delta\und U\coloneq \und U-\und U^{\LF}$ and $\Delta\nu\coloneq \nu-\nu^{\LF}$.  The \eqref{eq:TU} constraint gives $\Delta\nu\geq 0$, and the \eqref{eq:IR} constraint of the lowest type gives $\Delta\und U\geq 0$.  The change in weighted total surplus relative to laissez-faire is
\begin{align*}
&\left[\E[\omega]-\a\right]\Delta\und U + \int_{\und\theta}^{\BAR\theta}\int_\theta^{\BAR\theta}\left[\omega(s)-\a\right]\,\dd F(s)\cdot \Delta\nu(\theta)\,\dd\theta\\
&\qquad + \a\int_{\und\theta}^{\BAR\theta}\left[\theta\Delta\nu(\theta) - c\left[\Psi(\nu(\theta)) - \Psi(\nu^{\LF}(\theta))\right]\right]\,\dd F(\theta).
\end{align*}
The first two terms are weakly negative by assumption, while the last term is weakly negative because $\nu^{\LF}(\theta)$ maximizes $z\mapsto \theta z-c\Psi(z)$ pointwise.  Hence, laissez-faire is optimal.

\begin{figure}[ht!]
\centering
\begin{tikzpicture}[scale=1]

\begin{scope}
    \fill[darkspringgreen, opacity=0.4] (6.9,1.7555707) -- plot[domain=6.9:7.5, smooth] (\x, {\x^4/3000+1.7}) -- 
          plot[domain=7.5:6.9, smooth] (\x, {\x^4/3000+1}) -- cycle;
\end{scope}

\begin{scope}
    \fill[orange, opacity=0.4] (6.9,1.2555707) -- plot[domain=6.9:3.72309, smooth] (\x, {4*\x*6.9^3/3000-0.5667121}) -- 
          plot[domain=3.72309:6.9, smooth] (\x, {\x^4/3000+1}) -- cycle;
\end{scope}

\begin{scope}
    \fill[orange, opacity=0.4] (7.5,1.5546875) -- plot[domain=7.5:9.7619, smooth] (\x, {0.5625*\x-1.4640625}) -- 
          plot[domain=9.7619:7.5, smooth] (\x, {\x^4/3000+1}) -- cycle;
\end{scope}

\draw[scale=1, domain=1:10.5, densely dashed, variable=\x, line width=2.pt, color=black!50] plot ({\x}, {\x^4/3000+1});
\draw (8.5,2.5) node[right] {$\color{gray}U^{\LF}$};

\draw [-{latex[scale=1.2]}, line width=1.5pt, ceruleanblue] (6.9, 1.7555707) -- (6.9,2.1055707) node [xshift=0.8em,yshift=0.3em] {\scriptsize $\Delta\s$} -- (6.9, 2.4555707);
\draw[line width=1.pt, color=gray] (7.5, 2.0546875) -- (7.5, 2.7546875);

\draw[scale=1, domain=1:3.72309, smooth, variable=\x, line width=2.pt, color=ceruleanblue] plot ({\x}, {\x^4/3000+1});
\draw[scale=1, domain=3.72309:6.9, smooth, variable=\x, line width=2.pt, color=ceruleanblue] plot ({\x}, {4*\x*6.9^3/3000-0.5667121});
\draw[scale=1, domain=6.9:7.5, smooth, variable=\x, line width=2.pt, color=ceruleanblue] plot ({\x}, {\x^4/3000+1.7});
\draw[scale=1, domain=7.5:9.7619, smooth, variable=\x, line width=2.pt, color=ceruleanblue] plot ({\x}, {0.5625*\x-1.4640625});
\draw[scale=1, domain=9.7619:10.5, smooth, variable=\x, line width=2.pt, color=ceruleanblue] plot ({\x}, {\x^4/3000+1});

\draw[line width=1.pt, dashed, color=black] (3.72309,1.0640461288) -- (3.72309, 0.);
\draw[line width=1.pt, dashed, color=black] (6.9,1.7555707) -- (6.9, 0.);
\draw[line width=1.pt, dashed, color=black] (7.5,2.0546875) -- (7.5, 0.);
\draw[line width=1.pt, dashed, color=black] (9.7619,4.027026083) -- (9.7619, 0.);
\draw[{latex[scale=1.2]}-{latex[scale=1.2]}, line width=1.5pt, ceruleanblue] (3.77309,0.2) -- (5.311545,0.2) node [above] {\footnotesize $\calO(\abs{\Delta\s}^{1/2})$} -- (6.85,0.2);
\draw[{latex[scale=1.2]}-{latex[scale=1.2]}, line width=1.5pt, ceruleanblue] (7.55,0.2) -- (8.63095,0.2) node [above] {\footnotesize $\calO(\abs{\Delta\s}^{1/2})$} -- (9.7119,0.2);

\draw [-{latex[scale=1.2]}, line width=1.5pt] (0,0) node [left] {0} -- (0,0) -- (12,0) node [below] {$\theta$};
\draw [-{latex[scale=1.2]}, line width=1.5pt] (0, -.5) -- (0,0) -- (0,6) node [above] {utility, $U$};

\draw[line width=1.5pt,] (1,0.15) -- (1,-0.15) node[below] {$\und\theta$};
\draw[line width=1.5pt,] (3.72309,0.15) -- (3.72309,-0.15) node[below] {$\hat\theta_L$};
\draw[line width=1.5pt,] (6.9,0.15) -- (6.9,-0.15) node[below] {$\hat\theta$};
\draw[line width=1.5pt,] (7.5,0.15) -- (7.5,-0.15) node[below,xshift=0.7em] {$\hat\theta+\Delta\theta$};
\draw[line width=1.5pt,] (9.7619,0.15) -- (9.7619,-0.15) node[below] {$\hat\theta_H$};
\draw[line width=1.5pt,] (10.5,0.15) -- (10.5,-0.15) node[below] {$\BAR\theta$};

\end{tikzpicture}
\caption{Proof sketch of \Cref{thm:scope} without topping up.}
\label{fig:scope}
\end{figure}
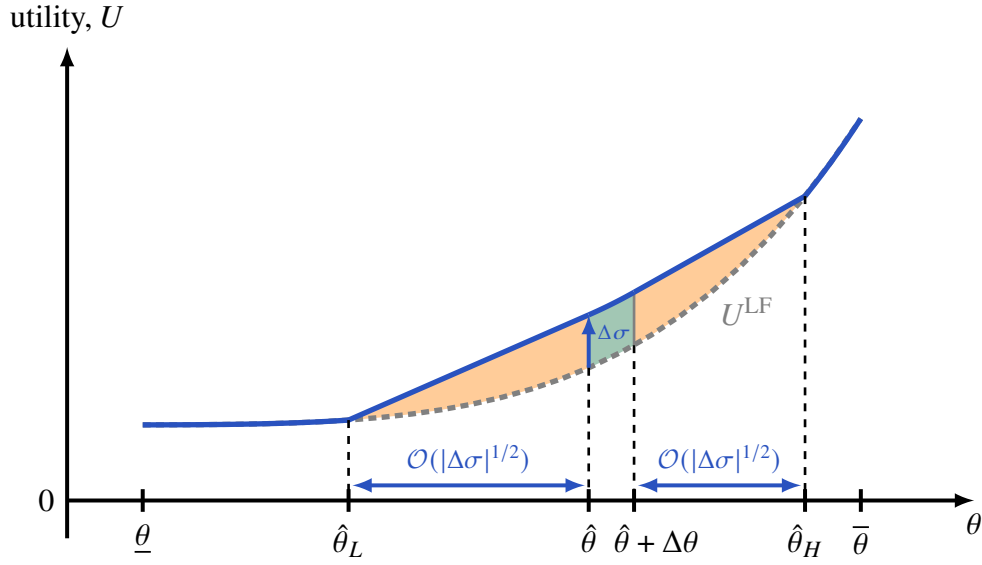

\paragraph{No topping up.} We now consider the case of no topping up.  As before, we consider a small subsidy $\Delta\s>0$ that is targeted at consumers of types $\theta\in[\hat\theta,\hat\theta+\Delta\theta]$, as shown in \Cref{fig:scope}.  The main difference is that, because topping up is not allowed, the consumption of types above $\hat\theta+\Delta\theta$ is in fact distorted: these types underconsume, and their utility is obtained by extending the blue utility curve linearly upwards until it intersects $U^{\LF}$ at type $\hat\theta_H$.  A similar argument as before yields
\[\hat\theta_H = \hat\theta + \Delta\theta + \calO\!\(\abs{\Delta\s}^{1/2}\).\]
Thus, the net impact on weighted total surplus due to distortions in consumption---given by the orange area in \Cref{fig:scope}---is bounded from above by $\calO(\abs{\Delta\s}^{3/2})$.  Since $\omega(\hat\theta)>\a$, by continuity there exists an interval $[\hat\theta,\hat\theta+\Delta\theta]$ such that
\[\int_{\hat\theta}^{\hat\theta+\Delta\theta}\left[\omega(\theta)-\a\right]\,\dd F(\theta)>0.\]
Fix such an interval and take $\Delta\s\to0$.  The first-order gain is of order $\Delta\s$, while the distortion terms are bounded from above by $\calO(\Delta\s^{3/2})$.  Hence, the perturbation is profitable for sufficiently small $\Delta\s$.

Conversely, suppose that $\omega(\theta)\leq \a$ for all $\theta\in[\und\theta,\BAR\theta]$.  Let $u(\theta)=\theta v(q(\theta))-t(\theta)$ be the utility delivered by any feasible mechanism.  The contribution of type $\theta$ to the objective can be written as
\[\a\left[\theta v(q(\theta))-cq(\theta)\right]+\left[\omega(\theta)-\a\right]u(\theta).\]
Subtracting the laissez-faire contribution gives
\[\a\left[\theta v(q(\theta))-cq(\theta)-U^{\LF}(\theta)\right]+\left[\omega(\theta)-\a\right]\left[u(\theta)-U^{\LF}(\theta)\right].\]
The first term is weakly negative by the definition of $U^{\LF}$.  The second term is weakly negative because $\omega(\theta)-\alpha\le0$ and $u(\theta)\ge U^{\LF}(\theta)$ by the \eqref{eq:IR} constraint.  Hence, laissez-faire is optimal.

\subsection{Proof of \texorpdfstring{\Cref{thm:optimal}}{Theorem 2}}

\paragraph{Sufficiency.}  Let $\mu^*=\(\E[\omega]-\a\)_+$, and suppose that
\[\begin{dcases}
\BAR{J_{\mu^*}|_{[\theta^L,\BAR\theta]}}(\theta)\geq\theta,&\text{for }\theta^L\leq \theta\leq\BAR\theta,\\
\omega(\theta)\leq \a,&\text{for }\und\theta\leq\theta<\theta^L.
\end{dcases}\]
We show that the optimal allocation function for both programs is given by
\[q^*(\theta) = q^*_{\TU}(\theta) = \begin{dcases}
D\!\(c,\BAR{J_{\mu^*}|_{[\theta^L,\BAR\theta]}}(\theta)\),&\text{for }\theta^L\leq \theta\leq\BAR\theta,\\
q^{\LF}(\theta),&\text{for }\und\theta\leq\theta<\theta^L.
\end{dcases}\]
To this end, we consider two cases:
\begin{itemize}
\item \textbf{\ul{Case \#1}:} $\theta^L=\und\theta$.

Consider the social planner's relaxed problem without the \eqref{eq:IR} and \eqref{eq:TU} constraints:
\[\max_{\und U\in\R,\,\nu\in\calI}\left\{\left[\E[\omega]-\a\right]\und U + \a\int_{\und\theta}^{\BAR\theta}\left[J_0(\theta)\nu(\theta)-c \Psi(\nu(\theta))\right]\,\dd F(\theta): \und U \leq \und\theta\nu(\und\theta)\right\}.\]
The solution to the relaxed problem is given by
\[\und U^*=\begin{dcases}
\und\theta \nu^*(\und\theta),&\text{if }\mu^*>0,\\
\und U^{\LF},&\text{if }\mu^*=0,
\end{dcases}\AND \nu^*(\theta) = v(q^*(\theta)) = v\!\(D\!\(c,\BAR{J_{\mu^*}}(\theta)\)\).\]
By assumption, $\BAR{J_{\mu^*}}(\theta)\geq\theta$ for all $\theta\in[\und\theta,\BAR\theta]$; hence, $q^*\geq q^{\LF}$, meaning that the \eqref{eq:TU} constraint is satisfied.  Moreover, $\und U^*\geq \und U^{\LF}$; together with $q^*\geq q^{\LF}$, this implies that the \eqref{eq:IR} constraint is also satisfied.  Thus, $(\und U^*,\nu^*)$ solves both programs, as claimed.

\item \textbf{\ul{Case \#2}:} $\theta^L>\und\theta$.

By definition of $\theta^L$, we infer that $\E[\omega]<\a$; hence, $\mu^*=\(\E[\omega]-\a\)_+=0$.  Define the mechanism $(\und U^*,\nu^*) = (\und U^{\LF}, v\circ q^*)$, and consider a feasible mechanism $(\und U,\nu)$ for the problem without topping up.  By the \eqref{eq:IR} constraint,
\[\Delta(\theta)\coloneq \und U - \und U^{\LF} + \int_{\und\theta}^\theta\left[\nu(s)-\nu^{\LF}(s)\right]\,\dd s\geq 0,\qquad\text{for all }\theta\in[\und\theta,\BAR\theta].\]
Since $\nu-\nu^{\LF}$ is bounded and measurable, $\Delta$ is absolutely continuous and $\Delta'(\theta)=\nu(\theta)-\nu^{\LF}(\theta)$ almost everywhere.  Moreover, the definition of $\theta^L$ implies that
\[\a\left[J_0(\theta) - \theta\right]f(\theta) = \int_\theta^{\BAR\theta}\left[\omega(s)-\a\right]\,\dd F(s) = \int_\theta^{\theta^L}\left[\omega(s)-\a\right]\,\dd F(s),\qquad\text{for all }\theta\in[\und\theta,\BAR\theta].\]
By assumption, $\omega(\theta)\leq\a$ for $\und\theta\leq\theta<\theta^L$, implying that
\begin{align}
0&\geq \int_{\und\theta}^{\theta^L}\left[\omega(\theta)-\a\right]\Delta(\theta)\,\dd F(\theta)\tag*{}\\
&=\left.-\Delta(\theta)\int_{\theta}^{\theta^L}\left[\omega(s)-\a\right]\,\dd F(s)\,\right|_{\und\theta}^{\theta^L}+ \a\int_{\und\theta}^{\theta^L}\left[J_0(\theta)-\theta\right]\Delta'(\theta)\,\dd F(\theta)\tag*{}\\
&=\left[\E[\omega]-\a\right]\(\und U-\und U^{\LF}\) + \a\int_{\und\theta}^{\theta^L}\left[J_0(\theta)-\theta\right]\left[\nu(\theta)-\nu^{\LF}(\theta)\right]\,\dd F(\theta).\label{eq:suff_1}
\end{align}
Since $\nu^{\LF}(\theta)$ maximizes the function $z\mapsto \theta z-c\Psi(z)$ for each $\theta\in[\und\theta,\BAR\theta]$, we also obtain
\begin{equation}\label{eq:suff_2}
\a\int_{\und\theta}^{\theta^L}\left[\theta\left[\nu(\theta)-\nu^{\LF}(\theta)\right]-c\left[\Psi(\nu(\theta))-\Psi(\nu^{\LF}(\theta))\right]\right]\,\dd F(\theta)\leq 0.
\end{equation}
Finally, let $\calM$ denote the set of measurable functions that map from $[\und\theta,\BAR\theta]$ into $[v(0),v(A)]$.  A standard ironing argument (\citealp{myerson81}; \citealp{toikka11}) yields
\begin{align}
\a\int_{\theta^L}^{\BAR\theta}\left[J_0(\theta)\nu(\theta)-c\Psi(\nu(\theta))\right]\,\dd F(\theta)
&\leq \a\max_{\hat\nu\in\calI}\int_{\theta^L}^{\BAR\theta}\left[J_0(\theta)\hat\nu(\theta)-c\Psi(\hat\nu(\theta))\right]\,\dd F(\theta)\tag*{}\\
&=\a\max_{\hat\nu\in\calM}\int_{\theta^L}^{\BAR\theta}\left[\BAR{J_0|_{[\theta^L,\BAR\theta]}}(\theta)\hat\nu(\theta)-c\Psi(\hat\nu(\theta))\right]\,\dd F(\theta)\tag*{}\\
&=\a\int_{\theta^L}^{\BAR\theta}\left[J_0(\theta)\nu^*(\theta)-c\Psi(\nu^*(\theta))\right]\,\dd F(\theta).\label{eq:suff_3}
\end{align}
Adding up inequalities \eqref{eq:suff_1}, \eqref{eq:suff_2}, and \eqref{eq:suff_3} yields
\[\left[\E[\omega]-\a\right]\(\und U-\und U^*\) + \a\int_{\und\theta}^{\BAR\theta}\left[J_0(\theta)\left[\nu(\theta)-\nu^*(\theta)\right]-c\left[\Psi(\nu(\theta))-\Psi(\nu^*(\theta))\right]\right]\,\dd F(\theta)\leq0.\]
Since $(\und U^*,\nu^*)$ satisfies the \eqref{eq:IR} constraint and the above inequality holds for any feasible mechanism $(\und U,\nu)$ for the problem without topping up, $(\und U^*,\nu^*)$ solves the problem without topping up.  Also, since $(\und U^*,\nu^*)$ satisfies the \eqref{eq:TU} constraint, it follows that $(\und U^*,\nu^*)$ solves the problem with topping up as well.  Thus, $(\und U^*,\nu^*)$ solves both programs, as claimed.
\end{itemize}

\paragraph{Necessity.}  Suppose that $q^*=q^*_{\TU}$.  Define the cutoff type $\kappa$ where $q^*$ potentially first deviates from $q^{\LF}$:
\[\kappa\coloneq\begin{dcases}
\BAR\theta,&\text{if }q^*=q^{\LF},\\
\inf\left\{\hat\theta\in[\und\theta,\BAR\theta]:q^*(\hat\theta)>q^{\LF}(\hat\theta)\right\},&\text{otherwise}.
\end{dcases}\]

\begin{lemma}\label{lem:nec_1}
$\omega(\theta)\leq \a$ for all $\und\theta\leq \theta<\kappa$.
\end{lemma}
\begin{proof}
By definition of $\kappa$, $q^*(\theta)=q^{\LF}(\theta)$ for all $\und\theta\leq\theta<\kappa$.  Suppose on the contrary that there exists $\theta_0<\kappa$ such that $\omega(\theta_0)>\a$; by continuity of $\omega$, there must be a neighborhood containing $\theta_0$ on which $\omega>\a$.  Our construction in the proof of \Cref{thm:scope} implies that there must be a strict improvement to $q^*=q^{\LF}$ on this neighborhood in the problem without topping up, a contradiction.
\end{proof}

\begin{lemma}\label{lem:nec_2}
$\theta^L\leq\kappa$.
\end{lemma}
\begin{proof}
If $\E[\omega]\geq \a$, then the claim holds trivially:
\[\int_{\und\theta}^{\BAR\theta}\left[\omega(s)-\a\right]\,\dd F(s) = \E[\omega]-\a\geq 0\implies \theta^L=\und\theta\leq \kappa.\]
It remains only to consider the case $\E[\omega]<\a$, which implies that $\mu^*=\(\E[\omega]-\a\)_+=0$ and $\und U^*=\und U^{\LF}$.  Moreover, the claim holds trivially if $\kappa=\BAR\theta$; thus, we also suppose that $\kappa<\BAR\theta$.

Consider the problem without topping up.  By the Lagrange multiplier theorem \citep[Theorem 1, p.~217]{luenberger69} there exists a nonincreasing function $\La:[\und\theta,\BAR\theta]\to\R$ such that $\La(\BAR\theta)=0$ and 
\[q^*(\theta) =  D\!\(c,\BAR{\(J_0+\frac{\La}{\a f}\)}(\theta)\),\]
where $\La$ is the cumulative multiplier associated with the \eqref{eq:IR} constraint and satisfies complementary slackness.  Define the slack in the \eqref{eq:IR} constraint by
\[\Delta(\theta) \coloneq \und U^*-\und U^{\LF} +\int_{\und\theta}^{\theta}\left[v(q^*(s))-v\!\left(q^{\LF}(s)\right)\right]\,\dd s = \int_{\kappa}^{\theta}\left[v(q^*(s))-v\!\left(q^{\LF}(s)\right)\right]\,\dd s,\]
where the second equality uses $\und U^*=\und U^{\LF}$.  Since $q^*=q^*_{\TU}\ge q^{\LF}$ pointwise, $\Delta(\theta)\ge 0$ for all $\theta\in[\und\theta,\BAR\theta]$.  Moreover, $\Delta(\theta)=0$ for all $\theta\le \kappa$ because $q^*(\theta)=q^{\LF}(\theta)$ for all $\theta<\kappa$, by definition of $\kappa$.

We now show that $\Delta(\theta)>0$ for all $\theta>\kappa$.  To this end, fix $\theta>\kappa$.  By definition of $\kappa$, there exists some $\hat\theta\in[\kappa,\theta)$ such that $q^*(\hat\theta)>q^{\LF}(\hat\theta)$.  Since $q^*$ is nondecreasing and $q^{\LF}$ is continuous and strictly increasing, the inequality $q^*(s)>q^{\LF}(s)$ holds on a nondegenerate open subinterval of $[\kappa,\theta)$.  It follows that
\[\Delta(\theta) = \int_{\kappa}^{\theta} \left[v(q^*(s))-v\!\left(q^{\LF}(s)\right)\right]\,\dd s >0.\]
This implies that the \eqref{eq:IR} constraint is slack for every $\theta>\kappa$.  By complementary slackness, it follows that $\La(\theta)=0$ for all $\theta>\kappa$.  

In addition, ironing cannot cross $\kappa$.  Suppose an ironing interval $I$ intersects both $[\und\theta,\kappa)$ and $(\kappa,\BAR\theta]$.  On $I\cap[\und\theta,\kappa)$, the allocation is $q^{\LF}$, which is strictly increasing.  But the ironed function---and hence the induced allocation---is constant on $I$, a contradiction.

Hence, the restriction of $J_0+\La/(\a f)$ to $[\kappa,\BAR\theta]$ coincides with the ironing of $J_0$ on that interval:
\[q^*(\theta) = D\!\(c,\BAR{J_0|_{[\kappa,\BAR\theta]}}(\theta)\),\qquad \text{for all }\theta\in[\kappa,\BAR\theta].\]
In particular, this holds at $\kappa$.  Thus, we have
\begin{align*}
D\!\(c,\BAR{J_0|_{[\kappa,\BAR\theta]}}(\kappa)\) = q^*(\kappa)\geq q^{\LF}(\kappa)=D(c,\kappa)
&\iff \BAR{J_0|_{[\kappa,\BAR\theta]}}(\kappa)\geq \kappa\\
&\iff \inf_{\hat\theta\in(\kappa,\BAR\theta]} \frac{\int_{\kappa}^{\hat\theta}J_0(s)\,\dd F(s)}{F(\hat\theta)-F(\kappa)}\geq \kappa.
\end{align*}
Since $J_0$ is continuous, the latter inequality implies that $J_0(\kappa)\geq\kappa$.  Rearranging yields
\[\int_{\kappa}^{\BAR\theta}\left[\omega(s)-\a\right]\,\dd F(s)\geq0 \iff \kappa\in\left\{\hat\theta\in[\und\theta,\BAR\theta]:\int_{\hat\theta}^{\BAR\theta}\left[\omega(s)-\a\right]\,\dd F(s)\geq0\right\}.\]
Since $\theta^L$ is defined as the minimum of the set, we conclude that $\theta^L\leq \kappa$.
\end{proof}

Together, \Cref{lem:nec_1,lem:nec_2} imply that $\omega(\theta)\leq\a$ for all $\und\theta\leq\theta<\theta^L$, as claimed.  It remains to prove the first condition in \Cref{thm:optimal}.  Define
\[G(\theta)\coloneq \int_{\theta}^{\BAR\theta}\left[\omega(s)-\a\right]\,\dd F(s).\]
By \Cref{lem:nec_1} and the continuity of $\omega$, $\omega(\theta)\leq\a$ for all $\theta\in[\und\theta,\kappa]$.  Since $\theta^L\leq \kappa$ by \Cref{lem:nec_2}, this implies that $G$ is nondecreasing on $[\theta^L,\kappa]$.  By construction, $G(\theta^L)\geq 0$; hence, $G(\theta)\geq0$ for all $\theta\in[\theta^L,\kappa]$.  It follows that
\begin{equation}\label{eq:J_bound}
J_{\mu^*}(\theta) - \theta \geq \frac{G(\theta)}{\a f(\theta)}\geq 0\implies J_{\mu^*}(\theta)\geq\theta,\qquad\text{for all }\theta\in[\theta^L,\kappa].
\end{equation}

\begin{lemma}\label{lem:nec_3}
$q^*(\theta)=D\!\(c,\BAR{J_{\mu^*}|_{[\kappa,\BAR\theta]}}(\theta)\)$ for all $\kappa\leq \theta\leq \BAR\theta$.
\end{lemma}
\begin{proof}
As shown in the proof of \Cref{lem:nec_2}, the claim holds in the case of $\E[\omega]<\a$ and $\kappa<\BAR\theta$.  Moreover, the claim holds trivially if $\kappa=\BAR\theta$ (in which case $q^*=q^{\LF}$).  Thus, it remains only to consider the case $\E[\omega]\geq\a$, where $\theta^L=\und\theta$ and $\und U^*\geq \und U^{\LF}$.

As in the proof of \Cref{lem:nec_2}, consider the problem without topping up.  By the Lagrange multiplier theorem (cf.~\citealp[Theorem 1, p.~217]{luenberger69}) there exists a nonincreasing function $\La:[\und\theta,\BAR\theta]\to\R$ such that $\La(\BAR\theta)=0$ and 
\[q^*(\theta) =  D\!\(c,\BAR{\(J_{\mu^*}+\frac{\La}{\a f}\)}(\theta)\),\]
where $\La$ is the cumulative multiplier associated with the \eqref{eq:IR} constraint and satisfies complementary slackness.  Define the slack in the \eqref{eq:IR} constraint by
\[\Delta(\theta) \coloneq \und U^*-\und U^{\LF} +\int_{\und\theta}^{\theta}\left[v(q^*(s))-v\!\left(q^{\LF}(s)\right)\right]\,\dd s \geq \int_{\und\theta}^{\theta}\left[v(q^*(s))-v\!\left(q^{\LF}(s)\right)\right]\,\dd s,\]
where the inequality uses $\und U^*\geq \und U^{\LF}$.  Since $q^*=q_{\TU}^*\geq q^{\LF}$ pointwise, we have $\Delta(\theta)\geq0$ for all $\theta$.  Moreover, $\Delta(\theta)>0$ for all $\theta>\kappa$.  This implies that the \eqref{eq:IR} constraint is slack for every $\theta>\kappa$.  By complementary slackness, it follows that $\La(\theta)=0$ for all $\theta>\kappa$.  Hence, the restriction of $\BAR{J_{\mu^*}+\La/\(\a f\)}$ to $[\kappa,\BAR\theta]$ coincides with $\BAR{J_{\mu^*}}$ on that interval:
\[q^*(\theta) = D\!\(c,\BAR{J_{\mu^*}|_{[\kappa,\BAR\theta]}}(\theta)\),\qquad \text{for all }\theta\in[\kappa,\BAR\theta].\]
This proves the desired result.
\end{proof}

\Cref{lem:nec_3} implies that $\BAR{J_{\mu^*}|_{[\kappa,\BAR\theta]}}(\theta)\geq \theta$ for all $\theta\in[\kappa,\BAR\theta]$.  We now show that $\BAR{J_{\mu^*}|_{[\theta^L,\BAR\theta]}}(\theta)\geq\theta$ for all $\theta\in[\theta^L,\BAR\theta]$.  Since this result holds trivially if $\theta^L=\BAR\theta$, suppose $\theta^L<\BAR\theta$.  Observe that:
\begin{enumerate}[label=\emph{(\roman*)}]
\item If $\kappa=\theta^L$, then our desired result holds trivially.
\item If $\kappa=\BAR\theta$, then our desired result is implied by \cref{eq:J_bound} and monotonicity of the ironing operator (\Cref{lem:ironing}).
\item If $\theta^L<\kappa<\BAR\theta$, then the min-max formula for ironing \citep{barlowetal72} implies that, for $\theta\in[\kappa,\BAR\theta]$,
\[\BAR{J_{\mu^*}|_{[\theta^L,\BAR\theta]}}(\theta) = \adjustlimits \inf_{z\in[\theta,\BAR\theta]}\sup_{y\in[\theta^L,\theta]}\frac{\int_{[y,z]}J_{\mu^*}\,\dd F}{F(z)-F(y)} \geq \adjustlimits \inf_{z\in[\theta,\BAR\theta]}\sup_{y\in[\kappa,\theta]}\frac{\int_{[y,z]}J_{\mu^*}\,\dd F}{F(z)-F(y)}=\BAR{J_{\mu^*}|_{[\kappa,\BAR\theta]}}(\theta)\geq\theta.\]
Moreover, for $\theta\in[\theta^L,\kappa]$ and $z\geq\theta$, choose $y=\theta$ in the max-min formula.  If $z\leq\kappa$, \cref{eq:J_bound} gives
\[\int_{\theta}^{z}J_{\mu^*}(s)\,\dd F(s)\geq\theta\left[F(z)-F(\theta)\right].\]
If $z>\kappa$, \cref{eq:J_bound} gives
\[\int_{\theta}^{\kappa}J_{\mu^*}(s)\,\dd F(s)\geq\theta\left[F(\kappa)-F(\theta)\right].\]
Moreover, letting $\cx H$ denote the convex minorant of any function $H$,
\begin{align*}
\int_{\kappa}^z J_{\mu^*}(s)\,\dd F(s)\geq \cx\(\left.\hat z\mapsto \int_\kappa^{\hat z}J_{\mu^*}(s)\,\dd F(s)\,\right|_{[\kappa,\BAR\theta]}\)(z) &= \int_{\kappa}^z\BAR{J_{\mu^*}|_{[\kappa,\BAR\theta]}}(s)\,\dd F(s)\\
&\geq \theta\left[F(z)-F(\kappa)\right].
\end{align*}
Thus,
\[\int_{\theta}^{z}J_{\mu^*}(s)\,\dd F(s)\geq\theta\left[F(\kappa)-F(\theta)\right] + \theta\left[F(z)-F(\kappa)\right] = \theta\left[F(z)-F(\theta)\right].\]
It follows that for $\theta\in[\theta^L,\kappa]$,
\[\BAR{J_{\mu^*}|_{[\theta^L,\BAR\theta]}}(\theta) = \adjustlimits \inf_{z\in[\theta,\BAR\theta]}\sup_{y\in[\theta^L,\theta]}\frac{\int_{[y,z]}J_{\mu^*}\,\dd F}{F(z)-F(y)} \geq \theta.\]
\end{enumerate}

As noted earlier, $\omega(\theta)\leq\a$ for all $\und\theta\leq\theta<\theta^L$.  From the above discussion, we conclude that
\[\BAR{J_{\mu^*}|_{[\theta^L,\BAR\theta]}}(\theta)\geq \theta,\qquad\text{for all }\theta\in[\theta^L,\BAR\theta].\]
 This completes the proof of \Cref{thm:optimal}.

\subsection{Proof of \texorpdfstring{\Cref{thm:positive}}{Theorem 3}}

By \Cref{cor:positive}, the optimal mechanisms with and without topping up coincide; hence, it remains to identify the common optimal allocation function.  The proof of the sufficiency direction of \Cref{thm:optimal} shows that, whenever the two conditions in \Cref{thm:optimal} hold, the common optimal allocation is given by
\[q^*(\theta)=\begin{dcases}
D\!\(c,\BAR{J_{\mu^*}|_{[\theta^L,\BAR\theta]}}(\theta)\),&\text{for }\theta^L\leq\theta\leq\BAR\theta,\\
q^{\LF}(\theta),&\text{for }\und\theta\leq\theta<\theta^L.\end{dcases}\]
Under \Cref{ass:positive}, \Cref{cor:positive} verifies both conditions in \Cref{thm:optimal}; hence the allocation constructed in the sufficiency proof of \Cref{thm:optimal} is optimal.  As the proof of \Cref{thm:optimal} shows, the multiplier $\mu^*=\(\E[\omega]-\a\)_+$ is the shadow cost of the \eqref{eq:LS} constraint.  If $\E[\omega]>\a$, then $\mu^*>0$, and the \eqref{eq:LS} constraint binds.  Equivalently, the lowest type pays zero.  Since $J_{\mu^*}$ has an atom at $\und\theta$, the ironed allocation is flat on an initial interval, which gives a free public option.  If $\E[\omega]\le\a$, then $\mu^*=0$, so the \eqref{eq:LS} constraint does not force such an initial pooling region.  This completes the proof of \Cref{thm:positive}.

\subsection{Proof of \texorpdfstring{\Cref{thm:negative}}{Theorem 4}}

\paragraph{Topping up.}  As in the proof of \Cref{thm:optimal}, together with the \eqref{eq:IR} constraint  $\und U\geq \und U^{\LF}$ of the lowest type, the \eqref{eq:TU} constraint implies all \eqref{eq:IR} constraints.  Define the candidate subutility function
\[\nu_{\TU}^*(\theta)\coloneq \begin{dcases}
\nu^{\LF}(\theta),&\text{for }\theta^H_{\TU}(\mu^*_{\TU})<\theta\leq\BAR\theta,\\
v\!\(D\!\(c,\BAR{J_{\mu_{\TU}^*}|_{[\und\theta,\theta^H_{\TU}(\mu_{\TU}^*)]}}(\theta)\)\), &\text{for }\und\theta\le\theta\le \theta^H_{\TU}(\mu^*_{\TU}).\end{dcases}\]
Also define
\[\und U_{\TU}^*\coloneq\begin{dcases}
\und U^{\LF},&\text{if }\a>\E[\omega],\\
\und\theta\nu_{\TU}^*(\und\theta),&\text{if }\a\leq\E[\omega].\end{dcases}\]

We begin by verifying feasibility.  To ensure that the candidate $\nu^*_{\TU}$ is nondecreasing, it suffices to check that it is nondecreasing at $\theta_{\TU}^H(\mu_{\TU}^*)$.  If $\theta_{\TU}^H(\mu_{\TU}^*)=\BAR\theta$, there is no jump to check.  If $\theta_{\TU}^H(\mu_{\TU}^*)=\und\theta$, the allocation matches laissez-faire almost everywhere; by our right-continuity convention, we evaluate the singleton at the lower endpoint as $\nu_{\TU}^*(\und\theta) = \nu^{\LF}(\und\theta)$, which trivially satisfies the \eqref{eq:TU} constraint exactly everywhere.  Hence, we proceed by assuming the interior case $\theta_{\TU}^H(\mu_{\TU}^*) \in(\und\theta, \BAR\theta)$.  To this end, observe that $\theta\mapsto \BAR{J_{\mu_{\TU}^*}|_{[\und\theta,\theta]}}(\theta) - \theta$ is continuous (following the min-max formula of ironing in the proof of \Cref{thm:optimal}).  Thus, since $\theta^H_{\TU}(\mu^*_{\TU})\in(\und\theta,\BAR\theta)$, the continuity of $\theta\mapsto \BAR{J_{\mu_{\TU}^*}|_{[\und\theta,\theta]}}(\theta)-\theta$ and \cref{eq:negative_cutoff_TU} imply that 
\[\BAR{J_{\mu_{\TU}^*}|_{[\und\theta,\theta^H_{\TU}(\mu_{\TU}^*)]}}(\theta^H_{\TU}(\mu_{\TU}^*)) = \theta^H_{\TU}(\mu_{\TU}^*).\]
By \Cref{ass:negative}, $\theta\mapsto\int_{\theta}^{\BAR\theta}\left[\omega(s)-\a\right]\,\dd F(s)$ is quasiconvex on $[\und\theta,\BAR\theta]$ and equal to zero at $\theta=\BAR\theta$; hence, $\theta\mapsto J_{\mu^*_{\TU}}(\theta)$ crosses $\theta\mapsto\theta$ at most once from above on $[\und\theta,\BAR\theta]$.  Thus, for any $\theta\in(\und\theta,\BAR\theta]$, $\hat\theta\mapsto \BAR{J_{\mu^*_{\TU}}|_{[\und\theta,\theta]}}(\hat\theta)$ crosses $\hat\theta\mapsto\hat\theta$ at most once from above on $[\und\theta,\theta]$ by \Cref{lem:single-crossing}.  By the definition of $\theta^H_{\TU}$ in \cref{eq:negative_cutoff_TU}, 
\[\BAR{J_{\mu^*_{\TU}}|_{[\und\theta,\theta]}}(\theta)> \theta,\qquad\text{for all }\theta < \theta_{\TU}^H(\mu^*_{\TU}).\]
It follows that
\[\BAR{J_{\mu^*_{\TU}}|_{[\und\theta,\theta]}}(\hat\theta)\geq\hat\theta,\qquad\text{for all }\hat\theta\leq\theta<\theta_{\TU}^H(\mu^*_{\TU}).\]
By the continuity of $\theta\mapsto\BAR{J_{\mu^*_{\TU}}|_{[\und\theta,\theta]}}(\hat\theta)$ (again following the min-max formula in the proof of \Cref{thm:optimal}), this implies
\[\BAR{J_{\mu^*_{\TU}}|_{[\und\theta,\theta_{\TU}^H(\mu^*_{\TU})]}}(\theta)\geq \theta,\qquad\text{for all }\theta \leq \theta_{\TU}^H(\mu_{\TU}^*),\] 
with equality at $\theta=\theta_{\TU}^H(\mu_{\TU}^*)$.  Thus, $\nu_{\TU}^*$ is nondecreasing at $\theta_{\TU}^H(\mu^*_{\TU})$ and $\nu_{\TU}^*\geq\nu^{\LF}$, so the \eqref{eq:TU} constraint holds.  With our choice of $\und U_{\TU}^*$, this implies that the \eqref{eq:LS} and \eqref{eq:IR} constraints hold.

Next, we verify optimality.  Let $\la_{\TU}^*\coloneq \(\a-\E[\omega]\)_+$ denote the multiplier on the \eqref{eq:IR} constraint $\und U\geq\und U^{\LF}$ of the lowest type.  In addition, define the multiplier on the \eqref{eq:TU} constraint by
\[\rho_{\TU}^*(\theta)\coloneq \a f(\theta) \left[\theta-J_{\mu_{\TU}^*}(\theta)\right]\bone_{\left\{\theta>\theta^H_{\TU}(\mu_{\TU}^*)\right\}}.\]
We first show that $\rho_{\TU}^*(\theta)\geq 0$ for all $\theta\in[\und\theta,\BAR\theta]$.  If $\theta_{\TU}^H(\mu_{\TU}^*)=\BAR\theta$, $\rho_{\TU}^*$ has empty support and nonnegativity is vacuous.  If $\theta_{\TU}^H(\mu_{\TU}^*)=\und\theta$, we already established that $J_{\mu_{\TU}^*}(\theta)\leq\theta$.  Hence, we proceed by assuming $\theta_{\TU}^H(\mu_{\TU}^*)\in(\und\theta,\BAR\theta)$.  As argued above, continuity implies that $\BAR{J_{\mu^*_{\TU}}|_{[\und\theta,\theta^H_{\TU}(\mu^*_{\TU})]}}(\theta^H_{\TU}(\mu_{\TU}^*))=\theta^H_{\TU}(\mu_{\TU}^*)$.  By the standard min-max formula for ironing, the ironed function evaluated at the upper endpoint is weakly bounded below by the unironed function: $J_{\mu_{\TU}^*}(\theta_{\TU}^H(\mu_{\TU}^*))\leq \theta_{\TU}^H(\mu_{\TU}^*)$ in the signed measure order.  Because $\theta\mapsto J_{\mu^*_{\TU}}(\theta)$ crosses $\theta\mapsto\theta$ at most once from above on $[\und\theta,\BAR\theta]$, it follows that
\[J_{\mu^*_{\TU}}(\theta)\leq \theta,\qquad\text{for all }\theta\geq\theta_{\TU}^H(\mu_{\TU}^*).\]
Thus, $\rho_{\TU}^*(\theta)\geq 0$ for all $\theta\in[\und\theta,\BAR\theta]$, as claimed.  Up to terms independent of $(\und U,\nu)$, the Lagrangian can be written as
\[\calL_{\TU}(\und U,\nu) = \left[\E[\omega]-\a-\mu_{\TU}^*+\la_{\TU}^*\right]\und U + \a\int_{\und\theta}^{\BAR\theta}\left[\left[J_{\mu_{\TU}^*}(\theta) + \frac{\rho_{\TU}^*(\theta)}{\a f(\theta)}\right]\nu(\theta) - c\Psi(\nu(\theta))\right]\,\dd F(\theta).\]
By construction, the first term is zero.  For $\theta\leq \theta^H_{\TU}(\mu_{\TU}^*)$, the coefficient of $\nu(\theta)$ in the integral is $J_{\mu_{\TU}^*}(\theta)$; the candidate $\nu_{\TU}^*$ maximizes the Lagrangian over nondecreasing subutility functions on $[\und\theta,\theta_{\TU}^H(\mu^*_{\TU})]$.  For $\theta>\theta^H_{\TU}(\mu_{\TU}^*)$, the coefficient of $\nu(\theta)$ in the integral is $\theta$, so the pointwise maximizer of the integrand is exactly $\nu^{\LF}(\theta)$.  Thus, $(\und U_{\TU}^*,\nu_{\TU}^*)$ maximizes the Lagrangian.

Finally, we verify that the complementary slackness conditions hold.  The multiplier $\rho_{\TU}^*$ is supported only where $\nu_{\TU}^*=\nu^{\LF}$.  If $\E[\omega]>\a$, then $\mu_{\TU}^*>0$ and the \eqref{eq:LS} constraint binds by construction.  If $\E[\omega]<\a$, then $\la_{\TU}^*>0$ and the \eqref{eq:IR} constraint of the lowest type binds by construction.  If $\E[\omega]=\a$, then both multipliers are zero.  Consequently, by the \cite{luenberger69} sufficiency theorem (see \citealp[Theorem 1 of Appendix B]{amadorbagwell13}), $(\und U_{\TU}^*,\nu_{\TU}^*)$ solves the topping-up problem.

\paragraph{No topping up.}  We first show that $\mu^*$ is well-defined in \cref{eq:negative_multiplier}.  By \cref{eq:mu_max}, 
\[\mu_{\max}\geq \max\left\{0,\int_{\und\theta}^{\BAR\theta}\left[\omega(s)-\a\right]\,\dd F(s)\right\}=\(\E[\omega]-\a\)_+.\] 
Thus, the interval $[\(\E[\omega]-\a\)_+,\mu_{\max}]$ is nonempty.  By the definition of $\mu_{\max}$, for any $\theta\in[\und\theta,\BAR\theta]$,
\[J_{\mu_{\max}}(\theta) + \frac{\mu_{\max} + \a - \E[\omega]}{\a f(\theta)} = \theta + \frac{\mu_{\max}\und\theta\cdot\d_{\und\theta}(\theta) + \mu_{\max} - \int_{\und\theta}^\theta\left[\omega(s)-\a\right]\,\dd F(s)}{\a f(\theta)}\geq \theta.\]
Here, the inequality is interpreted in the signed measure order.  \Cref{lem:ironing} then implies that for any $\theta\in[\und\theta,\theta^H(\mu_{\max})]$,
\[\BAR{\left.\(J_{\mu_{\max}} + \frac{\mu_{\max}+\a-\E[\omega]}{\a f}\)\right|_{[\und\theta,\theta^H(\mu_{\max})]}}(\theta)\geq\theta\implies \nu_{\mu_{\max}}(\theta)\geq\nu^{\LF}(\theta).\]
As such, 
\begin{align*}
U^H(\mu_{\max})
&= \und\theta\nu_{\mu_{\max}}(\und\theta) + \int_{\und\theta}^{\theta^H(\mu_{\max})} \nu_{\mu_{\max}}(s)\,\dd s\\
&\geq \und\theta\nu^{\LF}(\und\theta) + \int_{\und\theta}^{\theta^H(\mu_{\max})}\nu^{\LF}(s)\,\dd s \geq U^{\LF}(\theta^H(\mu_{\max})).
\end{align*}
This shows that the set in \cref{eq:negative_multiplier} is nonempty.  Moreover, observe that the function $\mu\mapsto U^{\LF}(\theta^H(\mu))$ is continuous since $U^{\LF}$ is continuous and $\theta^H(\mu)$ is, by \cref{eq:negative_cutoff}, the continuous right inverse of the upper branch of the quasiconcave function $\theta\mapsto\int_{\und\theta}^\theta\left[\omega(s)-\a\right]\,\dd F(s)$.  The function $\mu\mapsto U^H(\mu)$ is also continuous since ironing is continuous with respect to the total variation norm; hence, the function $\mu\mapsto U^H(\mu)-U^{\LF}(\theta^H(\mu))$ is continuous.  Consequently, the minimum in \cref{eq:negative_multiplier} is attained and $\mu^*$ is well-defined.

For any $\mu\in[0,\mu_{\max}]$, denote $\nu_\mu(\theta)\coloneq v(q_\mu(\theta))$ and $\nu^*(\theta)\coloneq v(q^*(\theta))$, so that
\[\nu^*(\theta)\coloneq\begin{dcases}
\nu^{\LF}(\theta),&\text{for }\theta^H(\mu^*)<\theta\leq\BAR\theta,\\
\nu_{\mu^*}(\theta),&\text{for }\und\theta\leq\theta\leq\theta^H(\mu^*).
\end{dcases}\] 
We choose the utility of the lowest type by
\[\und U^*\coloneq \begin{dcases}
\und\theta\nu_{\mu^*}(\und\theta), &\text{if }\mu^*>0,\\
U^{\LF}(\theta^H(\mu^*))-\int_{\und\theta}^{\theta^H(\mu^*)}\nu_{\mu^*}(s)\,\dd s,&\text{if }\mu^*=0.\end{dcases}\]

We begin by verifying feasibility.  We make our argument in three steps, showing each in a separate lemma: \emph{(i)}~$\nu^*$ crosses $\nu^{\LF}$ at most once from above on $[\und\theta,\theta^H(\mu^*)]$; \emph{(ii)}~$\nu^*$ is nondecreasing; and \emph{(iii)}~$(\und U^*,\nu^*)$ satisfies the \eqref{eq:LS} and \eqref{eq:IR} constraints.

\begin{lemma}\label{lem:nu_SCP}
The subutility function $\nu_{\mu^*}$ crosses $\nu^{\LF}$ at most once from above on $[\und\theta,\theta^H(\mu^*)]$.
\end{lemma}
\begin{proof}
By \Cref{ass:negative}, the following function is quasiconcave:
\[G(\theta)\coloneq\int_{\und\theta}^\theta\left[\omega(s)-\a\right]\,\dd F(s).\]
It follows that $\theta\mapsto\mu^*-G(\theta)$ is quasiconvex; by the definition of $\theta^H(\mu^*)$ in \cref{eq:negative_cutoff}, $\theta\mapsto\mu^*-G(\theta)$ crosses 0 at most once from above on $[\und\theta,\theta^H(\mu^*)]$.  Hence, the following generalized function crosses $\theta$ at most once from above on $[\und\theta,\theta^H(\mu^*)]$ in the signed measure order:
\[\theta\mapsto J_{\mu^*}(\theta) + \frac{\mu^* + \a-\E[\omega]}{\a f(\theta)} = \theta + \frac{\mu^*\und\theta\cdot\d_{\und\theta}(\theta) + \mu^* - G(\theta)}{\a f(\theta)}.\]
\Cref{lem:single-crossing} implies that single-crossing is preserved under ironing, so the following function crosses $\theta$ at most once from above on $[\und\theta,\theta^H(\mu^*)]$:
\[\theta\mapsto \BAR{\left.\(J_{\mu^*} + \frac{\mu^*+\a-\E[\omega]}{\a f}\)\right|_{[\und\theta,\theta^H(\mu^*)]}}(\theta).\]
We conclude that $\nu_{\mu^*}$ crosses $\nu^{\LF}$ at most once from above on $[\und\theta,\theta^H(\mu^*)]$, as claimed.
\end{proof}

\begin{lemma}\label{lem:nu_monotonicity}
The subutility function $\nu^*$ is nondecreasing.
\end{lemma}
\begin{proof}
Since $\nu_{\mu^*}$ is nondecreasing on $[\und\theta,\theta^H(\mu^*)]$ and $\nu^{\LF}$ is nondecreasing on $[\theta^H(\mu^*),\BAR\theta]$, it remains to check that $\nu^*$ is nondecreasing at $\theta^H(\mu^*)$ whenever $\theta^H(\mu^*)\in(\und\theta,\BAR\theta)$.

Suppose $\mu^*=0$.  By the definition of $\theta^H(\mu^*)$ in \cref{eq:negative_cutoff}, $G(\theta)\geq 0$ for all $\theta\in[\und\theta,\theta^H(\mu^*)]$.  Thus, for any $\theta\in[\und\theta,\theta^H(\mu^*)]$,
\[J_{\mu^*}(\theta) + \frac{\mu^* + \a-\E[\omega]}{\a f(\theta)} = \theta + \frac{\mu^*\und\theta\cdot\d_{\und\theta}(\theta) + \mu^* - G(\theta)}{\a f(\theta)} = \theta - \frac{G(\theta)}{\a f(\theta)}\leq \theta.\]
Here, the inequality is again interpreted in the signed measure order.  \Cref{lem:ironing} then implies that for any $\theta\in[\und\theta,\theta^H(\mu^*)]$,
\[\BAR{\left.\(J_{\mu^*} + \frac{\mu^*+\a-\E[\omega]}{\a f}\)\right|_{[\und\theta,\theta^H(\mu^*)]}}(\theta)\leq \theta.\]
In particular, $\nu_0(\theta^H(\mu^*))\leq \nu^{\LF}(\theta^H(\mu^*))$.  Thus, $\nu^*$ is nondecreasing at $\theta^H(\mu^*)$ if $\mu^*=0$.

Now suppose $\mu^*>0$.  By the minimality of $\mu^*$ and continuity of $\mu\mapsto U^H(\mu)-U^{\LF}(\theta^H(\mu))$, and since $\theta^H(\mu^*)\in(\und\theta,\BAR\theta)$ by assumption, the inequality in \cref{eq:negative_multiplier} must in fact bind at $\mu^*$:
\[U^H(\mu^*)=U^{\LF}(\theta^H(\mu^*)) \implies \und\theta\nu_{\mu^*}(\und\theta) + \int_{\und\theta}^{\theta^H(\mu^*)}\nu_{\mu^*}(s)\,\dd s = U^{\LF}(\theta^H(\mu^*)).\]
(Otherwise, if the inequality were strict, it would continue holding after a small decrease in $\mu$, contradicting the minimality of $\mu^*$.)  To show that $\nu^*$ is nondecreasing at $\theta^H(\mu^*)$, suppose on the contrary that $\nu_{\mu^*}(\theta^H(\mu^*))>\nu^{\LF}(\theta^H(\mu^*))$.  Since $\nu_{\mu^*}$ crosses $\nu^{\LF}$ at most once from above on $[\und\theta,\theta^H(\mu^*)]$ by \Cref{lem:nu_SCP}, we deduce that
\[\nu_{\mu^*}(\theta)\geq\nu^{\LF}(\theta),\qquad\text{for all }\theta\in[\und\theta,\theta^H(\mu^*)].\]
Strict inequality must hold on a subset of positive measure.  Therefore,
\begin{align*}
U^H(\mu^*)=\und\theta \nu_{\mu^*}(\und\theta) + \int_{\und\theta}^{\theta^H(\mu^*)}\nu_{\mu^*}(s)\,\dd s 
&> \und\theta \nu^{\LF}(\und\theta) + \int_{\und\theta}^{\theta^H(\mu^*)}\nu^{\LF}(s)\,\dd s \\
&= U^{\LF}(\theta^H(\mu^*)) + cq^{\LF}(\und\theta) \geq U^{\LF}(\theta^H(\mu^*)).
\end{align*}
But this contradicts $U^H(\mu^*)=U^{\LF}(\theta^H(\mu^*))$.  Thus, $\nu^*$ is nondecreasing at $\theta^H(\mu^*)$ if $\mu^*>0$.
\end{proof}

\begin{lemma}\label{lem:verification}
$(\und U^*, \nu^*)$ satisfies the \eqref{eq:LS} and \eqref{eq:IR} constraints.
\end{lemma}
\begin{proof}
If $\mu^*>0$, then $\und U^*=\und\theta\nu^*(\und\theta)$ by construction.  If $\mu^*=0$, then \cref{eq:negative_multiplier} implies that
\[\und U^* =  U^{\LF}(\theta^H(\mu^*)) - \int_{\und\theta}^{\theta^H(\mu^*)}\nu^*(s)\,\dd s\leq \und\theta \nu^*(\und\theta).\]
Thus, $(\und U^*,\nu^*)$ satisfies the \eqref{eq:LS} constraint.  

It remains to show that $(\und U^*,\nu^*)$ satisfies the \eqref{eq:IR} constraint.  To this end, define the slack in the \eqref{eq:IR} constraint as
\[\Delta(\theta)\coloneq \und U^* - \und U^{\LF} + \int_{\und\theta}^{\theta}\left[\nu^*(s)-\nu^{\LF}(s)\right]\,\dd s.\]
There are two cases:
\begin{enumerate}[label=\emph{(\roman*)}]
\item Suppose $\theta\leq \theta^H(\mu^*)$.  By \Cref{lem:nu_SCP}, $\nu^*$ crosses $\nu^{\LF}$ at most once from above on $[\und\theta,\theta^H(\mu^*)]$; hence, $\Delta$ is quasiconcave on $[\und\theta,\theta^H(\mu^*)]$.  Consequently, $\Delta$ attains its minimum on $[\und\theta,\theta^H(\mu^*)]$ at an endpoint of that interval.

At the upper endpoint $\theta^H(\mu^*)$, 
\[\Delta(\theta^H(\mu^*)) = \und U^* - \und U^{\LF} + \int_{\und\theta}^{\theta^H(\mu^*)}\left[\nu^*(s)-\nu^{\LF}(s)\right]\,\dd s.\]
If $\mu^*=0$, $\Delta(\theta^H(\mu^*))=0$ by our definition of $\und U^*$.  If $\mu^*>0$, \cref{eq:negative_multiplier} implies
\begin{align*}
0\leq U^H(\mu^*)-U^{\LF}(\theta^H(\mu^*)) 
&= \und\theta\nu_{\mu^*}(\und\theta) - \und U^{\LF} + \int_{\und\theta}^{\theta^H(\mu^*)}\left[\nu_{\mu^*}(s)-\nu^{\LF}(s)\right]\,\dd s\\
&= \und U^* - \und U^{\LF} + \int_{\und\theta}^{\theta^H(\mu^*)}\left[\nu_{\mu^*}(s)-\nu^{\LF}(s)\right]\,\dd s = \Delta(\theta^H(\mu^*)).
\end{align*}
Thus, in either case, $\Delta(\theta^H(\mu^*))\geq 0$.

At the lower endpoint $\und\theta$,
\[\Delta(\und\theta) = \und U^*-\und U^{\LF}.\]
If $\mu^*=0$, then our proof of \Cref{lem:nu_monotonicity} shows that $\nu_{\mu^*}(\theta)\leq \nu^{\LF}(\theta)$ for all $\theta\in[\und\theta,\theta^H(\mu^*)]$.  Moreover, since $\Delta(\theta^H(\mu^*))=0$ as shown above, 
\[\Delta(\und\theta) = \Delta(\theta^H(\mu^*)) - \int_{\und\theta}^{\theta^H(\mu^*)}\left[\nu_{\mu^*}(s)-\nu^{\LF}(s)\right]\,\dd s\geq 0.\]
If $\mu^*>0$, then $\und U^*=\und\theta\nu_{\mu^*}(\und\theta)$ by construction.  If $\nu_{\mu^*}(\und\theta)\geq \nu^{\LF}(\und\theta)$, then 
\[\Delta(\und\theta) = \und\theta\left[\nu_{\mu^*}(\und\theta) - \nu^{\LF}(\und\theta)\right]+cq^{\LF}(\und\theta) \geq 0.\]
If $\nu_{\mu^*}(\und\theta)<\nu^{\LF}(\und\theta)$, then \Cref{lem:nu_SCP} implies that $\nu_{\mu^*}(\theta)\leq \nu^{\LF}(\theta)$ for all $\theta\in[\und\theta,\theta^H(\mu^*)]$.  Therefore, 
\[\Delta(\und\theta) = \Delta(\theta^H(\mu^*)) - \int_{\und\theta}^{\theta^H(\mu^*)}\left[\nu_{\mu^*}(s)-\nu^{\LF}(s)\right]\,\dd s\geq 0.\]
Thus, in any case, $\Delta(\und\theta)\geq0$.

We conclude that for any $\theta\in[\und\theta,\theta^H(\mu^*)]$, the \eqref{eq:IR} constraint holds:
\[\Delta(\theta)\geq \min\left\{\Delta(\und\theta),\Delta(\theta^H(\mu^*))\right\} \geq 0.\]

\item Suppose $\theta>\theta^H(\mu^*)$. Since $\nu^*(\theta)=\nu^{\LF}(\theta)$ by construction,
\[\Delta(\theta) = \Delta(\theta^H(\mu^*)) \geq 0.\]
Hence, the \eqref{eq:IR} constraint holds as well.\qedhere
\end{enumerate}
\end{proof}

Next, we verify optimality.  Define the cumulative multiplier function $\La^*$ for the \eqref{eq:IR} constraint by
\[\La^*(\theta)\coloneq\begin{dcases}
\E[\omega]-\a+\int_{\und\theta}^\theta\left[\a-\omega(s)\right]\,\dd F(s),&\text{if }\theta^H(\mu^*)<\theta\leq\BAR\theta,\\
\E[\omega]-\a-\mu^*, &\text{if }\und\theta\leq\theta\leq \theta^H(\mu^*).\end{dcases}\]
By construction, $\La^*$ is nondecreasing.  Up to terms independent of $(\und U,\nu)$, the Lagrangian can be written as
\[\mathcal L(\und U,\nu)=\left[\E[\omega]-\a-\mu^*-\La^*(\und\theta)\right]\und U + \a\int_{\und\theta}^{\BAR\theta} \left[\left[J_{\mu^*}(\theta)-\frac{\La^*(\theta)}{\a f(\theta)}\right]\nu(\theta) - c\Psi(\nu(\theta))\right]\,\dd F(\theta).\]
The first term is zero by construction.  For $\theta\leq\theta^H(\mu^*)$, 
\[J_{\mu^*}(\theta)-\frac{\La^*(\theta)}{\a f(\theta)}=J_{\mu^*}(\theta) + \frac{\mu^*+\a-\E[\omega]}{\a f(\theta)}.\]
Thus, the subutility function $\nu_{\mu^*}$ maximizes the Lagrangian over nondecreasing subutility functions on $[\und\theta,\theta^H(\mu^*)]$.  For $\theta>\theta^H(\mu^*)$, the coefficient of $\nu(\theta)$ in the integral is
\[J_{\mu^*}(\theta)-\frac{\La^*(\theta)}{\a f(\theta)}=\theta.\]
Thus, the pointwise maximizer of the integrand is $\nu^{\LF}(\theta)$.  Hence, $(\und U^*,\nu^*)$ maximizes the Lagrangian.

Finally, we verify that the complementary slackness conditions hold.  For the \eqref{eq:LS} constraint, if $\mu^*>0$, then $\und U^*=\und\theta\nu^*(\und\theta)$ by construction, so the \eqref{eq:LS} constraint binds.  For the \eqref{eq:IR} constraint, the measure $\dd\La^*$ is supported on $[\theta^H(\mu^*),\BAR\theta]$.  If $\theta^H(\mu^*)<\BAR\theta$, then $\Delta(\theta^H(\mu^*))=0$: if $\mu^*=0$, this follows from the definition of $\und U^*$; if $\mu^*>0$, this follows from the minimality of $\mu^*$ and the continuity argument above.
Since $\nu^*=\nu^{\LF}$ on $(\theta^H(\mu^*),\BAR\theta]$, the slack in the \eqref{eq:IR} constraint remains zero on all of $[\theta^H(\mu^*),\BAR\theta]$.  Hence,
\[\int_{\und\theta}^{\BAR\theta}\Delta(\theta)\,\dd\La^*(\theta) = 0.\]
If $\theta^H(\mu^*)=\BAR\theta$, then $\mu^*=\E[\omega]-\a$, so $\La^*$ is identically zero; hence, the complementary slackness condition for the \eqref{eq:IR} constraint holds.

Consequently, applying the \cite{luenberger69} sufficiency theorem (see \citealp[Theorem 1 of Appendix B]{amadorbagwell13}), we conclude that $(\und U^*,\nu^*)$ solves the no-topping-up problem.

\clearpage

\section{Online Appendix: Additional Proofs}
\label{app:online}

\subsection{Preliminary Analysis}

We begin by proving two useful lemmas.  The first shows that the ironing operator preserves monotonicity, while the second shows that the ironing operator preserves the single-crossing property.

\begin{lemma}\label{lem:ironing}
Let $I_0=[\theta_L,\theta_0]\subseteq I_1=[\theta_L,\theta_1]\subseteq[\und\theta,\BAR\theta]$ be closed intervals with the same lower endpoint, $\theta_L$.  Let $h_0$ and $h_1$ be generalized functions defined on $I_0$ and $I_1$ respectively, such that $h_1\leq h_0$ on $I_0$ in the signed measure order.  That is, letting $\eta_0$ and $\eta_1$ denote the finite signed measures represented by $h_0\,\dd F$ and $h_1\,\dd F$,
\[\(\eta_1-\eta_0\)(B)\leq 0,\qquad\text{for every Borel set }B\subseteq I_0.\]   
Then, for every $\theta\in I_0$, $\BAR{h_1|_{I_1}}(\theta)\leq \BAR{h_0|_{I_0}}(\theta)$.
\end{lemma}

\begin{proof}
By the min-max formula for ironing (see, for example, \citealp{barlowetal72}), for any generalized function $h$ and any interval $I=[\theta_L,\theta_H]$,
\[\BAR{h|_I}(\theta) = \adjustlimits \inf_{z\in[\theta,\theta_H]}\sup_{y\in[\theta_L,\theta]}\frac{\int_{[y,z]}h\,\dd F}{F(z)-F(y)}.\]
This implies that
\begin{align*}
\BAR{h_1|_{I_1}}(\theta) = \adjustlimits \inf_{z\in[\theta,\theta_1]}\sup_{y\in[\theta_L,\theta]}\frac{\int_{[y,z]}h_1\,\dd F}{F(z)-F(y)}
&\leq \adjustlimits \inf_{z\in[\theta,\theta_0]}\sup_{y\in[\theta_L,\theta]}\frac{\int_{[y,z]}h_1\,\dd F}{F(z)-F(y)}=\BAR{h_1|_{I_0}}(\theta)\\
&\leq \adjustlimits \inf_{z\in[\theta,\theta_0]}\sup_{y\in[\theta_L,\theta]}\frac{\int_{[y,z]}h_0\,\dd F}{F(z)-F(y)}=\BAR{h_0|_{I_0}}(\theta).\qedhere
\end{align*}
\end{proof}

\begin{remark}\label{rem:ironing}
Symmetrically, let $I_0=[\theta_0,\theta_H]\subseteq I_1=[\theta_1,\theta_H]\subseteq[\und\theta,\BAR\theta]$ be closed intervals with the same upper endpoint, $\theta_H$.  Let $h_0$ and $h_1$ be generalized functions defined on $I_0$ and $I_1$ respectively, such that $h_1\geq h_0$ on $I_0$ in the signed measure order.  Then a similar proof to that of \Cref{lem:ironing} shows that for every $\theta\in I_0$, $\BAR{h_1|_{I_1}}(\theta)\geq \BAR{h_0|_{I_0}}(\theta)$.
\end{remark}

\begin{lemma}\label{lem:single-crossing}
Let $I_0=[\theta_L,\theta_0]\subseteq I_1=[\theta_L,\theta_1]\subseteq[\und\theta,\BAR\theta]$ be closed intervals with the same lower endpoint, $\theta_L$.  Let $h_0$ and $h_1$ be generalized functions defined on $I_0$ and $I_1$ respectively.  Suppose there exists $\hat\theta\in I_0$ such that in the signed measure order,
\[\begin{dcases}
h_1\leq h_0,&\text{on } [\hat\theta,\theta_0],\\
h_1\geq h_0,&\text{on } [\theta_L,\hat\theta].
\end{dcases}\]
(We say that \emph{$h_1$ crosses $h_0$ at most once from above in $I_0$}.)  Then there exists $\kappa\in I_0$ such that
\[\begin{dcases}
\BAR{h_1|_{I_1}}(\theta)-\BAR{h_0|_{I_0}}(\theta)\leq 0,&\qquad\text{for all }\theta\in [\kappa,\theta_0],\\
\BAR{h_1|_{I_1}}(\theta)-\BAR{h_0|_{I_0}}(\theta)\geq 0,&\qquad\text{for all }\theta\in [\theta_L,\kappa].
\end{dcases}\]
\end{lemma}

\begin{proof}
Let $\eta_0$ and $\eta_1$ denote the signed measures on $I_0$ and $I_1$ represented by $h_0\,\dd F$ and $h_1\,\dd F$, respectively.  We use the following standard variational characterization of ironing: for each $i\in\{0,1\}$,
\[\int_{I_i}\(y-\BAR{h_i|_{I_i}}\)\,\dd\eta_i \leq \int_{I_i}\(y-\BAR{h_i|_{I_i}}\)\BAR{h_i|_{I_i}}\,\dd F,\qquad\text{for all nondecreasing functions }y:I_i\to\R.\]
Suppose on the contrary that $\BAR{h_1|_{I_1}}$ does not cross $\BAR{h_0|_{I_0}}$ at most once from above in $I_0$.  Then there exist nondegenerate intervals $I_-=(\theta_L^-,\theta_H^-)$ and $I_+=(\theta_L^+,\theta_H^+)$, where $\theta_L\leq \theta_L^-<\theta_H^-\leq \theta_L^+<\theta_H^+\leq \theta_0$, such that
\[\begin{dcases}
\BAR{h_1|_{I_1}}(\theta)-\BAR{h_0|_{I_0}}(\theta)> 0,&\qquad\text{for all }\theta\in I_+,\\
\BAR{h_1|_{I_1}}(\theta)-\BAR{h_0|_{I_0}}(\theta)< 0,&\qquad\text{for all }\theta\in I_-.
\end{dcases}\]
Let $I_-$ and $I_+$ be the maximal such intervals, in the sense that there do not exist open intervals $I_-'\supsetneq I_-$ and $I_+'\supsetneq I_+$ such that $\BAR{h_1|_{I_1}}(\theta)-\BAR{h_0|_{I_0}}(\theta)>0$ for all $\theta\in I_+'$, and $\BAR{h_1|_{I_1}}(\theta)-\BAR{h_0|_{I_0}}(\theta)< 0$ for all $\theta\in I_-'$. 

We first show that the signed measure $\eta_1-\eta_0$ has negative mass on some Borel subset of $I_-$.  Since $\BAR{h_1|_{I_1}}<\BAR{h_0|_{I_0}}$ on $I_-$, the function $w\coloneq \BAR{h_0|_{I_0}}-\BAR{h_1|_{I_1}}$ defined on $I_-$ is positive.  Define
\[y_1\coloneq \begin{dcases}
\BAR{h_1|_{I_1}},&\text{on }I_1\varsetminus I_-,\\
\BAR{h_0|_{I_0}},&\text{on }I_-.
\end{dcases}\]
By the maximality of $I_-$, $y_1$ is nondecreasing on $I_1$.  Moreover, by construction,
\[y_1-\BAR{h_1|_{I_1}} = \begin{dcases}
0,&\text{on }I_1\varsetminus I_-,\\
w,&\text{on }I_-.\end{dcases}\]
By the earlier variational characterization of ironing, 
\[\int_{I_-} w\,\dd\eta_1\leq \int_{I_-}w\BAR{h_1|_{I_1}}\,\dd F.\]
Similarly, define 
\[y_0\coloneq \begin{dcases}
\BAR{h_0|_{I_0}},&\text{on }I_0\varsetminus I_-,\\
\BAR{h_1|_{I_1}},&\text{on }I_-.
\end{dcases}\]
By the maximality of $I_-$, $y_0$ is nondecreasing on $I_0$ as well.  Moreover, by construction,
\[y_0-\BAR{h_0|_{I_0}} = \begin{dcases}
0,&\text{on }I_0\varsetminus I_-,\\
-w,&\text{on }I_-.\end{dcases}\]
By the earlier variational characterization of ironing, 
\[\int_{I_-} \(-w\)\,\dd\eta_0\leq \int_{I_-}\(-w\)\BAR{h_0|_{I_0}}\,\dd F\implies \int_{I_-}w\,\dd\eta_0\geq \int_{I_-}w\BAR{h_0|_{I_0}}\,\dd F.\]
It follows that
\begin{align*}
\int_{I_-}w\,\dd\(\eta_1-\eta_0\)
&= \underbrace{\(\int_{I_-}w\,\dd\eta_1-\int_{I_-}w\BAR{h_1|_{I_1}}\,\dd F\)}_{\leq 0}-\underbrace{\(\int_{I_-}w\,\dd\eta_0 - \int_{I_-}w\BAR{h_0|_{I_0}}\,\dd F\)}_{\geq 0} \\ 
&\qquad+ \int_{I_-}w\underbrace{\(\BAR{h_1|_{I_1}}-\BAR{h_0|_{I_0}}\)}_{=-w}\,\dd F\leq -\int_{I_-}w^2\,\dd F< 0.
\end{align*}
Strict inequality holds because $I_-$ is nondegenerate and $w$ is positive on $I_-$.  Thus, there exists a Borel set $B_-\subseteq I_-$ such that $\(\eta_1-\eta_0\)(B_-)<0$, as claimed.

Next, we observe that a symmetric argument shows that the signed measure $\eta_1-\eta_0$ has positive mass on some Borel subset $B_+$ of $I_+$.

Finally, we derive a contradiction from both of the preceding observations.  By our assumption that $h_1$ crosses $h_0$ at most once from above, there exists $\hat\theta\in I_0$ for which $\eta_1\leq \eta_0$ on $[\hat\theta,\theta_0]$ and $\eta_1\geq\eta_0$ on $[\theta_L,\hat\theta]$.  But $I_-$ cannot lie on the left of $\hat\theta$ (i.e., $\theta_H^-\not\leq \hat\theta$) as there exists $B_-\subseteq I_-$ such that $\(\eta_1-\eta_0\)(B_-)<0$.  Similarly, $I_+$ cannot lie on the right of $\hat\theta$ (i.e., $\theta_L^+\not\geq\hat\theta$) as there exists $B_+\subseteq I_+$ such that $\(\eta_1-\eta_0\)(B_+)>0$.  Thus,
\[\hat\theta<\theta_H^-\leq\theta_L^+<\hat\theta,\]
a contradiction.  We conclude that our earlier assumption was wrong; hence, $\BAR{h_1|_{I_1}}$ crosses $\BAR{h_0|_{I_0}}$ at most once from above in $I_0$, as was desired to prove.
\end{proof}

\subsection{Proof of \texorpdfstring{\Cref{cor:positive}}{Corollary 2}}

We verify the conditions in \Cref{thm:optimal}.  First, consider any $\theta\in[\theta^L,\BAR\theta]$.  By \Cref{ass:positive} and the definition of $\theta^L$, 
\[\int_\theta^{\BAR\theta}\left[\omega(s)-\a\right]\,\dd F(s)\geq 0.\]
Since $\mu^*=\(\E[\omega]-\a\)_+\geq0$, it follows from \cref{eq:virtual} that $J_{\mu^*}(\theta)\geq\theta$ for all $\theta\in[\theta^L,\BAR\theta]$.  Therefore, by \Cref{lem:ironing}, its ironed version on the interval also satisfies
\[\BAR{J_{\mu^*}|_{[\theta^L,\BAR\theta]}}(\theta)\geq\theta,\qquad\text{for all }\theta\in[\theta^L,\BAR\theta].\]
Consequently, the first condition in \Cref{thm:optimal} holds.  Next, consider any $\theta<\theta^L$.  Since $\omega$ is increasing, if $\omega(\theta)>\a$, then $\omega(s)\geq \omega(\theta)>\a$ for all $s\geq\theta$, which would imply
\[\int_\theta^{\BAR\theta}\left[\omega(s)-\a\right]\,\dd F(s)\geq 0.\]
But this contradicts the definition of $\theta^L$, so the second condition in \Cref{thm:optimal} also holds.  We conclude from \Cref{thm:optimal} that the optimal mechanisms with and without topping up coincide.

\subsection{Proof of \texorpdfstring{\Cref{prop:positive_shift}}{Proposition 1}}

For welfare weight functions $\omega_H\geq \omega_L$ satisfying \Cref{ass:positive}, 
\[\(\E[\omega_H]-\a\)_+\geq \(\E[\omega_L]-\a\)_+.\]
It follows that $\mu^*$ is weakly higher under $\omega_H$, proving claim~\emph{\ref{it:positive_shift_multiplier}}.  Moreover, for all $\theta\in[\und\theta,\BAR\theta]$,
\[\int_\theta^{\BAR\theta}\left[\omega_H(s)-\a\right]\,\dd F(s) \geq \int_\theta^{\BAR\theta}\left[\omega_L(s)-\a\right]\,\dd F(s).\]
Consequently,
\[\left\{\hat\theta\in[\und\theta,\BAR\theta]:\int_{\hat\theta}^{\BAR\theta}\left[\omega_H(s)-\a\right]\,\dd F(s)\geq0\right\}\supseteq\left\{\hat\theta\in[\und\theta,\BAR\theta]:\int_{\hat\theta}^{\BAR\theta}\left[\omega_L(s)-\a\right]\,\dd F(s)\geq0\right\}.\]
\Cref{eq:cutoff} implies that $\theta^L$ is weakly lower under $\omega_H$, proving claim~\emph{\ref{it:positive_shift_cutoff}}.  Finally, let $\theta_H^L$ and $\theta_L^L$ denote the cutoffs under $\omega_H$ and $\omega_L$, respectively;  claim~\emph{\ref{it:positive_shift_cutoff}} establishes $\theta_H^L\leq\theta_L^L$.  For $\theta<\theta_H^L$, both environments assign the laissez-faire allocation $q^{\LF}$.  For $\theta\in[\theta_H^L,\theta_L^L)$, the allocation under $\omega_L$ is $q^{\LF}$, while the allocation under $\omega_H$ is weakly higher than $q^{\LF}$ (as verified in the proof of \Cref{cor:positive}).  Next, let $\mu_H^*$ and $\mu_L^*$ denote the multipliers under $\omega_H$ and $\omega_L$; claim~\emph{\ref{it:positive_shift_multiplier}} establishes that $\mu_H^*\geq \mu_L^*$.  For $\theta\geq\theta_L^L$, 
\[\theta + \frac{\mu_H^*\und\theta\cdot\d_{\und\theta}(\theta) + \int_{\theta}^{\BAR\theta}\left[\omega_H(s)-\a\right]\,\dd F(s)}{\a f(\theta)}\geq \theta + \frac{\mu_L^*\und\theta\cdot\d_{\und\theta}(\theta) + \int_{\theta}^{\BAR\theta}\left[\omega_L(s)-\a\right]\,\dd F(s)}{\a f(\theta)}.\]
Applying \Cref{rem:ironing} yields that $\BAR{J_{\mu_H^*}|_{[\theta_H^L,\BAR\theta]}}\geq \BAR{J_{\mu_L^*}|_{[\theta_L^L,\BAR\theta]}}$ on $[\theta_L^L,\BAR\theta]$.  Thus, the optimal allocation function $q^*$ is pointwise weakly higher under $\omega_H$, proving claim~\emph{\ref{it:positive_shift_allocation}} and completing the proof of \Cref{prop:positive_shift}.

\subsection{Proof of \texorpdfstring{\Cref{prop:positive_MPS}}{Proposition 2}}

Let $\omega_H$ and $\omega_L$ be welfare weight functions satisfying \Cref{ass:positive} with the same mean, and suppose that $\omega_H$ is a mean-preserving spread of $\omega_L$.  Since $\E[\omega_H]=\E[\omega_L]$,
\[\(\E[\omega_H]-\a\)_+=\(\E[\omega_L]-\a\)_+.\]
It follows that $\mu^*$ is unchanged, proving claim \emph{\ref{it:positive_MPS_multiplier}}.  Moreover, for all $\theta\in[\und\theta,\BAR\theta]$,
\[\int_\theta^{\BAR\theta}\omega_H(s)\,\dd F(s)\geq \int_\theta^{\BAR\theta}\omega_L(s)\,\dd F(s).\]
As in the proof of \Cref{prop:positive_shift}, this implies that $\theta^L$ is weakly lower under $\omega_H$, proving claim~\emph{\ref{it:positive_MPS_cutoff}}.  Finally, combining claims~\emph{\ref{it:positive_MPS_multiplier}} and \emph{\ref{it:positive_MPS_cutoff}} and invoking \Cref{rem:ironing}, we conclude as in the proof of \Cref{prop:positive_shift} that $\BAR{J_{\mu^*}|_{[\theta^L,\BAR\theta]}}$---and hence the optimal allocation function $q^*$---is pointwise weakly higher under $\omega_H$, proving claim~\emph{\ref{it:positive_MPS_allocation}} and completing the proof of \Cref{prop:positive_MPS}.

\subsection{Proof of \texorpdfstring{\Cref{cor:negative}}{Corollary 3}}

We consider three cases:
\begin{enumerate}[label={(\roman*)}]
\item If $\a\geq\max\omega$, then laissez-faire is optimal under both participation structures by \Cref{thm:scope}, so topping up does not affect the optimal mechanism.

\item If $\max\omega>\a\geq \E[\omega]$, then \Cref{thm:scope} implies that the social planner intervenes without topping up but not with topping up, so the two mechanisms differ.

\item If $\E[\omega]>\a$, then the topping-up cutoff type defined in \cref{eq:cutoff} is $\theta^L=\und\theta$ since 
\[\int_{\und\theta}^{\BAR\theta}\left[\omega(s)-\a\right]\,\dd F(s) = \E[\omega]-\a>0.\]
Observe from \cref{eq:virtual} that the ironed version $\BAR{J_{\mu^*}}(\theta)$ of the generalized virtual valuation is nonincreasing (from the signed measure monotonicity of $J_{\mu^*}$ in $\a$ and because ironing preserves monotonicity by \Cref{lem:ironing}) and continuous in $\a$ for fixed $\theta\in[\und\theta,\BAR\theta]$ (due to continuity of $J_{\mu^*}$ in $\a$ and because ironing preserves continuity by the min-max formula).  Define
\[\calC\coloneq \left\{\a\in(0,\E[\omega]):\BAR{J_{\mu^*}}(\theta)\geq\theta,\quad\text{for all }\theta\in[\und\theta,\BAR\theta]\right\}.\]
Let $\a_{\min}\coloneq \sup\calC$, with $\a_{\min}=0$ if $\calC$ is empty.  Monotonicity in $\a$ implies that $\calC$ is a lower set of $\R_{++}$.  By \Cref{thm:optimal}, the optimal mechanisms with and without topping up must differ for any $\a>\a_{\min}$, and they must coincide for any $\a\leq\a_{\min}$.

For $\a\leq\min\omega$, observe that
\[\int_\theta^{\BAR\theta}\left[\omega(s)-\a\right]\,\dd F(s)\geq0,\qquad\text{for all }\theta\in[\und\theta,\BAR\theta].\]
Therefore, $J_{\mu^*}(\theta)\geq \theta$---and hence $\BAR{J_{\mu^*}}(\theta)\geq\theta$, by \Cref{lem:ironing}---for all $\theta\in[\und\theta,\BAR\theta]$.  It follows that $\a_{\min}\geq\min\omega$.  Moreover, at $\a=\E[\omega]$, \Cref{thm:scope} implies that the optimal mechanisms differ because $\max\omega>\E[\omega]$ under \Cref{ass:negative}.  Thus, $\a_{\min}<\E[\omega]$.
\end{enumerate}

\subsection{Proof of \texorpdfstring{\Cref{prop:comparison}}{Proposition 3}}

First, observe that the comparison between $\mu^*_{\TU}$ and $\mu^*$ follows from \cref{eq:negative_multiplier_TU,eq:negative_multiplier}:
\[\mu_{\TU}^*= \(\E[\omega]-\a\)_+ \leq \mu^*.\]

Next, we show that the cutoff without topping up is weakly higher: $\theta_{\TU}^H(\mu^*_{\TU})\leq\theta^H(\mu^*)$.  To this end, consider two cases.
\begin{enumerate}[label=\emph{(\roman*)}]
\item Suppose $\mu^*_{\TU}=0$.  In this case, $\E[\omega]\leq \a$.  Since $\omega$ is decreasing by \Cref{ass:negative}, 
\[\int_\theta^{\BAR\theta}\left[\omega(s)-\a\right]\,\dd F(s)\leq0\implies J_0(\theta)\leq\theta,\qquad\text{for all }\theta\in[\und\theta,\BAR\theta].\]
Thus, \cref{eq:negative_cutoff_TU,eq:negative_cutoff} imply that
\[\theta_{\TU}^H(\mu_{\TU}^*) = \und\theta\leq \theta^H(\mu^*).\]

\item Suppose $\mu_{\TU}^*>0$.  Then $\mu^*_{\TU}=\E[\omega]-\a=\int_{\und\theta}^{\BAR\theta}\left[\omega(s)-\a\right]\,\dd F(s)$.

If $\mu^*=\mu^*_{\TU}$, then \cref{eq:negative_cutoff} implies that
\[\theta^H(\mu^*)=\theta^H(\E[\omega]-\a)=\BAR\theta\geq \theta^H_{\TU}(\mu^*_{\TU}).\]

If $\mu^*>\mu^*_{\TU}$, then \cref{eq:negative_cutoff} implies that $\theta^H(\mu^*)<\BAR\theta$.  Now, observe that
\begin{align*}
J_{\mu^*}(\theta) + \frac{\mu^* + \a -\E[\omega]}{\a f(\theta)} 
&= \theta + \frac{\mu^*\und\theta\cdot\d_{\und\theta}(\theta)+\mu^* - \int_{\und\theta}^\theta\left[\omega(s)-\a\right]\,\dd F(s)}{\a f(\theta)}\\
&\geq \theta + \frac{\mu_{\TU}^*\und\theta\cdot\d_{\und\theta}(\theta)+\mu_{\TU}^* - \int_{\und\theta}^\theta\left[\omega(s)-\a\right]\,\dd F(s)}{\a f(\theta)}\\
&= J_{\mu_{\TU}^*}(\theta) + \frac{\mu_{\TU}^* + \a - \E[\omega]}{\a f(\theta)} = J_{\mu_{\TU}^*}(\theta),\qquad\text{for all }\theta\in[\und\theta,\BAR\theta].
\end{align*}
Thus, \Cref{lem:ironing} implies that
\[\BAR{\(\left.J_{\mu^*}+\frac{\mu^*+\a-\E[\omega]}{\a f}\right|_{[\und\theta,\theta^H(\mu^*)]}\)}(\theta)\geq \BAR{J_{\mu^*_{\TU}}|_{[\und\theta,\theta^H(\mu^*)]}}(\theta),\qquad\text{for all }\theta\in[\und\theta,\theta^H(\mu^*)].\]
In particular, we obtain
\[\theta^H(\mu^*)\geq \BAR{\(\left.J_{\mu^*}+\frac{\mu^*+\a-\E[\omega]}{\a f}\right|_{[\und\theta,\theta^H(\mu^*)]}\)}(\theta^H(\mu^*))\geq \BAR{J_{\mu^*_{\TU}}|_{[\und\theta,\theta^H(\mu^*)]}}(\theta^H(\mu^*)).\]
Here, the first inequality follows from the feasibility verification in the proof of \Cref{thm:negative}, since the optimal mechanism without topping up is nondecreasing at the cutoff $\theta^H(\mu^*)$.  By minimality in \cref{eq:negative_cutoff_TU}, it follows that
\[\theta^H_{\TU}(\mu_{\TU}^*)\leq \theta^H(\mu^*).\]
\end{enumerate}

Finally, we show that $q^*$ crosses $q^*_{\TU}$ at most once from above.  For a given $\theta_0\in[\und\theta,\BAR\theta]$, suppose that $q^*(\theta_0)\leq q_{\TU}^*(\theta_0)$.  We consider two cases.
\begin{enumerate}[label=\emph{(\roman*)}]
\item Suppose $\theta_0\leq \theta_{\TU}^H(\mu_{\TU}^*)$.  Our computation above showed that
\[h(\theta)\coloneq J_{\mu^*}(\theta) + \frac{\mu^*+\a-\E[\omega]}{\a f(\theta)} \geq J_{\mu^*_{\TU}}(\theta)\eqcolon h_{\TU}(\theta),\qquad\text{for all }\theta\in[\und\theta,\BAR\theta].\]
In particular, $h$ crosses $h_{\TU}$ at most once from above in $[\und\theta,\theta^H_{\TU}(\mu_{\TU}^*)]$.  By \Cref{lem:single-crossing}, $\BAR{h|_{[\und\theta,\theta^H(\mu^*)]}}$ crosses $\BAR{h_{\TU}|_{[\und\theta,\theta^H_{\TU}(\mu^*_{\TU})]}}$ at most once from above in $[\und\theta,\theta^H_{\TU}(\mu_{\TU}^*)]$; this implies that $q^*$ crosses $q_{\TU}^*$ at most once from above in $[\und\theta,\theta_{\TU}^H(\mu_{\TU}^*)]$.  Thus, $q^*(\theta)\leq q_{\TU}^*(\theta)$ for all $\theta\in[\theta_0,\theta_{\TU}^H(\mu_{\TU}^*)]$.

It remains to extend the comparison beyond $\theta_{\TU}^H(\mu_{\TU}^*)$.  If $\theta_{\TU}^H(\mu_{\TU}^*)=\BAR\theta$, there is nothing further to prove; hence, we assume below that $\theta_{\TU}^H(\mu_{\TU}^*)<\BAR\theta$.  By construction, $q_{\TU}^*(\theta)=q^{\LF}(\theta)$ for all $\theta\in(\theta_{\TU}^H(\mu_{\TU}^*),\BAR\theta]$.  In particular, because the ironed generalized virtual valuation is weakly below identity at the cutoff, \Cref{lem:nu_SCP} implies that
\[q^*(\theta^H_{\TU}(\mu_{\TU}^*))\leq q^*_{\TU}(\theta^H_{\TU}(\mu_{\TU}^*))\leq q^{\LF}(\theta^H_{\TU}(\mu_{\TU}^*)).\] 
By \Cref{lem:nu_SCP}, it follows that 
\[q^*(\theta)\leq q^{\LF}(\theta)=q_{\TU}^*(\theta),\qquad\text{for all }\theta\in[\theta_{\TU}^H(\mu_{\TU}^*),\BAR\theta].\] 
This shows that, in fact, $q^*(\theta)\leq q^*_{\TU}(\theta)$ for all $\theta\geq\theta_0$.

\item Suppose $\theta_0>\theta_{\TU}^H(\mu_{\TU}^*)$.  Since $q^*$ crosses $q^{\LF}$ at most once from above in $[\und\theta,\theta^H(\mu^*)]$, a similar argument to that of the previous case shows that
\[q^*(\theta)\leq q^{\LF}(\theta) = q_{\TU}^*(\theta),\qquad\text{for all }\theta\in[\theta_0,\BAR\theta].\]
Hence, $q^*(\theta)\leq q_{\TU}^*(\theta)$ for all $\theta\geq\theta_0$.
\end{enumerate}
Thus, $q^*$ remains pointwise weakly lower than $q_{\TU}^*$ once $q^*$ crosses $q_{\TU}^*$ from above.  Consequently, $q^*$ crosses $q_{\TU}^*$ at most once from above.

\subsection{Proof of \texorpdfstring{\Cref{prop:negative_shift}}{Proposition 4}}

\paragraph{Topping up.} For welfare weight functions $\omega_H\geq \omega_L$ satisfying \Cref{ass:negative}, 
\[\(\E[\omega_H]-\a\)_+\geq \(\E[\omega_L]-\a\)_+.\]
It follows that $\mu_{\TU}^*$ is weakly higher under $\omega_H$, proving claim~\emph{\ref{it:negative_shift_multiplier}}.  Moreover, for all $\theta\in[\und\theta,\BAR\theta]$,
\[\int_\theta^{\BAR\theta}\left[\omega_H(s)-\a\right]\,\dd F(s) \geq \int_\theta^{\BAR\theta}\left[\omega_L(s)-\a\right]\,\dd F(s).\]
Thus, $J_{\mu_{\TU}^*}|_{[\und\theta,\theta]}(\theta)$---and hence $\BAR{J_{\mu_{\TU}^*}|_{[\und\theta,\theta]}}(\theta)$, by \Cref{lem:ironing}---is weakly higher under $\omega_H$ for all $\theta\in[\und\theta,\BAR\theta]$.  Consequently, the following set is weakly smaller (in the sense of set inclusion) under $\omega_H$: 
\[\left\{\theta\in[\und\theta,\BAR\theta]:\BAR{J_{\mu_{\TU}^*}|_{[\und\theta,\theta]}}(\theta)\leq\theta\right\}.\]
\Cref{eq:negative_cutoff_TU} therefore implies that $\theta^H_{\TU}(\mu^*_{\TU})$ is weakly higher under $\omega_H$, proving claim~\emph{\ref{it:negative_shift_cutoff}}.  Finally, the topping-up allocation is the unique solution to the constrained problem
\[\max_{\nu\in\mathcal I}\left\{\int_{\und\theta}^{\BAR\theta}\left[J(\theta)\nu(\theta)-c\Psi(\nu(\theta))\right]\,\dd F(\theta):\nu(\theta)\geq\nu^{\LF}(\theta),\quad\text{for all }\theta\in[\und\theta,\BAR\theta]\right\}.\]
Observe that the feasible set is a sublattice under pointwise maximum and minimum, the objective has increasing differences in $(\nu,J)$ in the signed measure order, and strict convexity of $\Psi$ yields a unique maximizer.  Since $J_{\mu_{\TU}^*}$ is weakly higher under $\omega_H$, it follows that $\nu_{\TU}^*$ is weakly higher under $\omega_H$ \citep[Theorem 2.8.1]{topkis98}.  Thus, $q^*_{\TU}$ is weakly higher under $\omega_H$, proving claim~\emph{\ref{it:negative_shift_allocation}}.

\paragraph{No topping up.} For $i\in\{H,L\}$, let $\theta^H_i(\mu)$ denote the cutoff under $\omega_i$:
\[\theta_i^H(\mu)\coloneq\max\left\{\theta\in[\und\theta,\BAR\theta]:\int_{\und\theta}^{\theta}\left[\omega_i(s)-\a\right]\,\dd F(s)\geq\mu\right\}.\] 
Also define
\[\begin{dcases}
\begin{aligned}
\mu_{\max,i}&\coloneq\max_{\theta\in[\und\theta,\BAR\theta]}\int_{\und\theta}^\theta\left[\omega_i(s)-\a\right]\,\dd F(s),\\
q_{\mu,i}(\theta)&\coloneq D\!\(c,\BAR{\left.\(z\mapsto z+\frac{\mu \und\theta\cdot\d_{\und\theta}(z) + \mu -\int_{\und\theta}^{z}\left[\omega_i(s)-\a\right]\,\dd F(s)}{\a f(z)}\)\right|_{[\und\theta,\theta_i^H(\mu)]}}(\theta)\),\\
U_i^H(\mu) &\coloneq \und\theta v(q_{\mu,i}(\und\theta)) + \int_{\und\theta}^{\theta_i^H(\mu)}v(q_{\mu,i}(s))\,\dd s,\\
\Phi_i(\mu)&\coloneq U_i^H(\mu)-U^{\LF}(\theta^H_i(\mu)).
\end{aligned}
\end{dcases}\]

\begin{lemma}\label{lem:Phi}
If $\Phi_H(\mu)\geq0$ for some $\mu\in[0,\mu_{\max,L}]$, then $\Phi_L(\mu)\geq0$.
\end{lemma}
\begin{proof}
We begin by showing that $\theta_H^H(\mu)\geq \theta^H_L(\mu)$.  Since $\omega_H\geq\omega_L$,
\[\int_{\und\theta}^{\theta}\left[\omega_H(s)-\a\right]\,\dd F(s)\geq \int_{\und\theta}^{\theta}\left[\omega_L(s)-\a\right]\,\dd F(s),\qquad\text{for all }\theta\in[\und\theta,\BAR\theta].\]
Thus,
\[\left\{\theta\in[\und\theta,\BAR\theta]:\int_{\und\theta}^\theta\left[\omega_H(s)-\a\right]\,\dd F(s)\geq\mu\right\}\supseteq\left\{\theta\in[\und\theta,\BAR\theta]:\int_{\und\theta}^\theta\left[\omega_L(s)-\a\right]\,\dd F(s)\geq\mu\right\}.\]
As the maximum of a larger set, $\theta_H^H(\mu)\geq\theta_L^H(\mu)$.  Now, for $\theta\in[\und\theta,\theta^H_L(\mu)]$,
\[\frac{\mu\und\theta\cdot\d_{\und\theta}(\theta) + \mu -\int_{\und\theta}^\theta\left[\omega_H(s)-\a\right]\,\dd F(s)}{\a f(\theta)}\leq\frac{\mu\und\theta\cdot\d_{\und\theta}(\theta) + \mu -\int_{\und\theta}^\theta\left[\omega_L(s)-\a\right]\,\dd F(s)}{\a f(\theta)}.\]
By \Cref{lem:ironing}, this implies that $q_{\mu,H}(\theta)\leq q_{\mu,L}(\theta)$ for all $\theta\in[\und\theta,\theta_L^H(\mu)]$.  Next, consider two cases:
\begin{itemize}
\item {\bf\ul{Case \#1}.}  Suppose that
\[\int_{\theta_L^H(\mu)}^{\theta_H^H(\mu)}\left[v(q_{\mu,H}(s))-v(q^{\LF}(s))\right]\,\dd s\leq0.\]
Then
\begin{align*}
0\leq\Phi_H(\mu)
&= \und\theta v(q_{\mu,H}(\und\theta)) - U^{\LF}(\und\theta) + \int_{\und\theta}^{\theta_H^H(\mu)}\left[v(q_{\mu,H}(s))-v(q^{\LF}(s))\right]\,\dd s\\
&\leq \und\theta v(q_{\mu,H}(\und\theta)) - U^{\LF}(\und\theta) + \int_{\und\theta}^{\theta_L^H(\mu)}\left[v(q_{\mu,H}(s))-v(q^{\LF}(s))\right]\,\dd s\\
&\leq \und\theta v(q_{\mu,L}(\und\theta)) - U^{\LF}(\und\theta) + \int_{\und\theta}^{\theta_L^H(\mu)}\left[v(q_{\mu,L}(s))-v(q^{\LF}(s))\right]\,\dd s=\Phi_L(\mu).
\end{align*}
Consequently, $\Phi_H(\mu)\geq0$ implies that $\Phi_L(\mu)\geq0$ also, as claimed.

\item {\bf\ul{Case \#2}.} Suppose that
\[\int_{\theta_L^H(\mu)}^{\theta_H^H(\mu)}\left[v(q_{\mu,H}(s))-v(q^{\LF}(s))\right]\,\dd s>0.\]
This implies that $q_{\mu,H}>q^{\LF}$ on a subset of positive measure in the interval $\(\theta_L^H(\mu),\theta_H^H(\mu)\)$.  However, by \Cref{lem:single-crossing}, $q_{\mu,H}$ crosses $q^{\LF}$ at most once in $[\und\theta,\theta^H_H(\mu)]$ from above.  It follows that 
\[q_{\mu,L}(\theta_L^H(\mu))\geq q_{\mu,H}(\theta_L^H(\mu))\geq q^{\LF}(\theta_L^H(\mu)).\]
By \Cref{lem:single-crossing} again, $q_{\mu,L}$ crosses $q^{\LF}$ at most once in $[\und\theta,\theta^H_L(\mu)]$ from above, implying that $q_{\mu,L}(\theta)\geq q^{\LF}(\theta)$ for all $\theta\in[\und\theta,\theta_L^H(\mu)]$.

Consequently, 
\begin{align*}
\Phi_L(\mu) 
&= \und\theta v(q_{\mu,L}(\und\theta)) - U^{\LF}(\und\theta) + \underbrace{\int_{\und\theta}^{\theta_L^H(\mu)}\left[v(q_{\mu,L}(s))-v(q^{\LF}(s))\right]\,\dd s}_{\geq 0}\\
&\geq \und\theta v(q^{\LF}(\und\theta)) - U^{\LF}(\und\theta) = cq^{\LF}(\und\theta)\geq0.
\end{align*}
\end{itemize}
Thus, $\Phi_L(\mu)\geq0$ in either case, concluding the proof of the lemma.
\end{proof}

We now show that $\mu^*$ is weakly higher under $\omega_H$.  Suppose on the contrary that $\mu_H^*<\mu_L^*$.  Then
\[\mu_H^*<\mu_L^*\leq \mu_{\max,L}\implies \mu_H^*\in[0,\mu_{\max,L}].\]
Since $\Phi_H(\mu_H^*)\geq 0$, \Cref{lem:Phi} implies that $\Phi_L(\mu_H^*)\geq0$; but this contradicts the minimality of $\mu_L^*$.  Thus, the multiplier $\mu^*$ is weakly higher under $\omega_H$, proving claim~\emph{\ref{it:negative_shift_multiplier}}.  

Next, we show that the cutoff $\theta^H(\mu^*)$ is weakly higher under $\omega_H$.  For $i\in\{H,L\}$, define
\[\begin{dcases}
\begin{aligned}
\BAR\kappa_i&\coloneq \max\left\{\theta\in[\und\theta,\BAR\theta]:\int_{\und\theta}^{\theta}\left[\omega_i(s)-\a\right]\,\dd F(s)\geq \(\E[\omega_i]-\a\)_+\right\},\\
\und\kappa_i&\in\argmax_{\theta\in[\und\theta,\BAR\theta]}\int_{\und\theta}^\theta\left[\omega_i(s)-\a\right]\,\dd F(s).
\end{aligned}
\end{dcases}\]
Since each $\omega_i$ is decreasing by \Cref{ass:negative}, $\und\kappa_i$ is uniquely defined.  In addition, define $\mu_i:[\und\kappa_i,\BAR\kappa_i]\to[\(\E[\omega_i]-\a\)_+,\mu_{\max,i}]$ by
\[\mu_i(\kappa)\coloneq \int_{\und\theta}^{\kappa}\left[\omega_i(s)-\a\right]\,\dd F(s).\]
As defined, $\und\kappa_L\leq\und\kappa_H$ \citep[Theorem 2.8.1]{topkis98}.  Moreover, we claim that $\BAR\kappa_H\geq \BAR\kappa_L$:
\begin{enumerate}[label=\emph{(\roman*)}]
\item Suppose $\E[\omega_L]>\a$.  Then $\E[\omega_H]\geq \E[\omega_L]>\a$, in which case $\BAR\kappa_H=\BAR\kappa_L=\BAR\theta$.
\item Suppose $\E[\omega_H]>\a\geq \E[\omega_L]$.  Then $\BAR\kappa_H=\BAR\theta\geq \BAR\kappa_L$, so the claim holds.
\item Suppose $\a\geq \E[\omega_H]$.  Then $\(\E[\omega_H]-\a\)_+=\(\E[\omega_L]-\a\)_+=0$.  Since $\omega_H\geq\omega_L$ pointwise, it follows from the definition of $\BAR\kappa_i$ that $\BAR\kappa_H\geq\BAR\kappa_L$.
\end{enumerate}
Since $\omega_H$ and $\omega_L$ are decreasing, $\mu_H$ and $\mu_L$ are decreasing as well.  By construction,
\[\kappa = \max\left\{\theta\in[\und\theta,\BAR\theta]:\int_{\und\theta}^\theta\left[\omega_i(s)-\a\right]\,\dd F(s)\geq \mu_i(\kappa)\right\}.\]
Next, for $i\in\{H,L\}$, define
\[q_i(\theta;\kappa)\coloneq D\!\(c,\BAR{\left.\(z\mapsto z+\frac{\mu_i(\kappa) \und\theta\cdot\d_{\und\theta}(z) + \int_z^{\kappa}\left[\omega_i(s)-\a\right]\,\dd F(s)}{\a f(z)}\)\right|_{[\und\theta,\kappa]}}(\theta)\).\]
For every $\kappa\in[\und\kappa_L,\BAR\kappa_L]\cap[\und\kappa_H,\BAR\kappa_H]$, $\omega_H\geq\omega_L$ implies
\[q_H(\cdot;\kappa)\geq q_L(\cdot;\kappa),\qquad\text{on }[\und\theta,\kappa].\]
By \cref{eq:negative_multiplier}, the optimal cutoff under $\omega_i$ is
\[\kappa_i^*\coloneq \max\left\{\kappa\in[\und\kappa_i,\BAR\kappa_i]:\und\theta v(q_i(\und\theta;\kappa)) + \int_{\und\theta}^{\kappa} v(q_i(s;\kappa))\,\dd s\geq U^{\LF}(\kappa)\right\}.\]
If $\kappa_H^*<\kappa_L^*$, then $\kappa_L^*>\kappa_H^*\geq \und\kappa_H$ and $\kappa_L^*\leq\BAR\kappa_L\leq \BAR\kappa_H$; hence, $\kappa_L^*\in[\und\kappa_H,\BAR\kappa_H]$.  By the maximality of $\kappa^*_H$, we obtain a contradiction:
\begin{align*}
U^{\LF}(\kappa_L^*)
&> \und\theta v(q_H(\und\theta;\kappa_L^*)) + \int_{\und\theta}^{\kappa_L^*}v(q_H(s;\kappa_L^*))\,\dd s\\
&\geq \und\theta v(q_L(\und\theta;\kappa_L^*)) + \int_{\und\theta}^{\kappa_L^*}v(q_L(s;\kappa_L^*))\,\dd s\geq U^{\LF}(\kappa_L^*).
\end{align*}
Therefore, $\kappa_H^*\geq \kappa_L^*$, proving claim~\emph{\ref{it:negative_shift_cutoff}}.

Finally, for $i\in\{H,L\}$, define the generalized function $h_i$ on the interval $[\und\theta,\kappa_i^*]$ by
\[h_i(\theta)\coloneq \theta+\frac{\mu_i^*\und\theta\cdot\d_{\und\theta}(\theta)+\mu_i^*-\int_{\und\theta}^\theta\left[\omega_i(s)-\a\right]\,\dd F(s)}{\a f(\theta)}.\]
For all $\theta\in[\und\theta,\kappa_L^*]$,
\[h_H(\theta) - h_L(\theta) = \frac{\(\mu_H^*-\mu_L^*\)\und\theta\cdot\d_{\und\theta}(\theta)+\(\mu_H^*-\mu_L^*\)-\int_{\und\theta}^\theta\left[\omega_H(s)-\omega_L(s)\right]\,\dd F(s)}{\a f(\theta)}.\]
Thus, $h_H$ crosses $h_L$ at most once from above in $[\und\theta,\kappa_L^*]$.  By \Cref{lem:single-crossing}, $\BAR{h_H|_{[\und\theta,\kappa_H^*]}}$ crosses $\BAR{h_L|_{[\und\theta,\kappa_L^*]}}$ at most once from above in $[\und\theta,\kappa^*_L]$; this implies that $q_{\mu^*_H,H}$ crosses $q_{\mu^*_L,L}$ at most once from above in $[\und\theta,\kappa^*_L]$.  Moreover, by \Cref{lem:nu_SCP}, $q_{\mu^*_H,H}$ crosses $q^{\LF}$ at most once from above in $[\und\theta,\kappa_H^*]$.  The feasibility verification in the proof of \Cref{thm:negative} also implies that $q_{\mu_L^*,L}(\kappa_L^*)\leq q^{\LF}(\kappa_L^*)$.  It remains to show that $q_H^*$ crosses $q_L^*$ at most once from above in the full interval $[\und\theta,\BAR\theta]$:
\begin{enumerate}[label=\emph{(\roman*)}]
\item If $q_H^*(\hat\theta)\leq q_L^*(\hat\theta)$ for some $\hat\theta\leq \kappa_L^*$, then
\[q_{\mu_H^*,H}(\theta)\leq q_{\mu_L^*,L}(\theta),\qquad\text{for all }\theta\in[\hat\theta,\kappa_L^*].\]
In particular,
\[q_{\mu_H^*,H}(\kappa_L^*)\leq q_{\mu_L^*,L}(\kappa_L^*)\leq q^{\LF}(\kappa_L^*).\]
Since $q_{\mu_H^*,H}$ crosses $q^{\LF}$  at most once from above in $[\und\theta,\kappa_H^*]$ by \Cref{lem:nu_SCP}, it follows that 
\[q_{\mu_H^*,H}(\theta)\leq q^{\LF}(\theta),\qquad\text{for all }\theta\in[\kappa_L^*,\kappa_H^*].\]
Thus, $q_H^*(\theta)\leq q_L^*(\theta)$ for all $\theta\geq\hat\theta$.

\item If $q_H^*(\hat\theta)\leq q_L^*(\hat\theta)$ for some $\hat\theta> \kappa_L^*$, then $q_L^*(\hat\theta)=q^{\LF}(\hat\theta)$.  If $\hat\theta\leq\kappa_H^*$, then \Cref{lem:nu_SCP} implies that $q_H^*(\theta)\leq q^{\LF}(\theta)=q_L^*(\theta)$ for all $\theta\geq\hat\theta$.  If $\hat\theta>\kappa_H^*$, then $q_H^*(\theta)=q_L^*(\theta)=q^{\LF}(\theta)$ for all $\theta\geq\hat\theta$.
\end{enumerate}  
We conclude that the optimal allocation function under $\omega_H$ crosses the optimal allocation function under $\omega_L$ at most once from above in $[\und\theta,\BAR\theta]$, proving claim~\emph{\ref{it:negative_shift_allocation}}.

\subsection{Proof of \texorpdfstring{\Cref{prop:negative_MPS}}{Proposition 5}}

\paragraph{Topping up.} Since $\E[\omega_H]=\E[\omega_L]$, it follows that
\[\(\E[\omega_H]-\a\)_+=\(\E[\omega_L]-\a\)_+.\]
Hence, $\mu_{\TU}^*$ is unchanged, proving claim~\emph{\ref{it:negative_MPS_multiplier}}.  Moreover, for all $\theta\in[\und\theta,\BAR\theta]$:
\[\theta + \frac{\mu_{\TU}^*\und\theta\cdot\d_{\und\theta}(\theta) + \int_\theta^{\BAR\theta}\left[\omega_H(s)-\a\right]\,\dd F(s)}{\a f(\theta)}\leq \theta + \frac{\mu_{\TU}^*\und\theta\cdot\d_{\und\theta}(\theta) + \int_\theta^{\BAR\theta}\left[\omega_L(s)-\a\right]\,\dd F(s)}{\a f(\theta)}.\]
Thus, $J_{\mu_{\TU}^*}|_{[\und\theta,\theta]}(\theta)$---and hence $\BAR{J_{\mu_{\TU}^*}|_{[\und\theta,\theta]}}(\theta)$---is weakly lower for $\omega_H$ for every fixed $\theta\in[\und\theta,\BAR\theta]$.  As such, the following set is weakly larger (in the sense of set inclusion) under $\omega_H$:
\[\left\{\theta\in[\und\theta,\BAR\theta]:\BAR{J_{\mu_{\TU}^*}|_{[\und\theta,\theta]}}(\theta)\leq\theta\right\}.\]
By \cref{eq:negative_cutoff_TU}, $\theta_{\TU}^H(\mu_{\TU}^*)$ is weakly lower under $\omega_H$, proving claim~\emph{\ref{it:negative_MPS_cutoff}}.  Finally, as in our proof of \Cref{prop:negative_shift}, the topping-up allocation is the unique solution to the constrained problem
\[\max_{\nu\in\mathcal I}\left\{\int_{\und\theta}^{\BAR\theta}\left[J(\theta)\nu(\theta)-c\Psi(\nu(\theta))\right]\,\dd F(\theta):\nu(\theta)\geq\nu^{\LF}(\theta),\quad\text{for all }\theta\in[\und\theta,\BAR\theta]\right\}.\]
The feasible set is common across $\omega_H$ and $\omega_L$.  Since $J_{\mu_{\TU}^*}$ is weakly lower under $\omega_H$, it follows that $\nu_{\TU}^*$ is weakly lower under $\omega_H$ \citep[Theorem 2.8.1]{topkis98}.  Thus, $q^*_{\TU}$ is weakly lower under $\omega_H$, proving claim~\emph{\ref{it:negative_MPS_allocation}}.

\paragraph{No topping up.}  Our proof of claim~\emph{\ref{it:negative_shift_multiplier}} in \Cref{prop:negative_shift}---and, in particular, \Cref{lem:Phi}---holds with $\omega_H$ being a mean-preserving spread of $\omega_L$.  Thus, the same argument applies here, proving claim~\emph{\ref{it:negative_MPS_multiplier}}.

\end{document}